\documentclass[a4paper,12pt,twoside]{article}

\usepackage{amsfonts}
\usepackage{amsmath}
\usepackage{amsopn}
\usepackage{amsthm}
\usepackage{amssymb}
\usepackage{graphicx}
\usepackage{rotating}
\usepackage[]{natbib}
\usepackage[english]{babel}
\usepackage[utf8]{inputenc}
\usepackage[colorlinks=true,citecolor=blue,linkcolor=blue]{hyperref}
\usepackage{xspace}
\usepackage{setspace}
\usepackage{epstopdf}
\usepackage{dsfont}
\usepackage{bbm}
\usepackage{booktabs}
\usepackage[T1]{fontenc}
\usepackage{anysize}
\usepackage{multirow}
\usepackage[labelfont=bf,labelsep=period]{caption}
\usepackage{floatrow}
\newfloatcommand{capbtabbox}{table}[][\FBwidth]
\usepackage{soul}
\marginsize{1.87cm}{1.87cm}{1.87cm}{1.87cm} 
\usepackage{fancyhdr}
\RequirePackage{txfonts}
\usepackage{times}

\newtheorem{Proposition}{Proposition}

\newtheorem{Corollary}{Corollary}
\newtheorem{Theorem}{Theorem}
\newtheorem{Lemma}{Lemma}

\renewcommand{\d}{\mathrm{d}}
\DeclareMathOperator{\R}{\mathbb{R}}
\DeclareMathOperator{\N}{\mathbb{N}}
\newcommand{\E}{\mathbb{E}}

\newcommand{\e}{\mathrm{e}}
\renewcommand{\P}{\mathbb{P}}

\newcommand{\I}{\mathbb{I}}
\newcommand{\bX}{\mathbf{X}}
\newcommand{\by}{\mathbf{y}}
\newcommand{\bY}{\mathbf{Y}}
\newcommand{\bx}{\mathbf{x}}
\newcommand{\bz}{\mathbf{z}}
\newcommand{\bZ}{\mathbf{Z}}
\newcommand{\bU}{\mathbf{U}}

\newcommand{\btheta}{\boldsymbol\theta}
\newcommand{\bvarepsilon}{\boldsymbol\varepsilon}
\newcommand{\1}{\mathbbm{1}}
\newcommand{\var}{\hbox{var}}

\newcommand{\Z}{\mathcal{Z}}

\allowdisplaybreaks

\begin{document}

\def\figureautorefname{Figure}
\def\tableautorefname{Table}
\def\algorithmautorefname{Algorithm}
\def\sectionautorefname{Section}
\def\subsectionautorefname{Section}
\def\Propositionautorefname{Proposition}
\def\Theoremautorefname{Theorem}
\def\Lemmaautorefname{Lemma}
\def\Assumptionautorefname{Assumption}
\def\Corollaryautorefname{Corollary}
\renewcommand*\footnoterule{}

\title{Improving multiple-try Metropolis with local balancing}

\author{Philippe Gagnon$^{1}$, Florian Maire$^{1}$ and Giacomo Zanella$^2$}

\maketitle

\thispagestyle{empty}

\noindent $^{1}$Department of Mathematics and Statistics, Universit\'{e} de Montr\'{e}al, Canada.

\noindent $^{2}$Department of Decision Sciences, BIDSA \& IGIER, Bocconi University, Italy.

\begin{abstract}
Multiple-try Metropolis (MTM) is a popular Markov chain Monte Carlo method with the appealing feature of being amenable to parallel computing. At each iteration, it samples several candidates for the next state of the Markov chain and randomly selects one of them based on a weight function. The canonical weight function is proportional to the target density. We show both theoretically and empirically that this weight function induces pathological behaviours in high dimensions, especially during the convergence phase. We propose to instead use weight functions akin to the \textit{locally-balanced} proposal distributions of \cite{zanella2019informed}, thus yielding MTM algorithms that do not exhibit those pathological behaviours. To theoretically analyse these algorithms, we study the high-dimensional performance of \textit{ideal} schemes that can be thought of as MTM algorithms which sample an infinite number of candidates at each iteration, as well as the discrepancy between such schemes and the MTM algorithms which sample a finite number of candidates. Our analysis unveils a strong distinction between the convergence and stationary phases: in the former, local balancing is crucial and effective to achieve fast convergence, while in the latter, the canonical and novel weight functions yield similar performance. Numerical experiments include an application in precision medicine involving a computationally-expensive forward model, which makes the use of parallel computing within MTM iterations beneficial.
\end{abstract}


\section{Introduction}\label{sec:intro}

\subsection{Multiple-try Metropolis}\label{sec:MTM}

In this paper, we study a specific Markov chain Monte Carlo (MCMC) method introduced by \cite{Liu2000} called \textit{Multiple-try Metropolis} (MTM). It can be seen as a generalization of the Metropolis--Hastings (MH, \cite{Metropolis1953} and \cite{Hastings1970}) algorithm: at each iteration, several candidates for the next state of the Markov chain \emph{(instead of one)} are sampled, hence the name of the algorithm. MTM can be used to sample from an \textit{intractable} distribution $\pi$ for Monte Carlo integration purposes, where \textit{intractable} here refers to the impossibility to compute integrals exactly with respect to that distribution. In a sampling context, such a distribution is called the \textit{target} distribution. This distribution often is a posterior distribution of model parameters resulting from a Bayesian model. In the following, we assume for simplicity that $\pi$ admits a strictly positive probability density function (PDF) with respect to Lebesgue measure, implying that the model parameters in a Bayesian context are continuous random variables; to simplify the notation, we will also use $\pi$ to denote the PDF. When the latter is a posterior density, it is proportional to the product of the likelihood function and a prior density. Its normalizing constant is not available, but it is assumed here that the target density can be evaluated pointwise (up to that normalizing constant).

In its simplest and most popular form, an iteration of MTM is as follows\footnote{We follow the formulation of \cite{bedard2012scaling}.}:
\begin{enumerate}

 \item $N$ values $\by_1, \ldots, \by_N \in \R^d$ are sampled independently from a proposal distribution with density $q_\sigma(\bx, \cdot \,)$, where $\bx \in \R^d$ is the current state of the chain and $\sigma > 0$ is a fixed scale parameter;

 \item one of the $\by_i$'s is randomly selected to be the proposal for the next state of the chain, say $\by_j$, with probability proportional to a weight $w(\bx, \by_j)$, where $w$ is a strictly positive weight function;

 \item $N - 1$ values $\bz_1, \ldots, \bz_{N - 1}$ are sampled independently from $q_\sigma(\by_j, \cdot \,)$;
 \item the proposal is accepted, meaning that the next state is set to be $\by_j$, with probability
 \begin{align}\label{eqn:acc_ratio}
  \alpha(\bx, \by_j) := 1 \wedge \frac{\pi(\by_j) \, q_\sigma(\by_j, \bx) \, w(\by_j, \bx) \bigg/ \left(\sum_{i = 1}^{N - 1} w(\by_j, \bz_i) + w(\by_j, \bx)\right)}{\pi(\bx) \, q_\sigma(\bx, \by_j) \, w(\bx, \by_j) \bigg/ \left(\sum_{i = 1}^N w(\bx, \by_i)\right)},
 \end{align}
 where $a \wedge b := \min(a, b)$ and the dependence of $\alpha(\bx, \by_j)$ on $\by_i, i \neq j$, and $\bz_1, \ldots, \bz_{N - 1}$ is made implicit to simplify the notation; when the proposal is rejected the chain remains at $\bx$.
 \end{enumerate}
Step 3 and the specific form of the acceptance probability \eqref{eqn:acc_ratio} are crucial ingredients to make the resulting Markov chains reversible with respect to $\pi$ and thus to ensure that $\pi$ is an invariant distribution and the algorithm is valid (see \cite{Liu2000} for full details). Typically, in Step 1, the sampling is performed through a random-walk scheme, and in particular, using a normal with a mean $\bx$ and a diagonal covariance matrix with diagonal elements given by $\sigma^2$, denoted by $q_\sigma(\bx, \cdot \,) = \mathcal{N}(\mathbf{x}, \sigma^2 \I_d)$, with $\I_d$ being the identity matrix of size $d$ (we also use $q_\sigma(\bx, \cdot \,)$ to denote the distribution to simplify). Note that the terms $q_\sigma(\by_j, \bx)$ and $q_\sigma(\bx, \by_j)$ in \eqref{eqn:acc_ratio} cancel each other when the density is symmetric as with the normal. In Steps 2 and 4, the weight function $w(\bx, \by)$ is typically a function of the ratio $\pi(\by) / \pi(\bx)$, the most popular choice by far being $w(\bx, \by) = \pi(\by) / \pi(\bx) \propto \pi(\by)$. Randomly choosing among the candidates $\by_1, \ldots, \by_N$ based on a function of their target-density evaluations $\pi(\by_1), \ldots, \pi(\by_N)$ makes MTM an \textit{informed} scheme, in the sense of \cite{zanella2019informed}, meaning that it leverages target-distribution information in the proposal mechanism. Such schemes can improve in terms of asymptotic variance of ergodic averages and mixing properties over their uninformed counterparts.

\subsection{Potential of MTM and parallel computing in MCMC}\label{sec:MTM}

MTM is appealing in situations where it is computationally expensive to evaluate the target density (typically because of the likelihood function), and it is not possible to obtain an explicit form of its gradient. In this situation, gradient-based MCMC methods may be either inapplicable or computationally intensive, e.g., if repeated numerical derivative approximations are required. Practitioners facing such a situation may naturally consider using a MH sampler with a random-walk proposal mechanism $q_\sigma$, or its MTM generalization.\footnote{One could also use MTM with an informed gradient-based proposal distribution, but this is less common in practice.} The weights in Steps 2 and 4 in MTM can be computed in parallel, which results in a significant computation-time reduction compared to serial computation when the iteration cost is largely dominated by that of evaluating the target density. In this situation, the iteration cost of MTM is roughly twice that of a MH sampler using the proposal distribution $q_\sigma$.\footnote{When the iteration cost is largely dominated by that of evaluating the target density, we know that the computational cost of the other operations (like sampling from $q_\sigma(\bx, \cdot \,)$) is negligible compared with that of evaluating the target density. This implies that the iteration cost of the MH sampler roughly corresponds to that of evaluating the target density once for evaluating the acceptance probability ($\pi(\bx)$ can have been recorded from the previous iteration), and the iteration cost of MTM roughly corresponds to that of evaluating the target density twice (because the weights in Steps 2 and 4 in MTM can be computed in parallel).}

We refer to the type of parallelization employed in MTM as \textit{in-step} parallelization, given that parallel computing is used within each algorithm iteration. It is to be contrasted with a basic use of parallel computing where one runs $N$ MH algorithms in parallel each using the proposal distribution $q_\sigma$. The advantage of the latter is that it is an embarrassingly parallel workload with essentially no communication cost (see, e.g., \cite{rosenthal2000parallel} and \cite{jacob2020unbiased}). It however does not reduce the burn-in required by each chain to reach stationarity, and thus its computational speed-up compared to running one MH algorithm is fundamentally limited if the chains have a large mixing time. As a result, in-step parallelization approaches that reduce burn-in can lead to significant gain in efficiency in the convergence phase; see, e.g., discussions in \cite{neal2003markov}, \cite{tjelmeland2004using}, \cite{frenkel2004speed}, \cite{calderhead2014general} and \cite{holbrook2022generating}.

In order to reduce the burn-in significantly enough to gain in efficiency with MTM, one has to choose carefully the weight function $w$. As mentioned, the most popular weight function is $w(\bx, \by) = \pi(\by) / \pi(\bx) \propto \pi(\by)$. The intuition behind choosing this weight function is to select $\by_j$ among the candidates $\by_1, \ldots, \by_N$ proportionally to its ``probability'' under the target. Although intuitively sensible, this choice lacks theoretical justification and, in fact, has been observed to yield an MTM algorithm with pathological behaviours. For example, \cite{Martino2017} highlights that MTM with $w(\bx, \by) \propto \pi(\by)$ may have issues of convergence if initialized in the tails of the target density because of acceptance probabilities that are near zero in this area, resulting in Markov chains that get stuck. Also, performance often decreases as $N$ increases, which is counter-intuitive as the algorithm has a larger group of candidates to select the proposal from at each iteration. In this paper, we show that the cause of those issues is precisely the choice of weight function $w(\bx, \by) \propto \pi(\by)$ which makes the resulting MTM \textit{globally balanced}, a qualifier coined by \cite{zanella2019informed} and that will be justified. Note that alternative choices are discussed in, e.g., \cite{Liu2000} and \cite{Pandolfi2010}, but these are similar in spirit to the choice $w(\bx, \by) \propto \pi(\by)$ and are exposed to the same pathologies.

The global objective of this work is to identify effective weight functions based on theoretical arguments and to study the resulting MTM algorithms. The effective weight functions that we identify, such as $w(\bx, \by) \propto \sqrt{\pi(\by)}$, yield MTM algorithms that have a connection with the \textit{locally-balanced} samplers proposed by \cite{zanella2019informed} in the context of discrete state-spaces. Therein, a large class of informed samplers are studied, with proposal distributions resulting from a combination of a random-walk proposal scheme and what is referred to as a \textit{balancing function}. A specific choice of balancing function makes the sampler \textit{locally balanced}. As already noted in \cite{zanella2019informed}, the weight function in MTM and the balancing function play an analogous role. Given the domination of the locally-balanced sampler over the \textit{globally-balanced} one shown in \cite{zanella2019informed}, it is thus natural to consider the use of locally-balanced weight functions in MTM, especially given that the resulting samplers have the same computational cost as their globally-balanced counterpart and are as easy to implement. In this paper, we refer to MTM with the weight function $w(\bx, \by) = \pi(\by) / \pi(\bx)$ as \textit{globally-balanced (GB) MTM} and to MTM with a locally-balanced weight function as \textit{locally-balanced (LB) MTM}. It will be highlighted in the following that LB MTM is closely related to gradient-based MCMC methods such as the \emph{Metropolis-adjusted Langevin algorithm} (MALA, \cite{roberts1996exponential}) and the \textit{Barker proposal scheme} of \cite{livingstone2019robustness}.

\subsection{Organization of the paper}\label{sec:contributions}

We now describe how the rest of the paper is organized. In \autoref{sec:ideal}, we establish a weak convergence of the Markov chains simulated by MTM towards those simulated by an \textit{ideal} sampling scheme, as $N \rightarrow \infty$ with the dimension $d$ fixed. This ideal sampling scheme is the continuous state-space counterpart of the class of informed samplers studied in \cite{zanella2019informed}. In \autoref{sec:convergence}, we study the convergence phase of MTM algorithms through both theoretical and empirical results. The analysis shows that GB MTM indeed often has issues of convergence, while LB MTM does not. In fact, LB MTM provides drastic convergence speed-ups and improved robustness compared to GB MTM. In \autoref{sec:scaling_convergence}, we study MTM properties in stationarity. \autoref{theorem:scaling_limit_ideal} provides high-dimensional scaling-limit results (as $d \rightarrow \infty$) for the ideal scheme with $w(\bx, \by) \propto \pi(\by)$ and that with a LB weight function, namely $w(\bx,\by)\propto \sqrt{\pi(\by)}$.
\autoref{theorem:MTM_app_ideal} then shows that these ideal schemes are approached by MTM under conditions on the rate at which the number of candidates increases with the dimension. Combining Theorems \ref{theorem:scaling_limit_ideal} and \ref{theorem:MTM_app_ideal} leads to an informative but somewhat negative conclusion: the ideal scheme with the LB weight function scales better with dimension, however, the (sufficient) condition under which the corresponding MTM algorithm reaches its full potential by approaching the ideal scheme is that $N$ scales exponentially with the dimension. As a result, in common situations, LB MTM provides only mild performance improvement compared to GB MTM in stationarity. The theoretical results of Sections \ref{sec:convergence} and \ref{sec:scaling_convergence} are derived in simple scenarios and numerical results are provided to illustrate them, including adaptive MTM implementations. To explore whether the scope of those theoretical results extends beyond simple scenarios, we study in \autoref{sec:real_example} the application of GB and LB MTM to a real-world inference problem of immunotherapy in precision medicine where the likelihood is expensive to compute and its gradient is not available in closed form. The empirical results are consistent with the theoretical ones, including observing a drastically reduced burn-in time for LB MTM. We finish the manuscript with a discussion in \autoref{sec:discussion}. The proofs of all theoretical results are deferred to \autoref{sec:proofs}. The code to reproduce all numerical results is available online.\footnote{See ancillary files on \href{https://arxiv.org/abs/2211.11613}{arXiv:2211.11613}.}

While working on our manuscript, it came to our attention that \cite{chang2022rapidly} independently and concurrently propose to use LB weight functions within MTM and study the resulting samplers, but the context and focus are quite different. The context and focus are that of sampling from target distributions defined on discrete state-spaces and more precisely from target distributions resulting from model-selection problems. Our contributions and theirs thus have virtually no overlap. That being said, the conclusions in \cite{chang2022rapidly} and ours are consistent, and in this sense, the two studies are complementary.

\section{Ideal schemes and locally-balanced weight functions}\label{sec:ideal}

We start in \autoref{sec:justification} with the identification of LB weight functions as effective weight functions. The arguments motivating the use of such weight functions rest upon a theoretical result for which a sketch of a proof was presented in \cite{Liu2000}. The result is stated informally in \autoref{sec:justification}, while a formal statement is presented in \autoref{sec:convergence_ideal}. In \autoref{sec:connection_gradient}, a connection between LB MTM and gradient-based methods is highlighted.

\subsection{Locally-balanced weight functions: A motivation}\label{sec:justification}

\cite{Liu2000} presented a sketch of a proof of the convergence of the Markov kernel of MTM as $N \rightarrow \infty$ to that of a MH algorithm using a proposal distribution with a PDF defined as
\begin{align*}
 Q_{w, \sigma}(\bx, \by) := \frac{w(\bx, \by) \, q_{\sigma}(\bx, \by)}{\int w(\bx, \by) \, q_{\sigma}(\bx, \by) \, \d\by},
\end{align*}
 assuming that the integral in the denominator exists and is finite. We will refer to the MH algorithm using the proposal distribution $Q_{w, \sigma}$ as an \textit{ideal} scheme as it cannot in general be implemented and it can be thought of as an MTM algorithm sampling an infinite number of candidates $N$ at each iteration.\footnote{For consistency, we will use the terminology ``ideal scheme'' also for the globally-balanced version even if using a large number of candidates $N$ is not effective in that case.}

Whenever there exists a positive continuous function $g$ such that $w(\bx,\by)$ is of the form $w(\bx,\by)=g(\pi(\by)/\pi(\bx))$, the ideal scheme corresponds to the continuous version of the class of informed samplers studied in \cite{zanella2019informed}. All the weight functions considered in our paper are of this form. Even though the work of \cite{zanella2019informed} is in the context of discrete state-spaces, a part of the analysis conducted therein is applicable to the continuous case as well. In particular, it indicates that the choice $w(\bx, \by) \propto \pi(\by)$ yields a proposal distribution which leaves the target distribution invariant (without the MH correction) in the situation where $\sigma \rightarrow \infty$, the approach being, in that sense, global; in fact, the limiting case corresponds to independent sampling. This justifies the fact that, MTM using the weight function $w(\bx, \by) \propto \pi(\by)$ is coined \textit{globally-balanced} (GB) MTM. Analogously, the ideal scheme using the proposal PDF $Q_{w, \sigma}(\bx, \by) \propto \pi(\by) \, q_{\sigma}(\bx, \by)$ will be referred to as the \textit{globally-balanced} (GB) ideal scheme.

The problem with using such a function $w$ is that, in high dimensions, the scale parameter of typical MCMC algorithms is required to be small to avoid near-zero acceptance probabilities, as indicated by the optimal-scaling theory (see, e.g., \cite{bedard2012scaling} in the specific context of MTM). 
As $\sigma \rightarrow 0$, using $Q_{w, \sigma}(\bx, \by) \propto \pi(\by) \, q_{\sigma}(\bx, \by)$ leaves $\pi^2$ invariant (without the MH correction), instead of $\pi$. In high dimensions, when $\sigma$ is small, there is thus a significant discrepancy between the proposal and target distributions, and this causes the pathological behaviours of GB MTM mentioned previously.

Locally-balanced (LB) proposal distributions aim at correcting this discrepancy. Indeed, by construction, LB proposal distributions leave $\pi$ invariant (without the MH correction) in the limiting situation where $\sigma \rightarrow 0$, which is the regime in agreement with high-dimensional settings. In our context, within ideal schemes, LB proposal distributions are such that $Q_{w, \sigma}(\bx, \by) \propto g(\pi(\by)/\pi(\bx)) \, q_{\sigma}(\bx, \by)$, where the balancing function $g$ is a positive continuous function such that $g(x) / g(1 / x) = x$ for $x > 0$. We thus propose to set the weight function in MTM to $w(\bx, \by) = g(\pi(\by)/\pi(\bx))$ with $g$ satisfying these conditions. Several functions $g$ satisfy these conditions (see, e.g., \cite{zanella2019informed}, \cite{sansone2022lsb} and \cite{vogrinc2022optimal}). In this paper, we focus on two of them which have been thoroughly studied in other contexts \citep{power2019accelerated, gagnon2020asymptotic, gagnon2019RJ, sun2021path, hird2022fresh, liang2022adaptive, livingstone2019robustness, sun2022optimal, sun2022discrete, zhou2022rapid}, namely $g(x) = \sqrt{x}$ and $g(x) = x /(1 + x)$, the latter yielding what is called the Barker proposal distribution in reference to \cite{barker1965monte}'s acceptance-probability choice. The ideal scheme using the proposal PDF $Q_{w, \sigma}(\bx, \by) \propto g(\pi(\by)/\pi(\bx)) \, q_{\sigma}(\bx, \by)$ with $g$ satisfying the conditions above will thus be referred to as the \textit{locally-balanced} (LB) ideal scheme. This justifies the fact that, MTM using the weight function $w(\bx, \by) \propto g(\pi(\by)/\pi(\bx))$ with $g$ satisfying the conditions above is coined \textit{locally-balanced} (LB) MTM.

LB MTM will be seen to not exhibit the pathological behaviours mentioned previously. Also, LB MTM with $g(x) = \sqrt{x}$ will be seen to have an advantage over that with the Barker weight function in terms of convergence speed-ups. This advantage has been observed in other contexts \citep{zhou2022rapid}. Significant convergence speed-ups with $g(x) = \sqrt{x}$ have also been observed for MTM in \cite{chang2022rapidly}. In stationarity, the performance of LB MTM with $g(x) = \sqrt{x}$ is similar to that with the Barker weight function. All that suggests the following practical recommendation to MTM users: \emph{use LB weight functions, and more specifically, $g(x) = \sqrt{x}$.}

\subsection{Convergence towards ideal schemes}\label{sec:convergence_ideal}

To understand why the convergence of MTM towards the ideal scheme might hold, it is useful to have a characterization of the distribution of a proposal sampled using MTM, that we denote by $\bY_J$ with a capital $J$ to represent that the choice among the candidates for the proposal is random. This distribution is conditional on the current state of the Markov chain $\bx$, and we use $\E_{\bx}$ to denote an expectation with respect to the associated PDF that depends on $\bx$. The PDF is based on the product measure $\prod_{i=1}^N q_\sigma(\bx, \by_i) \, \d\by_{1:N}$, where $\mathbf{y}_{1:N} := (\mathbf{y}_1, \ldots, \mathbf{y}_{N})$.  That characterization uses that $\bY_1, \ldots, \bY_N$ are conditionally independent and identically distributed (IID) given $\bx$.

\begin{Proposition}\label{prop:dist_MTM}
 Given a current state $\bx$ and function $h$, a proposal $\mathbf{Y}_{J}$ sampled using MTM is such that
 \begin{align*}
  \E_{\bx}[h(\mathbf{Y}_{J})] &= \int h(\by_{1}) \, \frac{w(\bx, \by_1)}{\frac{1}{N}\sum_{i = 1}^{N} w(\bx, \by_i)} \prod_{i=1}^{N} q_{\sigma}(\bx, \by_i) \, \d\by_{1:N}.
 \end{align*}
\end{Proposition}
The expectation in \autoref{prop:dist_MTM} is to be compared with that of $h(\bY)$ with $\bY \sim Q_{w, \sigma}(\bx, \cdot \,)$, given a current state $\bx$, that can be written as
 \begin{align}\label{eq:expectation_ideal}
  \E_{\bx}[h(\bY)] &= \int h(\by_{1}) \, \frac{w(\bx, \by_1)}{\int w(\bx, \by_1) \, q_{\sigma}(\bx, \by_1) \, \d\by_1} \prod_{i=1}^{N} q_{\sigma}(\bx, \by_i) \, \d\by_{1:N}.
 \end{align}
It is apparent that the expectation in \autoref{prop:dist_MTM} is an approximation that in \eqref{eq:expectation_ideal} (and that one sampling scheme approximates the other), and that, presumably, the approximation becomes more accurate as $N$ increases.

We now present a formal result about the weak convergence of Markov chains simulated by MTM towards those simulated by the ideal sampling scheme, under some conditions. Let us denote a Markov chain simulated by MTM as $\{\bX_N(m): m \in \N\}$ and that simulated by the corresponding ideal algorithm by $\{\bX_{\text{ideal}}(m): m \in \N\}$. Also, let us denote the Euclidean norm of a vector $\bx$ by $\|\bx\|$.

\begin{Theorem}\label{theorem:weak_con_MTM_d_fixed}
 Assume that $\E[w(\bX, \bY_1)^4] < \infty$ and $\E[w(\bX, \bY_1)^{-4}] < \infty$ with $\bX \sim \pi$ and $\bY_1 \mid \bX \sim q_\sigma(\bX, \cdot \,)$. As $N \rightarrow \infty$,
\begin{enumerate}

   \item given any state $\bx$, the total variation between the distribution of a proposal $\mathbf{Y}_{J}$ sampled using MTM and $Q_{w, \sigma}(\bx, \cdot \,)$ converges to 0 at a rate of $1/\sqrt{N}$;
   \item if additionally
   \begin{description}

    \item[(a)] $\pi$, $Q_{w, \sigma}(\, \cdot \,, \by)$ and $Q_{w, \sigma}(\by, \cdot \,)$ are continuous, for any $\by$,
    \item[(b)] for all $\bx \in \R^d$, there exists an $\varepsilon > 0$ and an integrable function $f(\bx, \cdot \,)$ such that 
        \[
         \sup_{\{\boldsymbol\epsilon \in \R^d: \|\boldsymbol\epsilon\| \leq \varepsilon\}} Q_{w, \sigma}(\bx + \boldsymbol\epsilon, \by) \leq f(\bx, \by),
        \]
          for all $\by \in \R^d$,
   \end{description}
   then $\{\bX_N(m): m \in \N\}$ converges weakly towards $\{\bX_{\text{ideal}}(m): m \in \N\}$ provided that $\bX_N(0) \sim \pi$ and $\bX_{\text{ideal}}(0) \sim \pi$.
\end{enumerate}
\end{Theorem}

This result indicates that MTM can be seen as an approximation to the ideal MH scheme using $Q_{w, \sigma}$ as a proposal distribution. The latter will thus be considered instead of the former for theoretical analyses in the next sections. The advantage of doing so is that the ideal scheme samples only one candidate at each iteration and is thus easier to analyse. A refined version of \autoref{theorem:weak_con_MTM_d_fixed} will be provided in \autoref{sec:convergence_MTM}; it allows to quantitatively evaluate the discrepancy between the chains simulated by MTM and the ideal scheme.

We finish this section with a result indicating that Assumption (b) in part 2 of \autoref{theorem:weak_con_MTM_d_fixed} is verified in great generality.

\begin{Proposition}\label{prop:ass_b}
 Assume that $\pi$ is upper bounded, $q_\sigma(\bx, \cdot \,) = \mathcal{N}(\mathbf{x}, \sigma^2 \I_d)$, and $w(\bx, \by) = \pi(\by) / \pi(\bx)$ or $w(\bx, \by) = \sqrt{\pi(\by) / \pi(\bx)}$. Then, Assumption (b) in part 2 of \autoref{theorem:weak_con_MTM_d_fixed} is satisfied.
\end{Proposition}

\subsection{Locally-balanced MTM as a gradient-free alternative}\label{sec:connection_gradient}

An indirect connection can be established between MTM using $g(x) = \sqrt{x}$ and MALA, and between MTM using $g(x) = x /(1 + x)$ and the MH sampler based on the Barker proposal scheme of \cite{livingstone2019robustness}. Recall that both MALA and the Barker scheme are gradient-based MCMC algorithms which require that $\pi$ is differentiable and that $\nabla \log\pi$ can be evaluated pointwise. Interestingly, the MALA proposal can be viewed as an approximation to the ideal scheme using $g(x) = \sqrt{x}$, whereas the Barker proposal can be viewed as an approximation to that using $g(x) = x /(1 + x)$. Indeed, they both result from an approximation of $Q_{w, \sigma}(\bx, \cdot \,)$ based on a first-order Taylor series expansion of $\log \pi$:
$$
Q_{w, \sigma}(\bx, \by \,) \propto g(\e^{\log\pi(\by)-\log\pi(\bx)}) \, q_\sigma(\bx, \by) \approx g(\e^{(\nabla \log \pi(\bx))^T (\by - \bx)}) \, q_\sigma(\bx, \by). 
$$
When $g(x) = \sqrt{x}$ and $q_\sigma(\bx, \cdot \,) = \mathcal{N}(\mathbf{x}, \sigma^2 \I_d)$, normalizing the last expression gives exactly the MALA proposal, while, when $g(x) = x /(1 + x)$, it gives the Barker scheme (see, respectively, Section 5 of \citet{zanella2019informed} and Section 3 of \citealp{livingstone2019robustness}).
Thus, we can see MALA and the Barker proposal scheme as gradient-based approximations of the same ideal schemes as those approximated by LB MTM samplers with $g(x) = \sqrt{x}$ and $g(x) = x /(1 + x)$. The approximations used in MTM are different in several aspects. First, they are stochastic, by opposition to deterministic as in the gradient-based methods. Second, one has control over the approximations (through $N$).  Last but not least, the algorithms do not require to evaluate the gradient of $\log \pi$, which is advantageous when evaluating the gradient is either infeasible or computationally intensive.

GB MTM approximates an ideal scheme which does not correspond to any known first-order method. A rich body of literature has shown that MALA and the MH sampler with the Barker scheme behave quantitatively and qualitatively differently (in terms of ergodicity measure or scaling-limit regime) to most zero-th order methods such as random-walk Metropolis (see, e.g., \cite{roberts1998optimal}, \cite{bou2013nonasymptotic}, \cite{dwivedi2018log} and \cite{livingstone2019robustness}). It is thus important to study whether LB MTM inherits some of those favourable properties, which cannot be expected by GB MTM.

\section{Performance during the convergence phase}\label{sec:convergence}

In this section, we evaluate the performance during the convergence phase\footnote{By convergence phase, often called \emph{burn-in} in the MCMC literature, we mean the iterations until the Markov chain simulated by MTM is close in distribution to $\pi$.} of MTM with the different weight functions.
 We do this by analysing the acceptance probabilities in the tails in \autoref{sec:tails} and by empirically measuring the convergence time of adaptively tuned MTM in \autoref{sec:conv_numerics}. The adaptive tuning aims to represent how one would use and tune MTM in practice. As mentioned previously, the acceptance probabilities are near-zero in the tails with GB MTM, unless the step size $\sigma$ is made extremely small, which causes convergence issues in either case. Our analysis shows that these convergence issues do not arise with LB MTM. Also, numerical results show that MTM with the Barker weight function has higher acceptance probabilities than that with $g(x) = \sqrt{x}$. This advantage has been observed in other contexts \citep{zanella2019informed, livingstone2019robustness}, and is attributed to the boundedness of the function $g(x) = x /(1 + x)$. In the MTM context, it yields more stability in the approximation of the ideal scheme. Even though the unboundedness of $g(x) = \sqrt{x}$ yields less stability, it leads to more persistent movement from the tails to the high-probability region, which is shown in \autoref{sec:conv_numerics} to provide better convergence performance.

\subsection{Acceptance probabilities in the tails}\label{sec:tails}

In this section, we analyse the behaviour of MTM when initialized in the tails by evaluating the conditional expected acceptance probability, given an initial state $\bx$ with $\|\bx\|$ large, where the expectation is with respect to the random variables involved in the proposal mechanism. A low conditional expected acceptance probability implies that it is likely that the chain gets stuck and that an issue arises in terms of convergence to the target distribution as the algorithm progresses. The analysis rests upon a theoretical result about the ideal scheme. The result is established under a specific and simple scenario: the target density factorizes, and more precisely,
\begin{align}\label{eqn:target_product}
 \pi(\bx) = \prod_{i=1}^d \varphi(x_i), \quad \bx := (x_1, \ldots, x_d)^T \in \R^d,
\end{align}
and $q_\sigma(\bx, \cdot \,) = \mathcal{N}(\mathbf{x}, \sigma^2 \I_d)$, where $\varphi$ is the PDF of a standard normal distribution. The normal assumption and the factorization allow to make precise calculations and in particular to establish that $Q_{w, \sigma}$ is a normal distribution when the weight function factorizes as well, that is when using for instance $w(\bx, \by) = \pi(\by)/\pi(\bx)$ or $w(\bx, \by) = \sqrt{\pi(\by)/\pi(\bx)}$.

We acknowledge that the scenario limits the scope of the analysis. Note that, in \autoref{sec:proofs_upper_bound_acc}, we provide a result which is less precise, but which holds under weaker assumptions. With the result provided in \autoref{sec:proofs_upper_bound_acc}, we cannot conduct an analysis as thorough as that performed below. The assumptions are essentially that $U := - \log \pi$ is strongly convex and $L$-smooth, instead of assuming that the target density factorizes into a product of normal densities. The factorization assumption has a long history in analysis of MCMC, especially in the scaling-limit literature where it is a standard assumption (see, e.g., \cite{roberts1997weak}, \cite{roberts1998optimal}, \cite{bedard2007weak}, \cite{bedard2012scaling}, \cite{durmus2017opimal} and \cite{GAGNON201932}). The factorization is an important structural limitation which implies independence of the random variables. The normal assumption can be justified in Bayesian large-sample regimes where the models are regular enough\footnote{Models that are regular enough are those which satisfy regularity conditions; see \cite{schmon2021optimal} for more details.} \citep{schmon2021optimal}, but it is an important limitation as well. We thus expect the results to be informative at least when MTM is used to sample from a posterior distribution resulting from a large data set ($n \gg d$) and a regular model, provided that the model parameters are \textit{a posteriori} weakly dependent. The same scenario will be considered for the scaling-limit analysis in \autoref{sec:scaling_convergence}.

We now present the result in which we use the notation $\alpha_{\text{ideal}}$ for the acceptance probability in the ideal MH scheme.

\begin{Proposition}\label{prop:upper_bound_acc}
 Consider a current state $\bx$ and that $\bY \sim Q_{w, \sigma}(\bx, \cdot \,)$ with $q_\sigma(\bx, \cdot \,) = \mathcal{N}(\mathbf{x}, $ $ \sigma^2 \I_d)$ and $w(\bx, \by) = \pi(\by)/\pi(\bx)$. If the target distribution is defined as in \eqref{eqn:target_product},
\begin{align}\label{eqn:bound_acc_prob_glo}
 \E_{\bx}[\alpha_{\text{ideal}}(\bx, \bY)] \leq \exp\left(-\|\mathbf{x}\|^2 \frac{\sigma^2}{2((1+\sigma^2)^2 - \sigma^2)}\right) \left(1 - \frac{\sigma^2}{(1 + \sigma^2)^2}\right)^{-d/2}.
\end{align}
In particular, for any $\sigma$ and $d$, it holds that $\lim_{\|\bx\| \rightarrow \infty}  \E_{\bx}[\alpha_{\text{ideal}}(\bx, \bY)] = 0$.
\end{Proposition}

\autoref{prop:upper_bound_acc} highlights a pathological behaviour as most MCMC methods do not have acceptance probabilities that are near zero when the current state is in the tails of the target density. We obtained the corresponding upper bound for the ideal scheme with $w(\bx, \by) = \sqrt{\pi(\by)/\pi(\bx)}$ and it does not converge to 0. Of course, this does not guarantee mathematically that the conditional expected acceptance probability of that scheme is not near zero in the tails, but it indicates a significant difference. We also tried deriving other upper bounds to see if they yield different results, but we did not obtain any that allows to conclude otherwise. We provide below numerical results for both the ideal scheme and MTM using $w(\bx, \by) = \sqrt{\pi(\by)/\pi(\bx)}$ which corroborate those findings and suggest that the acceptance probabilities do not converge to 0 as the current/initial state gets further and further in the tails.

The result provided in \autoref{sec:proofs_upper_bound_acc} essentially states that $\lim_{\|\bx\| \rightarrow \infty}  \E_{\bx}[\alpha_{\text{ideal}}(\bx, \bY)] = 0$, for any $\sigma$ and $d$, when $U := - \log \pi$ is strongly convex and $L$-smooth. While being interesting, it does not provide an explicit upper bound on the conditional expected acceptance probability as in \eqref{eqn:bound_acc_prob_glo} and does not allow for a precise characterization in high-dimensional regimes where $d \rightarrow \infty$. Given that we are interested in such regimes, we study the implications of \eqref{eqn:bound_acc_prob_glo} when $d \rightarrow \infty$, with $\bx$ and $\sigma$ functions of $d$. Such a study allows to characterize the relation between $d$ and the location of $\bx$ in the tails in the situations where there are convergence issues with the GB ideal scheme and MTM. We highlight a dependence on $d$ of $\bx, \bY, \pi$ and $\sigma$ by denoting these for the rest of the section by $\bx_d, \bY_d, \pi_d$ and $\sigma_d$. For the analysis, we consider that $\sigma_d^2 = \ell^2 / d$ and $\|\bx_d\| = d^\kappa$ with $\kappa$ a positive constant; setting $\sigma_d^2 = \ell^2 / d$ will be seen to be an effective way of scaling $\sigma$ with $d$. With these choices, the conclusion is the following: \emph{the conditional expected acceptance probability of the GB ideal scheme \eqref{eqn:bound_acc_prob_glo} converges to 0 when $\kappa > 1/2$, implying that it is sufficient for $\|\bx_d\|$ to grow with $d$ at any rate faster than $\sqrt{d}$ to lead to near-zero acceptance probabilities.} We highlight that, with a target distribution such as that defined in \eqref{eqn:target_product}, $\|\bx_d\| = d^\kappa$ with $\kappa$ around $0.5$ is not even far in the tails of the density as $\|\bX_d\|^2$ has a chi-squared distribution with a mean of $d$ and a standard-deviation of $\sqrt{2 d}$. This implies that even a random initialization of GB MTM using a distribution slightly different from the target may lead to issues of convergence to the target distribution as the algorithm progresses.

We present in Figures \ref{fig:exp_acc_prob} and \ref{fig:exp_acc_prob_d_50} numerical results which complete the analysis. In both figures, we provide conditional expected acceptance probabilities as a function of $\kappa$ in the situation where the target distribution is defined as in \eqref{eqn:target_product}, $q_{\sigma_d}(\bx_d, \cdot \,) = \mathcal{N}(\mathbf{x}_d, \sigma_d^2 \I_d)$ with $\ell = 2.38$ and the current/initial state $\bx_d$ is set to $\bx_d = (d^{\kappa - 1/2}, \ldots, d^{\kappa - 1/2})$, ensuring that $\|\bx_d\| = d^\kappa$. The value $\ell = 2.38$ will be seen in \autoref{sec:scaling-limits} to be optimal for the GB ideal scheme in a high-dimensional regime. The expectations are approximated using independent Monte Carlo sampling. The approximations are based on samples of size 1,000,000.

The difference between \autoref{fig:exp_acc_prob} and \autoref{fig:exp_acc_prob_d_50} is that in the former the results are for the ideal schemes, whereas in the latter they are for MTM. The results in \autoref{fig:exp_acc_prob_d_50} are for $d = 50$; we observed similar results when $d = 200$. The results for GB samplers are consistent with the theoretical result about the convergence to 0 of the expectation in \eqref{eqn:bound_acc_prob_glo} when $\kappa > 1/2$, with conditional expected acceptance probabilities close to 1 for $\kappa$ smaller than $0.5$ (for moderate to high dimensions, and moderate values of $N$ for MTM), followed by a sharp drop around $\kappa = 0.5$. In \autoref{fig:exp_acc_prob} (a), we notice that the conditional expected acceptance probability converges to 0 even when $d = 5$ (thus in the case where the high-dimensional regime is not attained); this is because $\ell$ is not large enough to yield an algorithm that performs approximately IID sampling (recall the discussion towards the end of \autoref{sec:ideal}), suggesting that the conclusion of \autoref{prop:upper_bound_acc} holds.

The results in \autoref{fig:exp_acc_prob} (b) and \autoref{fig:exp_acc_prob_d_50} (b)-(c) suggest that LB schemes do not have issues of convergence to the target distribution when initialized in the tails. They also suggest that using the Barker weight function in MTM leads to higher acceptance probabilities. As mentioned, we attribute the difference to the fact that, with MTM, the normalizing constant of $Q_{w, \sigma}(\bx, \cdot \,)$ needs to be approximated and the boundedness of the function $g(x) = x /(1 + x)$ yields more stability in the approximation.

 \begin{figure}[ht]
  \centering
  \footnotesize
  $\begin{array}{cc}
 \vspace{-2mm}\hspace{-2mm}\includegraphics[width=0.50\textwidth]{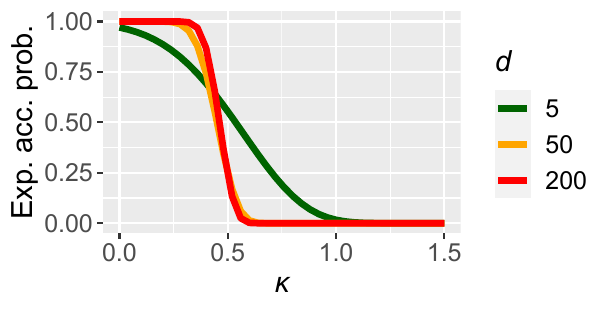} &  \hspace{-5mm} \includegraphics[width=0.50\textwidth]{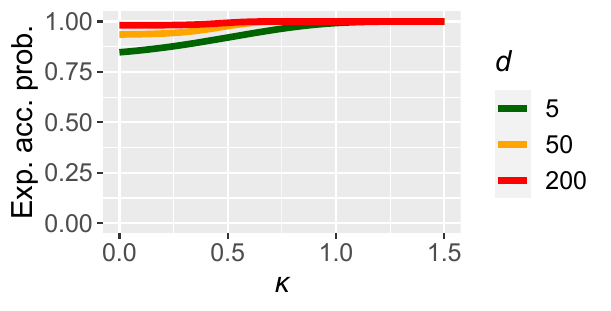}  \cr
   \hspace{-6mm} \textbf{(a) Ideal scheme w. $w(\mathbf{x}_d, \mathbf{y}_d) = \pi_d(\mathbf{y}_d) / \pi_d(\mathbf{x}_d)$} & \hspace{-6mm} \textbf{(b) Ideal scheme w. $w(\mathbf{x}_d, \mathbf{y}_d) = \sqrt{\pi_d(\mathbf{y}_d) / \pi_d(\mathbf{x}_d)}$} \cr
  \end{array}$\vspace{-2mm}
  \caption{\small Conditional expected acceptance probability as a function of $\kappa$ when $\mathbf{x}_d = (d^{\kappa - 1/2}, \ldots, d^{\kappa - 1/2})$ and $\ell = 2.38$, for several values of $d$ and: (a) the ideal scheme with $w(\mathbf{x}_d, \mathbf{y}_d) = \pi_d(\mathbf{y}_d) / \pi_d(\mathbf{x}_d)$, and (b) the ideal scheme with $w(\mathbf{x}_d, \mathbf{y}_d) = \sqrt{\pi_d(\mathbf{y}_d) / \pi_d(\mathbf{x}_d)}$.}\label{fig:exp_acc_prob}
 \end{figure}
\normalsize

  \begin{figure}[ht]
  \centering\scriptsize
  $\begin{array}{ccc}
 \vspace{-2mm}\hspace{-2mm}\includegraphics[width=0.34\textwidth]{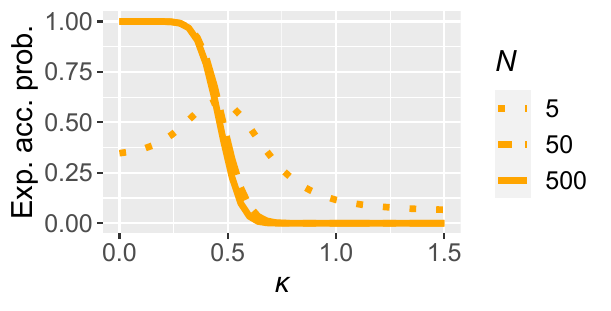} &  \hspace{-5mm} \includegraphics[width=0.34\textwidth]{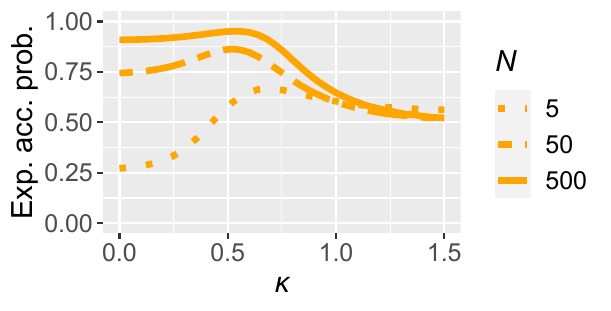} &  \hspace{-5mm} \includegraphics[width=0.34\textwidth]{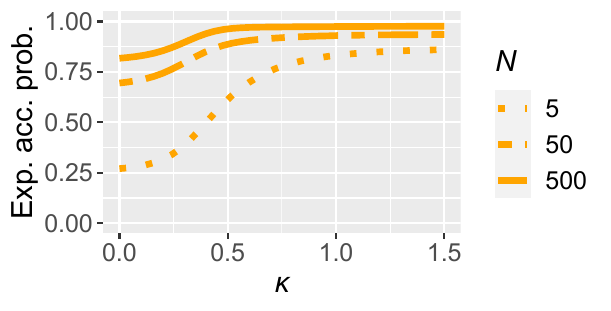} \cr
   \hspace{-5mm} \textbf{(a) MTM w. $w(\mathbf{x}_d, \mathbf{y}_d) = \frac{\pi_d(\mathbf{y}_d)}{\pi_d(\mathbf{x}_d)}$} & \hspace{-7mm} \textbf{(b) MTM w. $w(\mathbf{x}_d, \mathbf{y}_d) = \sqrt{\frac{\pi_d(\mathbf{y}_d)}{\pi_d(\mathbf{x}_d)}}$} & \hspace{-5mm} \textbf{(c) MTM w. $w(\mathbf{x}_d, \mathbf{y}_d) = \frac{\frac{\pi_d(\mathbf{y}_d)}{ \pi_d(\mathbf{x}_d)}}{1 + \frac{\pi_d(\mathbf{y}_d)}{\pi_d(\mathbf{x}_d)}}$} \cr
  \end{array}$\vspace{-2mm}
  \caption{\small Conditional expected acceptance probability as a function of $\kappa$ when $\mathbf{x}_d = (d^{\kappa - 1/2}, \ldots, d^{\kappa - 1/2})$, $\ell = 2.38$ and $d = 50$, for several values of $N$ and: (a) MTM with $w(\mathbf{x}_d, \mathbf{y}_d) = \pi_d(\mathbf{y}_d) / \pi_d(\mathbf{x}_d)$, (b) MTM with $w(\mathbf{x}_d, \mathbf{y}_d) = \sqrt{\pi_d(\mathbf{y}_d) / \pi_d(\mathbf{x}_d)}$, and (c) MTM with $w(\mathbf{x}_d, \mathbf{y}_d) = (\pi_d(\mathbf{y}_d) / \pi_d(\mathbf{x}_d))/(1+ \pi_d(\mathbf{y}_d) / \pi_d(\mathbf{x}_d))$.}\label{fig:exp_acc_prob_d_50}
 \end{figure}
\normalsize

\subsection{Convergence to stationarity: simulations with adaptive MCMC}\label{sec:conv_numerics}

In this section, we perform numerical simulations with adaptive MTM schemes, where the scale parameter $\sigma$ is tuned on the fly while the MTM algorithm progresses. Such simulations allow to: (i) reduce the sensitivity of the simulation set-up to the choice of a specific value for $\sigma$; (ii) assess the impact of the choice of weight function in a more advanced and realistic MTM implementation (arguably closer to one that a careful practitioner would use).

The study in this section is non-asymptotic; the mathematical objects like the target distribution, the scale parameter and the states are thus denoted without a subscript $d$, that is $\pi$, $\sigma$ and $\bx$. Algorithm 4 found in Section 5 of \cite{andrieu2008tutorial} is used to adaptively tune $\sigma$. The algorithm targets an acceptance rate to adapt tuning parameters. The targeted acceptance rates are $25\%$ and $50\%$ for GB and LB MTM, respectively. These targets are chosen according to theoretical and empirical results presented in the next sections. Algorithm 4 of \cite{andrieu2008tutorial} also uses a learning rate $\gamma(m)$, which here is set to $m^{-0.6}$, $m$ representing the iteration index. A power of $-0.6$ allows to reach a good balance between fast adaptation and stability in this example. We experimented with different power values and obtained similar conclusions.

The results are presented in Figures \ref{fig:Traceplots} and \ref{fig:Boxplots}. \autoref{fig:Traceplots} displays trace plots for GB MTM and LB MTM with $q_{\sigma}(\bx, \cdot \,) = \mathcal{N}(\mathbf{x}, \sigma^2 \I_d)$ and target distribution as in \eqref{eqn:target_product} with $d = 50$. Trace plots of $\|\bX(m)\|$ are presented, with a log-scale on the $x$-axis, in the situation where the algorithms are initialized from $\bx = (10, \dots, 10)$, which corresponds to $\mathbf{x} = (d^{\kappa - 1/2}, \ldots, d^{\kappa - 1/2})$ with $\kappa \approx 1.09$. We observe the pathological behaviour of GB MTM described before: increasing $N$ deteriorates the convergence performance up to having to set $\sigma$ to near-zero values to achieve non-negligible acceptance rate when $N$ equals $50$ or $500$. We observe the opposite and desirable results for LB MTM, whose convergence speed increases with $N$. To provide a more quantitative picture, \autoref{fig:Boxplots} shows the convergence times for the same algorithms, target distribution and starting state as in \autoref{fig:Traceplots}. The results are obtained from 100 independent runs for each algorithm using different values of $N$. Here, the convergence time is defined as the first time the chain reaches the 95th percentile of $\|\bX\|$ under the target distribution. \autoref{fig:Boxplots} also presents analogous results for a different target distribution, namely a $50$-dimensional product of standard Laplace distributions. The results are consistent with \autoref{fig:Traceplots}, with LB MTM providing a smooth and regular improvement in performance with $N$, unlike GB MTM. Also, it is interesting to note the difference between the LB MTM with the Barker weight function and that with $g(x) = \sqrt{x}$. The performance of the former stabilizes quicker as $N$ increases, and do not reach to same level of improvement as the latter.

  \begin{figure}[ht]
  \centering\footnotesize
  $\begin{array}{ccc}
 \vspace{-2mm}\hspace{-2mm}\includegraphics[width=0.34\textwidth]{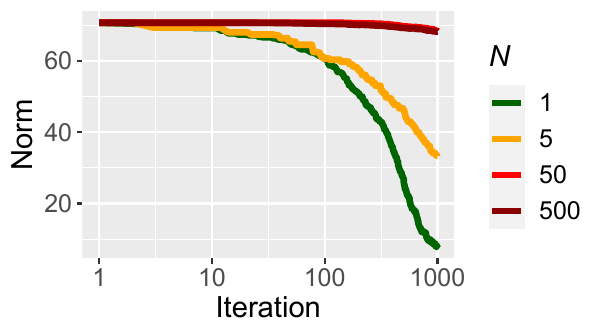} &  \hspace{-5mm} \includegraphics[width=0.34\textwidth]{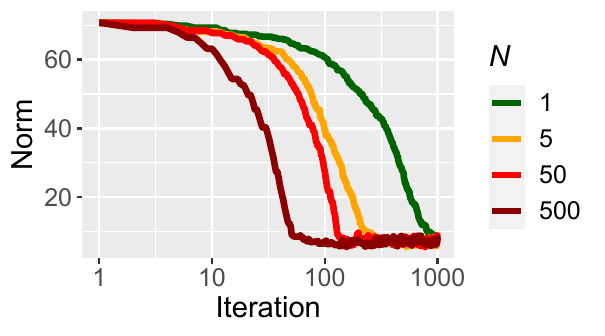} &  \hspace{-5mm} \includegraphics[width=0.34\textwidth]{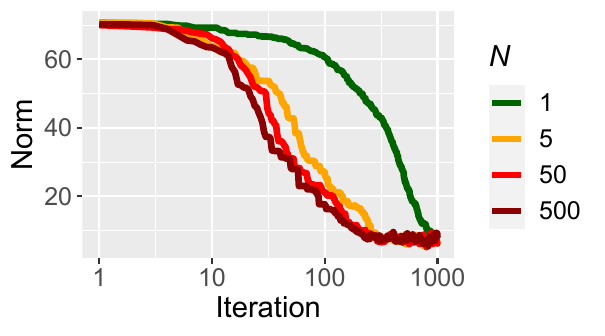} \cr
   \hspace{-5mm} \textbf{(a) MTM w. $w(\mathbf{x}, \mathbf{y}) = \frac{\pi(\mathbf{y})}{\pi(\mathbf{x})}$} & \hspace{-7mm} \textbf{(b) MTM w. $w(\mathbf{x}, \mathbf{y}) = \sqrt{\frac{\pi(\mathbf{y})}{\pi(\mathbf{x})}}$} & \hspace{-6mm} \textbf{(c) MTM w. $w(\mathbf{x}, \mathbf{y}) = \frac{\pi(\mathbf{y}) / \pi(\mathbf{x})}{1 + \pi(\mathbf{y}) / \pi(\mathbf{x})}$} \cr
  \end{array}$\vspace{-2mm}
\caption{\small Trace plots of the Euclidean norm of the state when $d=50$, the scale parameter is adaptively tuned and the initial state is $\mathbf{x} = (10, \ldots, 10)$, for several values of $N$ and: (a) GB MTM, (b) LB MTM with $w(\mathbf{x}, \mathbf{y}) = \sqrt{\pi(\mathbf{y}) / \pi(\mathbf{x})}$, and (c) LB MTM with $w(\mathbf{x}, \mathbf{y}) = (\pi(\mathbf{y}) / \pi(\mathbf{x}))/(1+ \pi(\mathbf{y}) / \pi(\mathbf{x}))$;
the scale on the $x$-axis is logarithmic.}\label{fig:Traceplots}
 \end{figure}
\normalsize

  \begin{figure}[ht]
  \centering\footnotesize
  $\begin{array}{ccc}
  \multicolumn{3}{c}{\large\textbf{Target: normal distribution}} \cr
 \vspace{-2mm}\hspace{-2mm}\includegraphics[width=0.34\textwidth]{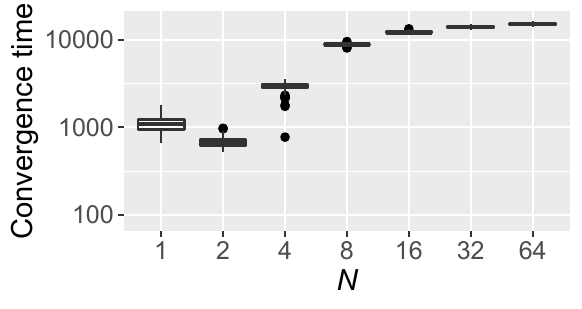} &  \hspace{-5mm} \includegraphics[width=0.34\textwidth]{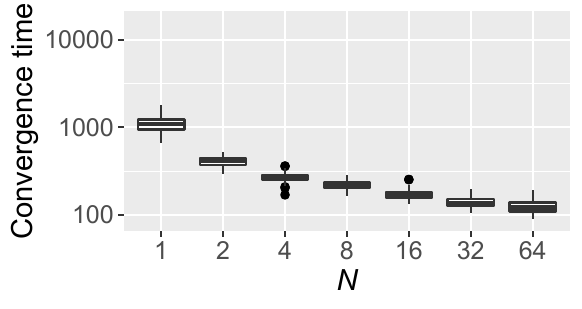} &  \hspace{-5mm} \includegraphics[width=0.34\textwidth]{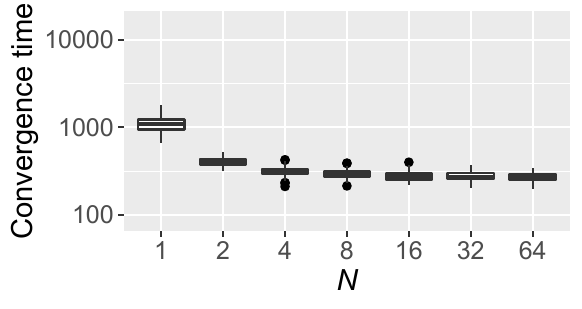} \cr
  \vspace{2mm} \hspace{-2mm} \textbf{(a) MTM w. $w(\mathbf{x}, \mathbf{y}) = \frac{\pi(\mathbf{y})}{\pi(\mathbf{x})}$} & \hspace{-5mm} \textbf{(b) MTM w. $w(\mathbf{x}, \mathbf{y}) = \sqrt{\frac{\pi(\mathbf{y})}{\pi(\mathbf{x})}}$} & \hspace{-6mm} \textbf{(c) MTM w. $w(\mathbf{x}, \mathbf{y}) = \frac{\pi(\mathbf{y}) / \pi(\mathbf{x})}{1 + \pi(\mathbf{y}) / \pi(\mathbf{x})}$} \cr
     \multicolumn{3}{c}{\large\textbf{Target: Laplace distribution}} \cr
 \vspace{-2mm}\hspace{-2mm}\includegraphics[width=0.34\textwidth]{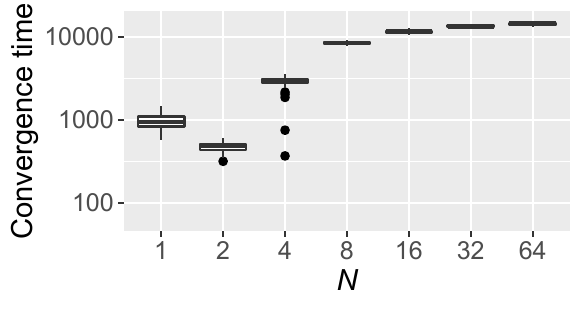} &  \hspace{-5mm} \includegraphics[width=0.34\textwidth]{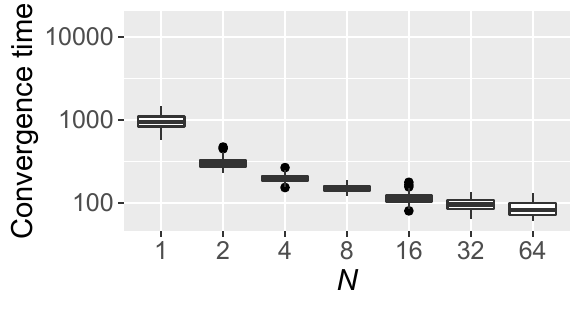} &  \hspace{-5mm} \includegraphics[width=0.34\textwidth]{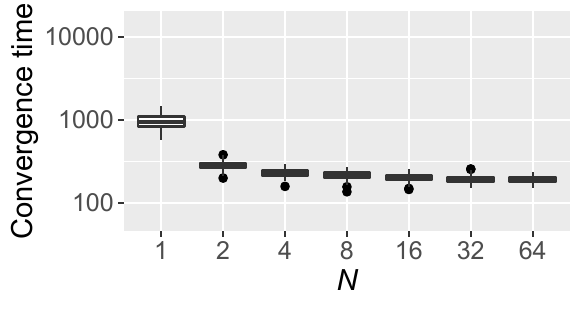} \cr
   \hspace{-2mm} \textbf{(a) MTM w. $w(\mathbf{x}, \mathbf{y}) = \frac{\pi(\mathbf{y})}{\pi(\mathbf{x})}$} & \hspace{-5mm} \textbf{(b) MTM w. $w(\mathbf{x}, \mathbf{y}) = \sqrt{\frac{\pi(\mathbf{y})}{\pi(\mathbf{x})}}$} & \hspace{-6mm} \textbf{(c) MTM w. $w(\mathbf{x}, \mathbf{y}) = \frac{\pi(\mathbf{y}) / \pi(\mathbf{x})}{1 + \pi(\mathbf{y}) / \pi(\mathbf{x})}$} \cr
  \end{array}$\vspace{-2mm}
  \caption{\small Convergence time as a function of $N$ when $d = 50$, the scale parameter is adaptively tuned and the initial state is $\mathbf{x} = (10, \ldots, 10)$, for several values of $N$ and: (a) MTM with $w(\mathbf{x}, \mathbf{y}) = \pi(\mathbf{y}) / \pi(\mathbf{x})$, (b) MTM with $w(\mathbf{x}, \mathbf{y}) = \sqrt{\pi(\mathbf{y}) / \pi(\mathbf{x})}$, and (c) MTM with $w(\mathbf{x}, \mathbf{y}) = (\pi(\mathbf{y}) / \pi(\mathbf{x}))/(1+ \pi(\mathbf{y}) / \pi(\mathbf{x}))$; on the first row, the target is a 50-dimensional standard normal distribution, whereas it is a 50-dimensional standard Laplace distribution on the second row.}\label{fig:Boxplots}
 \end{figure}
\normalsize

Based on all the observations made in this section, we provide a recommendation for an adaptive implementation of LB MTM.
\begin{enumerate}

    \item Set $g(x) = \sqrt{x}$.

    \item Set $N$ equal to the number of available cores for parallel computing.

    \item Initialize the step size as $\sigma = \ell / \sqrt{d}$ with $\ell = 2.38$.

    \item Run MTM with adaptive tuning of $\sigma$ as in Algorithm 4 of \cite{andrieu2008tutorial}, with $\gamma(m) = m^{-0.6}$ and a targeted acceptance rate of $50\%$.

\end{enumerate}

\section{Performance at stationarity}\label{sec:scaling_convergence}

In this section, we characterize the high-dimensional behaviour of MTM algorithms after they have reached stationarity. To this end, we first establish in \autoref{sec:scaling-limits} the weak convergence of a transformation of the Markov chains produced by ideal MH schemes towards Langevin diffusions, each started in stationarity, as $d \rightarrow \infty$. In \autoref{sec:convergence_MTM}, we bridge the gap between the MTM algorithms and diffusion processes by providing conditions about the scaling of $N$ with $d$ ensuring the asymptotic equivalence of MTM and ideal schemes. We finish with numerical experiments in \autoref{sec:numerical} that corroborate the findings of Sections \ref{sec:scaling-limits} and \ref{sec:convergence_MTM}.

It is worth mentioning that a scaling-limit analysis of MTM was conducted in \cite{bedard2012scaling}, but the analysis in that paper is quite different from that conducted here. First, in \cite{bedard2012scaling}, MTM is not seen as an approximation to an ideal scheme because $N$ is considered fixed; a weak convergence of a transformation of the Markov chains produced by MTM towards Langevin diffusions is directly obtained as $d \rightarrow \infty$. The asymptotic regime considered here is that where the number of candidates $N$ increases with $d$, whereas the asymptotic regime in \cite{bedard2012scaling} can be thought of as the situation where $N$ is small relatively to $d$. Another difference is that \cite{bedard2012scaling} considers only GB MTM, while we study also LB MTM.

Because of the nature of the analysis conducted in Sections \ref{sec:scaling-limits} and \ref{sec:convergence_MTM}, we, as for the asymptotic analysis in \autoref{sec:tails}, highlight a dependency on $d$ of the target distribution, the scale parameter, the number of candidates, and so on, by denoting them $\pi_d$, $\sigma_d$, $N_d$, etc.

\subsection{Scaling limits of ideal schemes}\label{sec:scaling-limits}

For the analysis, we consider the same scenario as in \autoref{sec:tails}; in particular, we consider that the target distribution is defined as in \eqref{eqn:target_product}. Also, for the analysis, we set $\sigma_d = \ell / d^\tau$ in $q_{\sigma_d}(\bx_d, \cdot \,) = \mathcal{N}(\mathbf{x}_d, \sigma_d^2 \I_d)$, with $\ell$ being a positive tuning parameter and $\tau$ a positive constant characterizing the scalability of the algorithm with respect to the dimension (the smaller is $\tau$, the better is the scalability with respect to $d$).

Before presenting the scaling-limit result, we introduce required notation. We use $\Phi$ to denote the cumulative distribution function of the standard normal distribution. We use $\{\bX_{d, \text{ideal}}(m): m \in \N\}$ to denote a Markov chain simulated by an ideal MH scheme using $Q_{w, \sigma_d}(\bx_d, \cdot \,)$ for proposal distribution, and define a re-scaled continuous-time version $\{\bZ_{d, \text{ideal}}(t): t \geq 0\}$ using:
\begin{align}\label{eqn:cont_proc_ideal}
 \bZ_{d, \text{ideal}}(t) := \bX_{d, \text{ideal}}(\lfloor d^{2\tau} t \rfloor),
\end{align}
with $\lfloor \, \cdot \, \rfloor$ being the floor function. A scaling limit consists in proving that the first component of $\{\bZ_{d, \text{ideal}}(t): t \geq 0\}$, denoted by $\{Z_{d, \text{ideal}}(t): t\geq0\}$, converges weakly to $\{Z(t): t\geq0\}$, a Langevin diffusion.

We are now ready to present the scaling-limit result. It is about the GB ideal scheme and the LB ideal scheme with $w(\mathbf{x}_d, \mathbf{y}_d) = \sqrt{\pi_d(\mathbf{y}_d) / \pi_d(\mathbf{x}_d)}$.

\begin{Theorem}\label{theorem:scaling_limit_ideal}
 Assume that $\pi_d$ is as in \eqref{eqn:target_product} and that the proposal distribution in a MH algorithm is $Q_{w, \sigma_d}(\bx_d, \cdot \,)$ with $q_{\sigma_d}(\bx_d, \cdot \,) = \mathcal{N}(\mathbf{x}_d, \sigma_d^2 \I_d)$ and $\sigma_d = \ell / d^\tau$. Assume also that $\bX_{d, \text{ideal}}(0) \sim \pi_d$ and that $\bZ_{d, \text{ideal}}(t)$ for $t \geq 0$ is defined as in \eqref{eqn:cont_proc_ideal}. Then, as $d \rightarrow \infty$, $\{Z_{d, \text{ideal}}(t): t\geq0\}$ converges weakly towards $\{Z(t): t\geq0\}$, a Langevin diffusion such that $Z(0) \sim \mathcal{N}(0, 1)$ and
 \[
  \d Z(t) = \ell^2 (\vartheta_{w, \tau}(\ell) / 2) (\log \varphi(Z(t)))' \, \d t + \sqrt{\ell^2 \vartheta_{w, \tau}(\ell)} \, \d B(t),
 \]
 with $\{B(t): t \geq 0\}$ being a standard Brownian motion and  $\vartheta_{w, \tau}$ being defined as follows:
 if $w(\mathbf{x}_d, \mathbf{y}_d) = \sqrt{\pi_d(\mathbf{y}_d) / \pi_d(\mathbf{x}_d)}$,
 \[
  \vartheta_{w, \tau}(\ell) = \begin{cases}
   2\Phi(-\ell^3/2^3) &\text{if} \quad \tau = 1/6, \cr
   1 & \text{if} \quad \tau > 1/6;
  \end{cases}
 \]
 if $w(\mathbf{x}_d, \mathbf{y}_d) = \pi_d(\mathbf{y}_d) / \pi_d(\mathbf{x}_d)$,
 \[
  \vartheta_{w, \tau}(\ell) = \begin{cases}
   2\Phi(-\ell/2) & \text{if} \quad \tau = 1/2, \cr
   1 & \text{if} \quad \tau > 1/2.
  \end{cases}
 \]
\end{Theorem}

To establish such a result, it is crucial that the expected acceptance probability (in stationarity)  $\E[\alpha_{\text{ideal}}(\bX_d, \bY_d)]$ (with $\bX_d \sim \pi_d$) converges towards a non-null function of $\ell$ that is independent of $d$; the function $\vartheta_{w, \tau}(\ell)$ in \autoref{theorem:scaling_limit_ideal} is precisely this function. \autoref{theorem:scaling_limit_ideal} thus indicates that $\tau \geq 1/2$ in the GB ideal scheme allows such a convergence, whereas $\tau \geq 1/6$ in the LB ideal scheme is sufficient. This implies that the LB ideal scheme has a better scaling with the dimension than the GB ideal scheme.

We present numerical results in \autoref{fig:acc_rates_d} of $\E[\alpha_{\text{ideal}}(\bX_d, \bY_d)]$ as a function of $d$. These numerical results allow to show that $\E[\alpha_{\text{ideal}}(\bX_d, \bY_d)]$ converges to 0 when $w(\mathbf{x}_d, \mathbf{y}_d) = \sqrt{\pi_d(\mathbf{y}_d) / \pi_d(\mathbf{x}_d)}$ and $\tau < 1/6$ and when $w(\mathbf{x}_d, \mathbf{y}_d) = \pi_d(\mathbf{y}_d) / \pi_d(\mathbf{x}_d)$ and $\tau < 1/2$. The expectations are approximated using independent Monte Carlo sampling; the approximations are based on samples of size 1,000,000. We stress that there is an important difference between $\E[\alpha_{\text{ideal}}(\bX_d, \bY_d)]$ and what we called the \textit{conditional expected acceptance probability} in \autoref{sec:tails}: $\E[\alpha_{\text{ideal}}(\bX_d, \bY_d)]$ is the unconditional expectation with $\bX_d \sim \pi_d$, whereas the \textit{conditional expected acceptance probability} in \autoref{sec:tails} is the conditional expectation given a current state $\bx_d$. We highlight the difference by referring to $\E[\alpha_{\text{ideal}}(\bX_d, \bY_d)]$ as (simply) the \textit{expected acceptance probability}.

 \begin{figure}[ht]
  \centering\footnotesize
  $\begin{array}{cc}
 \vspace{-2mm}\hspace{-2mm}\includegraphics[width=0.50\textwidth]{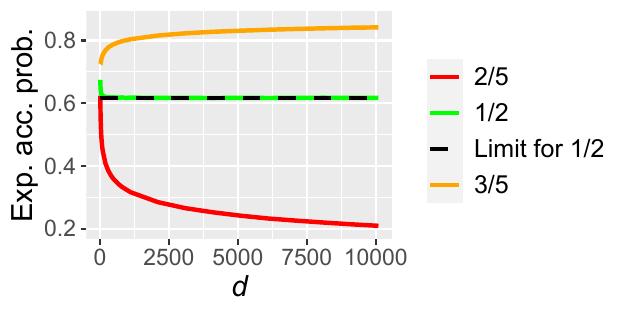} &  \hspace{-3mm} \includegraphics[width=0.50\textwidth]{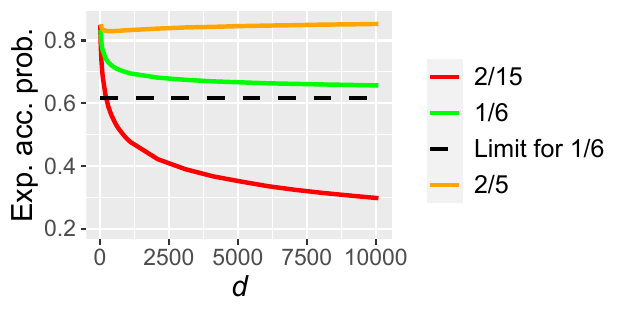} \cr
   \hspace{-4mm} \textbf{(a) Ideal scheme w. $w(\mathbf{x}_d, \mathbf{y}_d) = \pi_d(\mathbf{y}_d) / \pi_d(\mathbf{x}_d)$} & \hspace{-5mm} \textbf{(b) Ideal scheme w. $w(\mathbf{x}_d, \mathbf{y}_d) = \sqrt{\pi_d(\mathbf{y}_d) / \pi_d(\mathbf{x}_d)}$}
  \end{array}$\vspace{-0mm}
  \caption{\small Expected acceptance probabilities as a function of $d$ for: (a) the ideal scheme with $w(\mathbf{x}_d, \mathbf{y}_d) = \pi_d(\mathbf{y}_d) / \pi_d(\mathbf{x}_d)$, $\ell = 1$, and $\tau = 2/5$, $\tau = 1/2$ and $\tau = 3/5$; (b) the ideal scheme with $w(\mathbf{x}_d, \mathbf{y}_d) = \sqrt{\pi_d(\mathbf{y}_d) / \pi_d(\mathbf{x}_d)}$, $\ell = 2^{2/3}$, and $\tau = 2/15$, $\tau = 1/6$ and $\tau = 2/5$; the limiting expected acceptance probabilities for $\tau = 1/2$ and $\tau = 1/6$ are also presented, in the GB and LB cases, respectively; with the values used for $\ell$, the limits are the same; the values for $\tau$ other than $1/2$ and $1/6$ have been obtained by increasing and decreasing these by $20\%$.}\label{fig:acc_rates_d}
 \end{figure}
\normalsize

We finish this section with a discussion about the tuning of the parameter $\ell$. For this part, we analyse another characteristic of the limiting stochastic process in \autoref{theorem:scaling_limit_ideal}. This stochastic process can be seen as a function of another process:
\begin{align}\label{eqn:langevin}
 Z(t) = V(\ell^2 \vartheta_{w, \tau}(\ell) t),
 \end{align}
 where $\{V(t): t \geq 0\}$ is the Langevin diffusion with a stochastic differential equation given by
\[
 \d V(t)= (\log \varphi(V(t)))' / 2 \times \d t + \d B(t).
\]
The term $\ell^2 \vartheta_{w, \tau}(\ell)$ in \eqref{eqn:langevin} is sometimes referred to as the \textit{speed measure} of $\{Z(t): t\geq0\}$. From a MCMC perspective, the largest speed is best. Indeed, the stationary integrated autocorrelation time of \emph{any} function $h$ of the diffusion is proportional to the inverse of the diffusion speed. The following corollary presents the largest speeds and tuning procedures.

\begin{Corollary}\label{cor:optimization}
 The speed measure when $w(\mathbf{x}_d, \mathbf{y}_d) = \sqrt{\pi_d(\mathbf{y}_d) / \pi_d(\mathbf{x}_d)}$ and $\tau = 1/6$, given by $2\ell^2\Phi(-\ell^3/2^3)$, is maximized at $\ell^* = 1.650$ (to three decimal places), which yields a limiting expected acceptance probability of $2\Phi(-(\ell^*)^3/2^3) = 0.574$ (to three decimal places). The speed measure when $w(\mathbf{x}_d, \mathbf{y}_d) = \pi_d(\mathbf{y}_d) / \pi_d(\mathbf{x}_d)$ and $\tau = 1/2$, given by $2\ell^2\Phi(-\ell/2)$, is maximized at $\ell^{**} = 2.381$ (to three decimal places), which yields a limiting expected acceptance probability of $2\Phi(-\ell^{**}/2) = 0.234$ (to three decimal places).
\end{Corollary}

This result highlights that: (i) the (asymptotically) optimal expected acceptance probability for the GB ideal scheme is the same as that for random-walk Metropolis, with the same maximum speed for the limiting diffusion \citep{roberts1997weak}; (ii) the (asymptotically) optimal expected acceptance probability for the LB ideal scheme is the same as that for MALA, with the same maximum speed for the limiting diffusion \citep{roberts1998optimal}. The latter is expected as MALA can be viewed as an approximation to the ideal scheme with $w(\mathbf{x}_d, \mathbf{y}_d) = \sqrt{\pi_d(\mathbf{y}_d) / \pi_d(\mathbf{x}_d)}$ (recall the discussion in \autoref{sec:connection_gradient}), and the approximation is asymptotically exact if the step size in MALA diminishes adequately with $d$, as $d \rightarrow \infty$. The results in \autoref{cor:optimization} can be obtained from \cite{roberts1997weak} and \cite{roberts1998optimal}.

Presenting the largest speed and a tuning procedure for the ideal scheme with $w(\mathbf{x}_d, \mathbf{y}_d) = \sqrt{\frac{\pi_d(\mathbf{y}_d)}{\pi_d(\mathbf{x}_d)}}$ and $\tau = 1/6$ in \autoref{cor:optimization} is interesting as it allows to make that connection with MALA, but we will see in \autoref{sec:convergence_MTM} that in order to take advantage of such a value for $\tau$ in MTM, one would need to use computational resource well beyond what is reasonable and realistic. We will see that, in MTM, it is more reasonable to use values around $\tau = 1/2$. With $\tau = 1/2$, the limiting diffusion of the ideal scheme with $w(\mathbf{x}_d, \mathbf{y}_d) = \sqrt{\pi_d(\mathbf{y}_d) / \pi_d(\mathbf{x}_d)}$ has a speed measure given by $\ell^2$, which can be compared with that of the limiting diffusion of the ideal scheme with $w(\mathbf{x}_d, \mathbf{y}_d) = \pi_d(\mathbf{y}_d) / \pi_d(\mathbf{x}_d)$ in \autoref{cor:optimization}, given by $2\ell^2\Phi(-\ell/2)$, because both schemes use the same form for the scale parameter; see \autoref{fig:speed} for a comparison of the speed measures.

\begin{figure}[ht]
  \centering
  \includegraphics[width=0.40\textwidth]{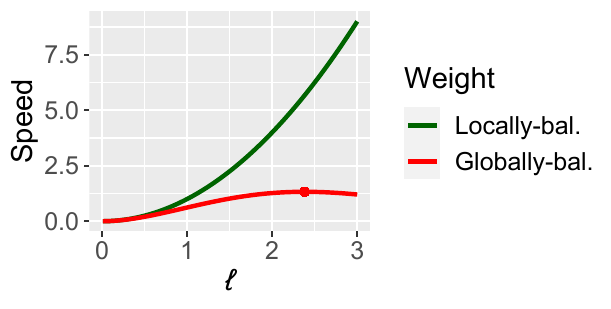}
  \vspace{-3mm}
  \caption{\small Speed measures as a function of $\ell$ of the limiting diffusions of the GB and LB ideal schemes when $\tau = 1/2$; the red point indicates the maximum speed when using the GB weight function.}\label{fig:speed}
 \end{figure}
\normalsize

\autoref{fig:speed} suggests that MTM using $w(\mathbf{x}_d, \mathbf{y}_d) = \sqrt{\pi_d(\mathbf{y}_d) / \pi_d(\mathbf{x}_d)}$ is at least as good as MTM using $w(\mathbf{x}_d, \mathbf{y}_d) = \pi_d(\mathbf{y}_d) / \pi_d(\mathbf{x}_d)$, in high dimensions if it approximates well the ideal scheme; this is observed empirically in \autoref{sec:numerical}. \autoref{fig:speed} also suggests to set $\ell$ in MTM with $w(\mathbf{x}_d, \mathbf{y}_d) = \sqrt{\pi_d(\mathbf{y}_d) / \pi_d(\mathbf{x}_d)}$ as large as possible. However, in practice when sampling from target distributions having high but fixed dimensions $d$, $\ell$ has to be constrained to small values compared to $d$ to reflect that it is held constant and thus does not grow with $d$ in our asymptotic analysis. In our numerical experiments in \autoref{sec:numerical}, we observed that in moderate to high dimensions, with moderate to large values for $N_d$, optimally tuned MTM using $w(\mathbf{x}_d, \mathbf{y}_d) = g(\pi_d(\mathbf{y}_d) / \pi_d(\mathbf{x}_d))$ with $g(x) = \sqrt{x}$ and $g(x) = x / (1 + x)$ have acceptance rates in a range of $50\%$ to $60\%$. We thus recommend to users to start their tuning procedures with a value of $\ell$ yielding an acceptance rate in that range (if they do not tune $\ell$ adaptively as in \autoref{sec:conv_numerics}). Note that the algorithms in \autoref{sec:numerical} were tuned using expected squared jumping distance (ESJD).

\subsection{Characterizing the approximation of ideal schemes by MTM}\label{sec:convergence_MTM}

In this section, we provide conditions on $N_d$ under which MTM with $w(\mathbf{x}_d, \mathbf{y}_d) = \pi_d(\mathbf{y}_d) / \pi_d(\mathbf{x}_d)$ and MTM with $w(\mathbf{x}_d, \mathbf{y}_d) = \sqrt{\pi_d(\mathbf{y}_d) / \pi_d(\mathbf{x}_d)}$ are asymptotically equivalent to their ideal counterparts as $d \rightarrow \infty$. The sense in which they are asymptotically equivalent implies a convergence of transformations of the Markov chains simulated by these MTM algorithms towards diffusions. We observed in the proof of \autoref{theorem:weak_con_MTM_d_fixed} that, for an MTM algorithm to be asymptotically equivalent to its ideal counterpart, it is sufficient that: (i) the weight normalization in the proposal distribution of $\bY_J$ in MTM be asymptotically equivalent to the normalizing constant of $Q_{w, \sigma_d}$ (recall \autoref{prop:dist_MTM} and \eqref{eq:expectation_ideal}), and (ii) the acceptance probability in MTM be asymptotically equivalent to that in the ideal scheme. We present below a result stating conditions on $N_d$ under which this holds. The result is established under the same scenario as in \autoref{sec:scaling-limits} (the target distribution is defined as in \eqref{eqn:target_product} and $q_{\sigma_d}(\bx_d, \cdot \,) = \mathcal{N}(\mathbf{x}_d, \sigma_d^2 \I_d)$ with $\sigma_d = \ell / d^\tau$).

Before presenting the result, we introduce required notation. We use $\{\bX_{d, \text{MTM}}(m): m \in \N\}$ to denote a Markov chain simulated by a MTM algorithm, and define a re-scaled continuous-time version $\{\bZ_{d, \text{MTM}}(t): t \geq 0\}$ using:
\begin{align}\label{eqn:cont_proc_ideal}
 \bZ_{d, \text{MTM}}(t) := \bX_{d, \text{MTM}}(\lfloor d^{2\tau} t \rfloor).
\end{align}
We use $\{Z_{d, \text{MTM}}(t): t \geq 0\}$ to denote the first component of $\{\bZ_{d, \text{MTM}}(t): t \geq 0\}$.

We are now ready to present the result.  It is about GB MTM and LB MTM with $w(\mathbf{x}_d, \mathbf{y}_d) = \sqrt{\pi_d(\mathbf{y}_d) / \pi_d(\mathbf{x}_d)}$.

\begin{Theorem}\label{theorem:MTM_app_ideal}
 Assume that $\pi_d$ is as in \eqref{eqn:target_product} and that $q_{\sigma_d}(\bx_d, \cdot \,) = \mathcal{N}(\mathbf{x}_d, \sigma_d^2 \I_d)$ with $\sigma_d = \ell / d^\tau$.  Let $\mathbf{Y}_1, \ldots, \mathbf{Y}_{N_d}$ be $N_d$ conditionally independent random variables given $\mathbf{X}_d$, each distributed as $\mathbf{Y}_i \mid \mathbf{X}_d \sim q_{\sigma_d}(\mathbf{X}_d, \cdot \,)$ with $\mathbf{X}_d \sim \pi_d$. Let $\mathbf{Y}_{J}$ be a proposal sampled using MTM. Then, there exists a positive integer $d_0$ such that for any $d \geq d_0$,
 \begin{align*}
  \E\left[d^{2\tau}\left|\frac{w(\mathbf{X}_d, \mathbf{Y}_1)}{\frac{1}{N_d}\sum_{i = 1}^{N_d} w(\mathbf{X}_d, \mathbf{Y}_i)} -  \frac{w(\mathbf{X}_d, \mathbf{Y}_1)}{\E[w(\mathbf{X}_d, \mathbf{Y}_1) \mid \mathbf{X}_d]}\right|\right] \leq \frac{d^{2\tau}}{N_d^{1/2}} \, \varrho_1(d),
 \end{align*}
 and
  \begin{align*}
  \E[d^{2\tau}|\alpha(\mathbf{X}_d, \mathbf{Y}_J) - \alpha_{\text{ideal}}(\mathbf{X}_d, \mathbf{Y}_J)|] \leq \frac{d^{2\tau}}{N_d^{1/2}} \, \varrho_2(d),
 \end{align*}
 whenever $w(\mathbf{x}_d, \mathbf{y}_d) = \sqrt{\pi_d(\mathbf{y}_d) / \pi_d(\mathbf{x}_d)}$ or $w(\mathbf{x}_d, \mathbf{y}_d) = \pi_d(\mathbf{y}_d) / \pi_d(\mathbf{x}_d)$, $\varrho_1$ and $\varrho_2$ being functions of $d$ that are explicitly defined in the proof. If $\tau \geq 1/2$, $\varrho_1(d)$ and $\varrho_2(d)$ converge to positive constants as $d \rightarrow \infty$; $N_d = d^{4\tau(1 + \rho)}$ with $\rho$ being any positive constant makes the expectations converge to 0. If $\tau < 1/2$, $\varrho_1(d)$ and $\varrho_2(d)$ grows with $d$ and $N_d = (1+\nu)^d$ with $\nu$ being any positive constant makes the expectations converge to 0. When these expectations converge to 0, $\{Z_{d, \text{MTM}}(t): t \geq 0\}$ converges weakly towards the same Langevin diffusion $\{Z(t): t \geq 0\}$ as in \autoref{theorem:scaling_limit_ideal}, under the same conditions.
 \end{Theorem}

 \autoref{theorem:MTM_app_ideal} indicates that, if one wanted to use a larger step size with $\tau = 1/6$,  a sufficient condition for MTM with $w(\mathbf{x}_d, \mathbf{y}_d) = \sqrt{\pi_d(\mathbf{y}_d) / \pi_d(\mathbf{x}_d)}$ to approximates well its ideal counterpart is to scale $N_d$ exponentially with $d$.
However, using such an enormous number of candidates is computationally prohibitive.
When $\tau=1/2$, \autoref{theorem:MTM_app_ideal} indicates that scaling $N_d$ quadratically (essentially) with $d$ is sufficient for both MTM
 with $w(\mathbf{x}_d, \mathbf{y}_d) = \sqrt{\pi_d(\mathbf{y}_d) / \pi_d(\mathbf{x}_d)}$ and MTM with $w(\mathbf{x}_d, \mathbf{y}_d) = \pi_d(\mathbf{y}_d) / \pi_d(\mathbf{x}_d)$ to approximate well their ideal counterparts.
We note that these are only sufficient conditions, rather than necessary ones, and that the implied scalings may not be tight. In particular, taking $N_d$ less than quadratic in $d$ may be enough when $\tau=1/2$ and in general more research is needed to establish optimal ways of scaling $N_d$ with $d$.
Nonetheless, the result suggests that, in high dimensions, LB MTM schemes will struggle to approximate their ideal counterparts when an aggressive step size with $\tau = 1/6$ is used.
The numerical results in \autoref{sec:numerical} complement this analysis: they show that in moderately high dimensions and with a moderately large number of candidates, optimally tuned MTM with $w(\mathbf{x}_d, \mathbf{y}_d) = \sqrt{\pi_d(\mathbf{y}_d) / \pi_d(\mathbf{x}_d)}$ and optimally tuned MTM with $w(\mathbf{x}_d, \mathbf{y}_d) = \pi_d(\mathbf{y}_d) / \pi_d(\mathbf{x}_d)$ have similar performance beyond the burn-in period, showing that MTM with $w(\mathbf{x}_d, \mathbf{y}_d) = \sqrt{\pi_d(\mathbf{y}_d) / \pi_d(\mathbf{x}_d)}$ is not able to take advantage of a larger step size because with such a step size it does not approximate well the ideal scheme and does not have a good performance.
This is in contrast with the results of \autoref{sec:convergence} (e.g., Figures \ref{fig:Traceplots} and \ref{fig:Boxplots}) where it is shown that even $N_d \ll d$ is sufficient to yield a significant improvement in terms of convergence speed.

\subsection{Numerical experiments}\label{sec:numerical}

In this section, we evaluate the empirical performance in stationarity of GB MTM and LB MTM under an non-asymptotic framework; the mathematical objects like the target distribution, the scale parameter and the number of candidates are thus denoted without a subscript $d$, that is $\pi$, $\sigma$ and $N$. As previously, two LB MTM algorithms are evaluated: both use a weight function given by $w(\mathbf{x}, \mathbf{y}) = g(\pi(\mathbf{y}) / \pi(\mathbf{x}))$; one uses $g(x) = \sqrt{x}$, and the other one, $g(x) = x / (1 + x)$. The performance is evaluated using the Monte Carlo estimate of ESJD in stationarity, the latter being defined as
\[
 \text{ESJD} := \E\left[\|\bX_{\text{MTM}}(m + 1) - \bX_{\text{MTM}}(m)\|^2\right] = \E\left[\|\bY_J - \bX\|^2 \alpha(\bX, \bY_J)\right],
\]
where $\bX \sim \pi$ and $\bY_J$ is sampled using the MTM mechanism. The performance evaluation is conducted under the same scenario as previously: the target distribution is defined as in \eqref{eqn:target_product}, $d = 50$, and $q_{\sigma}(\bx, \cdot \,) = \mathcal{N}(\mathbf{x}, \sigma^2 \I_d)$ with $\sigma = \ell / \sqrt{d}$. The results have been observed to be similar when the target is still a normal distribution but with components having different marginal variances and in higher dimensions. Note that in the scenario considered here, ESJD can be approximated efficiently using independent Monte Carlo sampling. The approximations are based on samples of size 100,000.

In \autoref{fig:ESJD_N_ell_fixed}, we present results of ESJD as a function of $N$, for $\ell$ fixed, while in \autoref{fig:ESJD_ell}, results of ESJD as a function of $\ell$ are presented, for several values of $N$. The value of $\ell$ used in \autoref{fig:ESJD_N_ell_fixed} corresponds to that which maximizes ESJD when $N = 5$; the same value of $\ell = 3.20$ is optimal for all algorithms (at least according to our grid search). The results in Figures \ref{fig:ESJD_N_ell_fixed} and \ref{fig:ESJD_ell} unveil that when $N$ is small, all algorithms are essentially equivalent in stationarity, but they also unveil another problem with GB MTM. The algorithm behaviour changes drastically as $N$ increases, to the extend that a performance reduction may be observed while keeping $\ell$ fixed. This is counter-intuitive and may be confusing to users that may diminish the value of $\ell$ to compensate, while the opposite is desirable. That pathological behaviour is not exhibited by LB MTM algorithms, which behave as one would expect, with a performance that increases monotonically with $N$, while keeping $\ell$ fixed. The results in \autoref{fig:ESJD_ell} also allow to notice that if a user manages to optimally tune GB MTM, then it is not significantly outperformed in stationarity by LB MTM for reasonable values of $N$, as mentioned previously.

\begin{figure}[ht]
  \centering
  \includegraphics[width=0.45\textwidth]{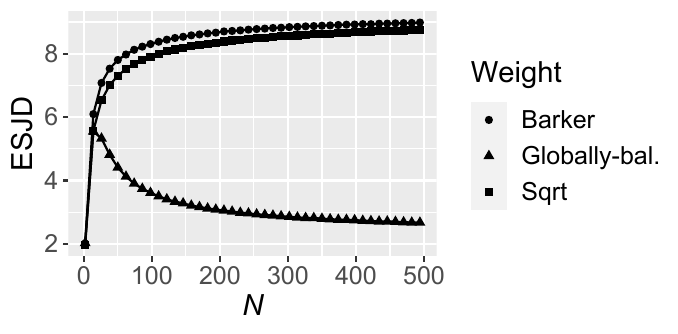}
  \vspace{-3mm}
  \caption{\small ESJD as a function of $N$ when $d = 50$, $\ell = 3.20$ for MTM with $w(\mathbf{x}, \mathbf{y}) = \pi(\mathbf{y}) / \pi(\mathbf{x})$, MTM with $w(\mathbf{x}, \mathbf{y}) = \sqrt{\pi(\mathbf{y}) / \pi(\mathbf{x})}$, and MTM with $w(\mathbf{x}, \mathbf{y}) = (\pi(\mathbf{y}) / \pi(\mathbf{x}))/(1+ \pi(\mathbf{y}) / \pi(\mathbf{x}))$.}\label{fig:ESJD_N_ell_fixed}
 \end{figure}
\normalsize

 \begin{figure}[ht]
  \centering\footnotesize
  $\begin{array}{ccc}
 \vspace{-2mm}\hspace{-2mm}\includegraphics[width=0.34\textwidth]{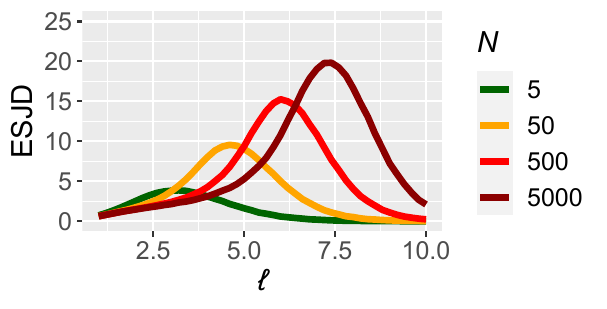} &  \hspace{-5mm} \includegraphics[width=0.34\textwidth]{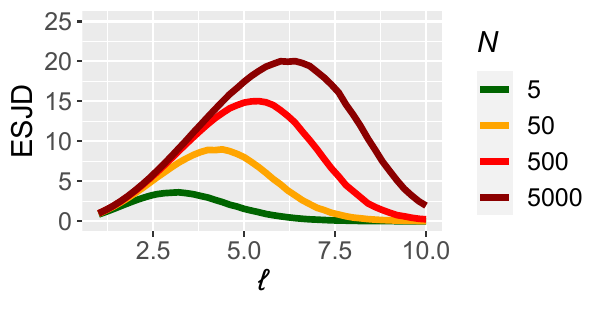} &  \hspace{-5mm} \includegraphics[width=0.34\textwidth]{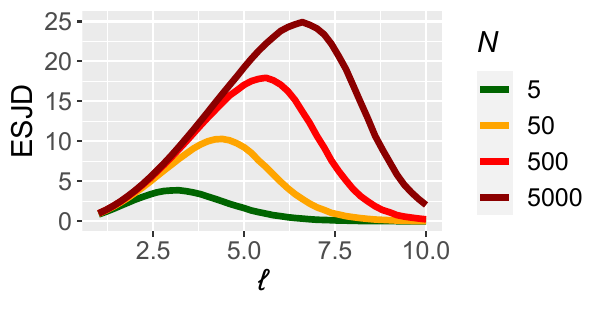} \cr
  \hspace{-2mm} \textbf{(a) MTM w. $w(\mathbf{x}, \mathbf{y}) = \frac{\pi(\mathbf{y})}{\pi(\mathbf{x})}$} & \hspace{-7mm} \textbf{(b) MTM w. $w(\mathbf{x}, \mathbf{y}) = \sqrt{\frac{\pi(\mathbf{y})}{\pi(\mathbf{x})}}$} & \hspace{-6mm} \textbf{(c) MTM w. $w(\mathbf{x}, \mathbf{y}) = \frac{\pi(\mathbf{y}) / \pi(\mathbf{x})}{1 + \pi(\mathbf{y}) / \pi(\mathbf{x})}$} \cr
  \end{array}$\vspace{-2mm}
  \caption{\small ESJD as a function of $\ell$ when $d = 50$, for several values of $N$ and: (a) MTM with $w(\mathbf{x}, \mathbf{y}) = \pi(\mathbf{y}) / \pi(\mathbf{x})$, (b) MTM with $w(\mathbf{x}, \mathbf{y}) = \sqrt{\pi(\mathbf{y}) / \pi(\mathbf{x})}$, and (c) MTM with $w(\mathbf{x}, \mathbf{y}) = (\pi(\mathbf{y}) / \pi(\mathbf{x}))/(1+ \pi(\mathbf{y}) / \pi(\mathbf{x}))$.}\label{fig:ESJD_ell}
 \end{figure}
\normalsize

\section{Application of MTM for Bayesian inference in immunotherapy}\label{sec:real_example}

We study in this section the application of MTM as an inference solution in a real-world context of immunotherapy in precision medicine involving a model with a likelihood function that is expensive to evaluate. The context is described in \autoref{sec:context_real_example}, the data and the model used to analyse them are presented in \autoref{sec:data_model_real_example}, and the application of MTM algorithms is studied in \autoref{sec:inference_real_example}. Our study suggests that the scope of the results (empirical and theoretical) of the previous sections about GB MTM and LB MTM extends beyond the simple contexts in which they were derived. Note that in this section, we adopt a notation which is consistent with typical Bayesian-statistics contexts; it is thus different from that adopted in the other sections.

\subsection{Context}\label{sec:context_real_example}

Precision medicine is an innovative approach of care whereby patients are subject to personalized treatment strategies that take into account their personal data (genetics, lifestyle, health history, etc.). In \cite{jenner2021silico}, the authors study the effect of such personalized treatments on advanced-stage cancer patients based on cancer vaccines and oncolytic immunotherapy. Oncolytic viruses such as vesicular stomatitis virus (VSV) and vaccinia virus (VV)  have the potential to destroy tumor cells and to induce a systemic anti-tumour immune response and, as such, can lead to what is commonly referred to as \textit{virotherapy}. At the moment, several related open questions form a very active strand of research. For instance, one would like to assess the potential benefits offered by combining different oncolytic viruses or the use of virus enhancers. Even though it is not restricted to cancer treatments, being able to find the optimal treatment schedule for a given patient is a central question in precision medicine. It is believed that, provided that these questions can be addressed, oncolytic virotherapy will form a major breakthrough in cancer treatment. Statistical modelling and Bayesian inference represent a way to help answering those questions. The reliability and efficiency of the numerical methods leading to the inference is thus of crucial importance. In \autoref{sec:inference_real_example}, we evaluate the performance of different MTM algorithms. 

\subsection{Data and model}\label{sec:data_model_real_example}

In practice, data are collected from a cohort of $K$ advanced-stage cancer patients that are examined at $T$ time points. At each time point $t\in\{1,\ldots,T\}$, the state of patient $k\in\{1,\ldots,K\}$ is summarized through statistics $\by_k(t) \in \R^{m_1}$, with $m_1\in\N$. These statistics represent the variables of interest. Covariate data points $\bx_k \in \R^{m_2}$ with $m_2\in\N$, independent of time but associated to a patient, are also collected. Each patient is assigned a personalized treatment schedule $r_k$ which is considered as a data point from a categorical variable. The covariate data points and the personalized treatment schedules are considered to be fixed and known, in contrast to $\by_k(t)$ which is assumed in a statistical model to be a realization of a random variable. The assumed statistical model is a \textit{forward model}:
\begin{equation}\label{eqn:model}
\bY_{k}(t)=\hat{\by}(t, \btheta, \bx_k, r_k) + \sigma \, \bvarepsilon_{k}(t)\,,
\end{equation}
 where $\hat{\by}(t, \btheta, \bx_k, r_k)$ is the output of a dynamical system proposed in \cite{jenner2021silico} and described below, $\sigma > 0$ is a scale parameter, and $\bvarepsilon_{1}(1), \ldots, \bvarepsilon_{1}(T), \ldots, \bvarepsilon_{K}(1), \ldots, \bvarepsilon_{K}(T)  \in \R^{m_1}$ are random standardized errors. The scale parameter is considered here fixed and known to simplify. The unknown parameter is $\btheta \in \R^d$ with $d = 14$. We assume that $\bvarepsilon_{k}(t) \sim  \mathcal{N}(\mathbf{0}, \I_{m_1})$, $t = 1, \ldots, T$ and $k = 1, \ldots, K$, are IID random variables independent of $\btheta$, which is a common assumption in the literature.

 The output $\hat{\by}(t, \btheta, \bx_k, r_k)$ is produced given $(t, \btheta, \bx_k, r_k)$ from a discretized version of the numerical solution of a system of time-delay differential equations (DDE):
\begin{equation}\label{eq:dynamical_system}
  \frac{\d \by_k(t)}{\d t}=\Phi_{\btheta}(t, \by_k(t), \by_k^-(t), \bx_k, r_k)\,,
\end{equation}
where $\Phi_{\btheta}$ is a differential operator, parameterized by $\btheta$. The notation $\by_k^-(t)$ refers to $\{\by_k(t')\,,t'<t\}$. Time-delay differential equations are commonly used to model dynamical systems of interest in fields such as epidemiology and demography. More details on the operator $\Phi_{\btheta}$ used can be found in \cite{jenner2021silico}. In particular, the parameter $\btheta$ (see Table TS3 in \cite{jenner2021silico}) consists essentially of a logarithmic transformation of biological rates such as rates at which certain cell cycles occur, infection rates of VV and VSV, etc.  It is important to stress that \eqref{eq:dynamical_system} cannot be solved exactly but that a computationally-intensive numerical solver (which requires solving a system of intermediate ordinary differential equations obtained by the so-called \textit{linear chain technique}) exists and is made available in \cite{jenner2021silico}. This solver needs to run at each evaluation of the likelihood function.

The work of \cite{jenner2021silico} is in a context of optimization of treatment schedule, and is based on a virtual cohort. The virtual patients are created from: 1) simulated covariate data points $\bx_k$ with summary statistics similar to real cohorts, 2) a treatment schedule $r_k$ that is assigned to each virtual patient, and 3) outputs $\hat{\by}(t, \btheta, \bx_k, r_k)$ resulting from the numerical solution of \eqref{eq:dynamical_system}. The parameter $\btheta$ used to produce the outputs $\hat{\by}(t, \btheta, \bx_k, r_k)$ is set to a value $\btheta^\ast$ based on expert opinion. Here the problem that we consider is that of numerical estimation of the unknown parameter $\btheta$ in the model \eqref{eqn:model} based on a data set.

Our data set has been simulated from a virtual cohort (generated as in \cite{jenner2021silico}) that is then fed into the model defined in \eqref{eqn:model}. The simulated data set consists of a virtual cohort of $K=10$ patients, where each patient is examined once per week for $T=20$ weeks. At each examination, $m_1=4$ statistics are measured; they are described in \autoref{sec5:fig1}. These data were simulated using the expert opinion ${\btheta^\ast}$. Having knowledge about the parameter (that is considered unknown in the Bayesian analysis) helps to evaluate the reliability of the different MTM algorithms.

 \begin{figure}[ht]
  \centering
  $\begin{array}{cc}
 \vspace{-2mm}\hspace{-5mm}\includegraphics[scale=.45]{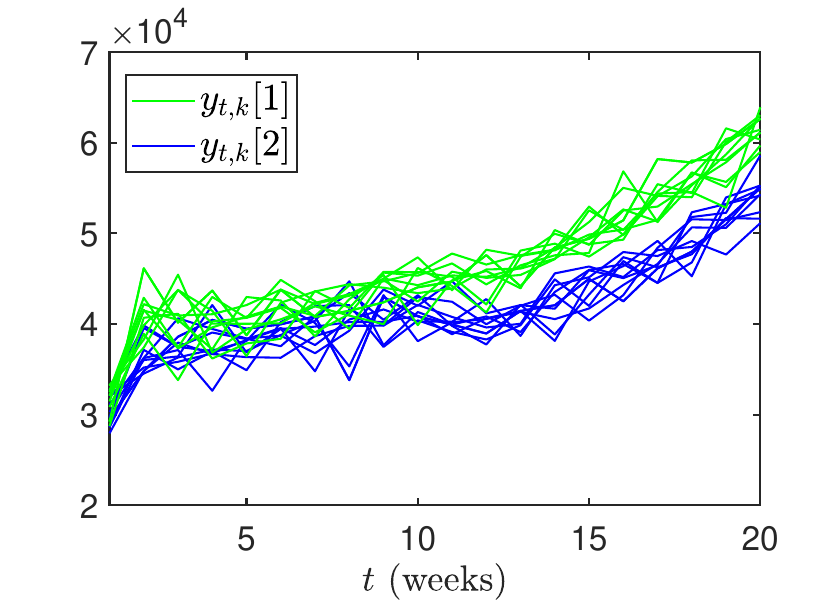}  &  \hspace{-4mm} \includegraphics[scale=.45]{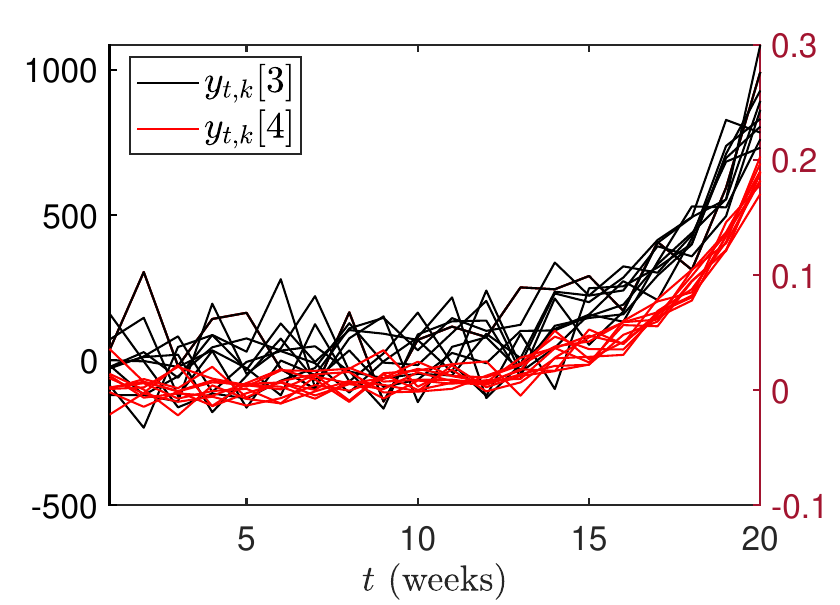} \cr
\end{array}$
\caption{Data $\{\by_k(t)\} := \{y_{t,k}[1], y_{t,k}[2], y_{t,k}[3], y_{t,k}[4]\}$ for the cohort, simulated using \eqref{eqn:model}; $y_{t,k}[1]$ is related to the number of quiescent tumor cells at time $t$ for individual $k$, $y_{t,k}[2]$ is related to the number of G1-phase tumour cell population, $y_{t,k}[3]$ is related to the total infected cell population, $y_{t,k}[4]$ is related to the total virus load. \label{sec5:fig1}}
\end{figure}

\subsection{Application study of MTM for Bayesian inference}\label{sec:inference_real_example}

The goal in practice is to perform Bayesian inference for $\btheta$, given the data related to the cohort of patients. We here discuss how one can proceed and present a study of the reliability of MTM algorithms.

A non-informative truncated Gaussian prior distribution is assigned to $\btheta$. The truncation stems from that the DDE solver is numerically unstable when the parameter value is beyond a compact set $\boldsymbol\Theta$, which includes $\btheta^\ast$. It can be readily checked that the posterior distribution verifies
\small
$$
\pi(\btheta\,|\,\{\by_{k}(t)\})\propto \bar\pi(\btheta\,|\,\{\by_{k}(t)\}):=\exp\left\{-\frac{1}{2\sigma^2}\sum_{t,k}(\by_{k}(t) - \hat{\by}(t, \btheta, \bx_k, r_k))^2-\frac{1}{2}\|\btheta-\boldsymbol\mu\|^2\right\}\1_{\btheta \in \boldsymbol\Theta}\,,
$$
\normalsize
where $\bar{\pi}$ is the unnormalized version of the posterior density and $\boldsymbol\mu\in\R^d$ is a prior hyperparameter. Because of the terms $\hat{\by}(t, \btheta, \bx_k, r_k)$, posterior expectations cannot be derived in closed form for non-trivial observables and IID sampling from $\pi$ is virtually impossible. Moreover, assuming that $\pi$ is differentiable, the gradient of $\pi$ is not available explicitly and is computationally expensive to approximate. This makes standard gradient-based methods, such as MALA or Hamiltonian Monte Carlo, computationally intensive and not straightforward to implement. Also, the posterior distribution exhibits an irregular and complex behaviour; see, e.g., the strong correlations and truncation effect in \autoref{sec5:fig2}. This makes posterior-approximation methods requiring a tractable approximation to the posterior distribution, such as importance sampling, independent MH, as well as default variational methods, highly non-trivial to apply successfully.

To perform Bayesian inference in such a situation, one may thus naturally turn to a random-walk Metropolis algorithm. If parallel computing is available, LB MTM with $w(\mathbf{x}, \mathbf{y}) = \sqrt{\pi(\mathbf{y}) / \pi(\mathbf{x})}$ is an appealing alternative as it can exploit parallel computing to speed up convergence. The non-trivial shape of the posterior distribution provides an interesting test case for the results derived in the previous sections with regular and isotropic target distributions. Here, we compare MH with LB MTM and GB MTM using different values of $N$. All samplers use the proposal $q_\sigma(\bx, \cdot\,) = \mathcal{N}(\bx, \sigma^2 \I_d)$, with $\sigma = \ell / \sqrt{d}$, and the parameter $\ell$ is adaptively tuned as in \autoref{sec:conv_numerics}.

\autoref{sec5:fig3} shows the convergence to stationarity of the different algorithms, monitored through the log-posterior-density value (up to a normalizing constant). The right panel of \autoref{sec5:fig3} displays the evolution of the step-size parameter  $\{\ell_m\}$ across MCMC iterations. We clearly observe the pathology of GB MTM described previously, with a performance that deteriorates as $N$ increases. In particular, it can be seen that GB MTM requires extremely small step sizes to navigate the low density regions (see \autoref{sec5:fig3}, right panel).

LB MTM instead exhibits a convergence speed that increases with $N$, as one expects. In particular, in this example, LB MTM with $N = 20$ converges significantly faster compared to $N=1$, roughly by a factor of $7$ based on a quantitative comparison of the trace plots in \autoref{sec5:fig3}, left panel. With our parallel implementation\footnote{In our experiment, we used a desktop computer with 32 cores (thus larger than the number of candidates $N$), AMD Ryzen 9 5950X processor,  Alma Linux 8.5 operating system, 64 GB of RAM, and off-the-shelf high-level MATLAB parallelization.}, the computational cost per iteration of MTM with $N=20$ is about 3.5 times higher than MH, resulting in an effective reduction of wall-clock time required for burn-in by a factor of roughly $2$. Note that the actual improvement ratio can depend heavily on the model and parallel implementation used, with the improvement typically larger for higher-dimensional problems.

\begin{figure}
\centering
\includegraphics[width=0.75\linewidth]{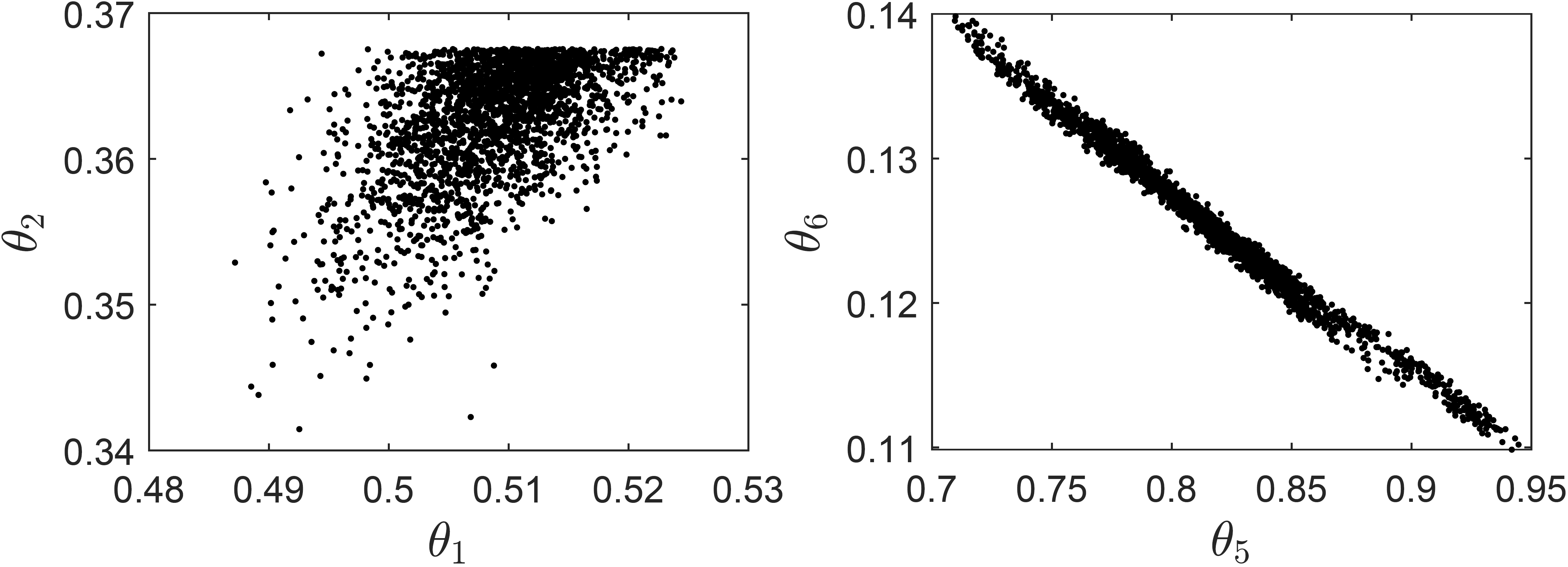}
\vspace{-2mm}
\caption{Pairwise marginal samples from $\pi$. \label{sec5:fig2}}
\end{figure}

\begin{figure}
  \centering
  $\begin{array}{cc}
 \vspace{-2mm}\hspace{-2mm}\includegraphics[width=.52\linewidth]{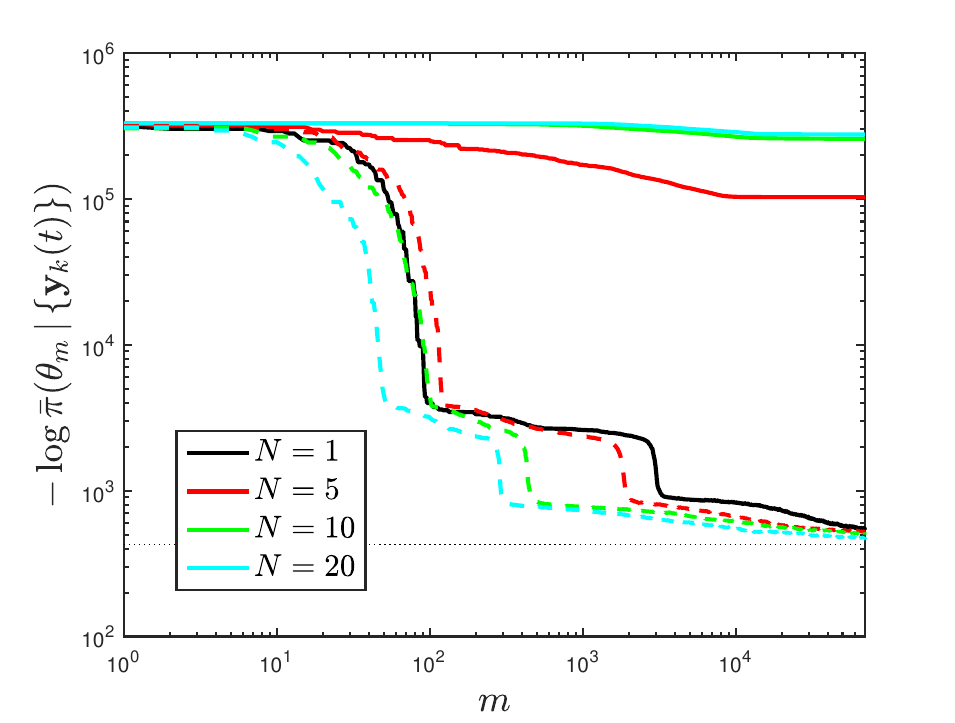}  &  \hspace{-4mm} \includegraphics[width=.48\linewidth]{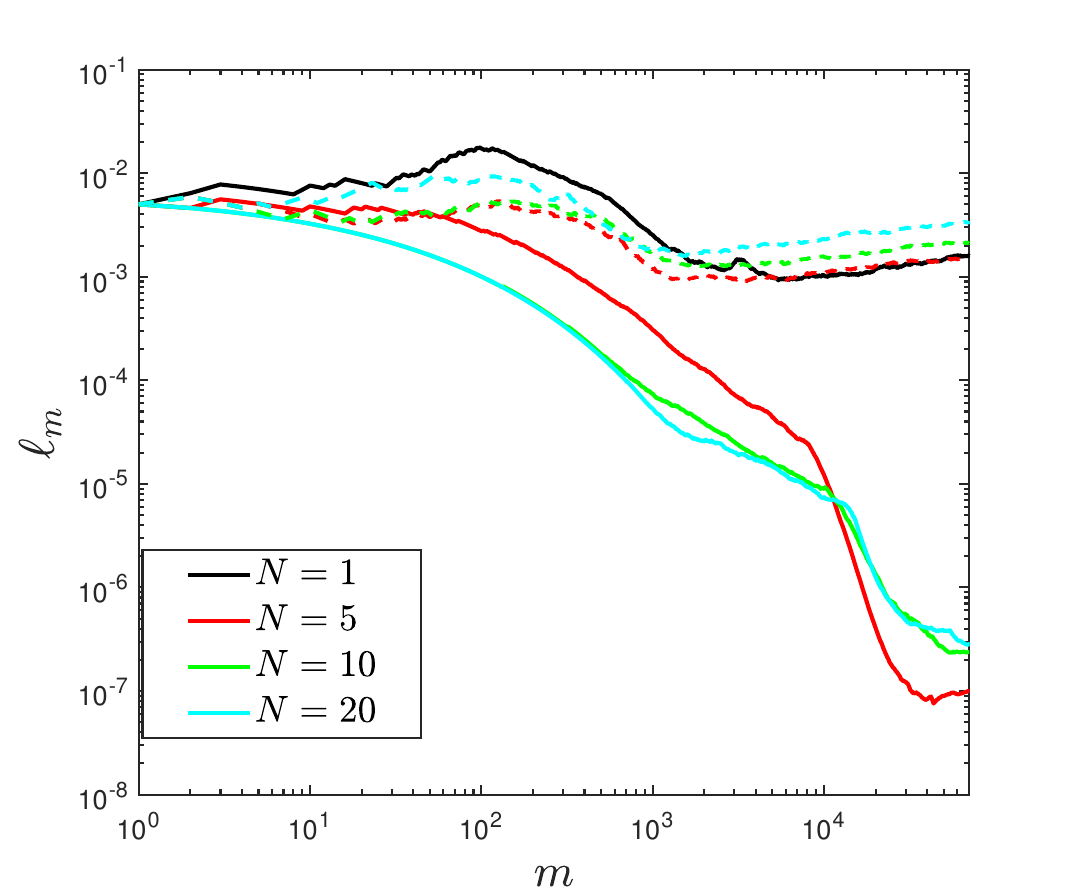} \cr
\end{array}$
\caption{Convergence of GB MTM (solid lines) and LB MTM (dashed lines), with a log-scale on the $x$-axis; all chains are started from the same state $\btheta_0$ belonging to the tails of $\pi$; left panel: trace plots of $-\log\bar\pi(\btheta_m\,|\,\{\by_{k}(t)\})$ indicating the progress towards the high-probability region (designated by the dotted black line); right panel: adaptation of the step-size parameter $\ell$ as the algorithms progress.}\label{sec5:fig3}
\end{figure}

\section{Discussion}\label{sec:discussion}

In this paper, we revisited the promises and pitfalls of a popular MCMC method, namely MTM, through several new theoretical and empirical results. We proposed to use a novel class of weight functions, based on the so-called locally-balanced proposal distributions, that can be employed instead of the classical GB weight function without impacting neither the computational cost nor the coding complexity of MTM. The resulting LB MTM scheme is remarkably and positively different compared to the GB MTM one, especially regarding the behaviour of the induced Markov chains during the convergence phase. This difference is associated with substantially reduced burn-in time and an element of stability that, together with an easier tuning of the method, make LB MTM an appealing and competitive algorithm for high-dimensional Bayesian inference problems for which nothing except the unnormalized posterior density is known and when the latter is computationally expensive to evaluate, which makes in-step parallel computing beneficial.

Part of the research effort conducted in this paper can also be cast as an attempt to generalize the use of LB samplers beyond the discrete-state-space scenario. In particular, LB MTM is similar in spirit to exact-approximation schemes, in the sense of \cite{andrieu2016establishing}, given that it can be thought of as approximations of ideal LB samplers that use a noisy version of the acceptance probability without affecting the invariant distribution. Further work in this direction is needed to explore these connections to strands of literature. Also, the results in this paper motivate further non-asymptotic (in $N$ and $d$) theoretical analyses to identify other quantitative or qualitative differences between LB MTM algorithms and their GB counterpart. The rich literature on non-asymptotic analysis of MH exact-approximations (see, e.g., \cite{andrieu2016establishing}, \cite{andrieu2020metropolis}, and \cite{wang2022theoretical}) may prove useful in this direction.

While we have focused mainly on MTM and the impact of the weight function, it would be interesting to provide a systematic comparison of LB MTM with other multiple-proposal MCMC schemes (e.g., \cite{neal2003markov,tjelmeland2004using,frenkel2004speed,calderhead2014general,holbrook2022generating} and references therein) and more generally to other approaches that could in principle exploit in-step parallelization to speed-up convergence of MCMC.
One such approach would be to exploit parallel target-density evaluations to derive numerical (e.g.\ finite-difference) approximations to the gradient, $\nabla\log\pi$, and then employ off-the-shelf gradient-based MCMC methods.
We leave the exploration of such comparisons for future work, noting that it is likely that the optimal scheme among the ones mentioned above will depend on features of the target distribution (such as dimensionality, smoothness and regularity), and on the parallel-computing environment available (see, e.g., \citet{glatt2022parallel} for examples of GPU implementations of multiple-proposal MCMC methods with large values of $N$)
as well as potentially other aspects.


\section*{Acknowledgments}

Philippe Gagnon acknowledges support from NSERC (Natural Sciences and Engineering Research Council of Canada) and FRQNT (Fonds de recherche du Québec -- Nature et technologies). Florian Maire acknowledges support from NSERC. Giacomo Zanella acknowledges support from the European Research Council (ERC), through StG ``PrSc-HDBayLe'' grant ID 101076564.
Also, the authors thank three anonymous referees for helpful suggestions that led to an improved manuscript. The authors additionally thank Morgan Craig, professor in the Department of Mathematics and Statistics of Universit\'{e} de Montr\'{e}al, for an introduction to the use of dynamical systems in precision medicine and for computer code to numerically solve the system of intermediate ordinary differential equations in the application of \autoref{sec:real_example}.

\bibliographystyle{rss}
\bibliography{references}

\appendix

\section{Proofs}\label{sec:proofs}

\begin{proof}[Proof of \autoref{prop:dist_MTM}]
 First, let us look at
  \[
  \P_{\bx}(J = j \mid \bY_1, \ldots, \bY_N) = \E_{\bx}[\1_{J = j} \mid \bY_1, \ldots, \bY_N] = \frac{w(\bx, \bY_j)}{\sum_{i = 1}^{N} w(\bx, \bY_i)}.
 \]
  Using this and that $\bY_1, \ldots, \bY_N$ are conditionally IID given $\bx$, we have that
 \begin{align*}
  \E_{\bx}[h(\mathbf{Y}_{J})] = \sum_{j = 1}^N \E_{\bx}[h(\bY_j) \, \1_{J = j}] &= \sum_{j = 1}^N \E_{\bx}[h(\bY_j) \, \E[\1_{J = j}\mid \bY_1, \ldots, \bY_N]] \cr
  &= \sum_{j = 1}^N \E_{\bx}\left[h(\bY_j) \, \frac{w(\bx, \bY_j)}{\sum_{i = 1}^{N} w(\bx, \bY_i)}\right] \cr
  &= \sum_{j = 1}^N \int h(\by_{j}) \, \frac{w(\bx, \by_j)}{\sum_{i = 1}^{N} w(\bx, \by_i)} \prod_{i=1}^{N} q_{\sigma_d}(\bx, \by_i) \, \d\by_{1:N} \cr
  &= \int h(\by_{1}) \, \frac{w(\bx, \by_1)}{\frac{1}{N}\sum_{i = 1}^{N} w(\bx, \by_i)} \prod_{i=1}^{N} q_{\sigma_d}(\bx, \by_i) \, \d\by_{1:N}.
 \end{align*}
\end{proof}

\subsection{Proof of \autoref{theorem:weak_con_MTM_d_fixed}}

In order to prove \autoref{theorem:weak_con_MTM_d_fixed}, we provide a general convergence result.

\begin{Theorem}\label{theorem:weak_con_gen}
 Let $\{P_N: N \geq 1\}$ be a collection of Markov transition kernels and $P$ a Markov transition kernel such that $P_N$ (for any $N$) and $P$ admit $\pi$ as invariant distribution. If the following assumptions hold:
   \begin{description}

    \item[(a)] the Markov transition kernels $\{P_N: N \geq 1\}$ satisfy
    \[
     \int |P_Nh(\mathbf{x}) - Ph(\mathbf{x})| \, \pi(\mathbf{x}) \, \d\mathbf{x} \rightarrow 0
    \]
    as $N \rightarrow \infty$ for all $h \in \mathcal{C}_{\text{b}}$, where
    \[
        P_Nh(\mathbf{x}) := \int P_N(\mathbf{x}, \d\mathbf{y}) \, h(\mathbf{y}),
       \]
       with an analogous definition for $Ph$, $\mathcal{C}_{\text{b}}$ being the space of bounded continuous functions;
    \item[(b)] $Ph$ is continuous for any $h \in \mathcal{C}_{\text{b}}$,
   \end{description}
   then $\{\mathbf{X}_N(m): m \in \N\}$ converges weakly to $\{\mathbf{X}(m): m \in \N\}$ provided that $\mathbf{X}_N(0) \sim \pi$ and $\mathbf{X}(0) \sim \pi$, $\{\mathbf{X}_N(m): m \in \N\}$ and $\{\mathbf{X}(m): m \in \N\}$ being the Markov chains with transition kernels $P_N$ and $P$, respectively.
\end{Theorem}

We present a proof here for self-containedness, but do not claim the originality of the result. The proof is strongly inspired by that of Theorem 2 in \cite{Schmon2020}.

\begin{proof}
 To prove that $\{\mathbf{X}_N(m): m \in \N\}$ converges weakly to $\{\mathbf{X}(m): m \in \N\}$, we only need to prove the convergence of finite-dimensional distributions \citep{karr1975weak}. It thus suffices to prove that for any positive integer $k$
 \[
  |\E[f_0(\mathbf{X}_N(0)) \ldots f_k(\mathbf{X}_N(k))] -  \E[f_0(\mathbf{X}(0)) \ldots f_k(\mathbf{X}(k))]| \rightarrow 0,
 \]
 as $N \rightarrow \infty$, for all $f_0, \ldots, f_k \in \mathcal{C}_{\text{b}}$ given that $(\bx(0), \ldots, \bx(k)) \mapsto \prod_{i=0}^k f_i(\bx(i))$  is measure determining by Proposition 4.6 in Chapter 2 of \cite{ethier1986markov}.

 We prove this by induction. For $k = 0$, the convergence is trivial because $\mathbf{X}_N(0) \sim \pi$ and $\mathbf{X} (0) \sim \pi$. Now assume that it is true for $k \geq 0$, and let us verify that it is true for $k + 1$.
 \begin{align*}
  &|\E[f_0(\mathbf{X}_N(0)) \ldots f_k(\mathbf{X}_N(k))f_{k + 1}(\mathbf{X}_N(k + 1))] -  \E[f_0(\mathbf{X}(0)) \ldots f_k(\mathbf{X}(k)) f_{k + 1}(\mathbf{X}(k + 1))]| \cr
  &=|\E[f_0(\mathbf{X}_N(0)) \ldots f_k(\mathbf{X}_N(k))P_Nf_{k + 1}(\mathbf{X}_N(k))] -  \E[f_0(\mathbf{X}(0)) \ldots f_k(\mathbf{X}(k)) Pf_{k + 1}(\mathbf{X}(k))]| \cr
  &\leq |\E[f_0(\mathbf{X}_N(0)) \ldots f_k(\mathbf{X}_N(k))P_Nf_{k + 1}(\mathbf{X}_N(k)) - f_0(\mathbf{X}_N(0)) \ldots f_k(\mathbf{X}_N(k))Pf_{k + 1}(\mathbf{X}_N(k))]| \cr
  &\quad + |\E[f_0(\mathbf{X}_N(0)) \ldots f_k(\mathbf{X}_N(k))Pf_{k + 1}(\mathbf{X}_N(k))] - \E[f_0(\mathbf{X}(0)) \ldots f_k(\mathbf{X}(k)) Pf_{k + 1}(\mathbf{X}(k))]|  \cr
  &\leq  M^k \E[|P_Nf_{k + 1}(\mathbf{X}_N(k)) - Pf_{k + 1}(\mathbf{X}_N(k))|] \cr
  &\quad + |\E[f_0(\mathbf{X}_N(0)) \ldots f_k(\mathbf{X}_N(k))Pf_{k + 1}(\mathbf{X}_N(k))] - \E[f_0(\mathbf{X}(0)) \ldots f_k(\mathbf{X}(k)) Pf_{k + 1}(\mathbf{X}(k))]|,
 \end{align*}
 using that there exists a positive constant $M$ such that $f_i \leq M$ for all $i$. The term on the penultimate line vanishes as a consequence of Assumption (a). That on the last line vanishes because $Pf_{k+1}$ is bounded and continuous by Assumption (b).
\end{proof}

Before presenting the proof of \autoref{theorem:weak_con_MTM_d_fixed}, we present a result that will be used in it and in other proofs.

\begin{Proposition}{\citep[Corollary 1.5.24]{muller2002comparison}}\label{prop:bound_convex}
 For any $N \geq 2$ exchangeable random variables $X_1, \ldots, X_N$ and any convex function $\phi$, we have
 \[
  \E\left[\phi\left(\frac{1}{N} \sum_{i=1}^N X_i\right)\right] \leq \E\left[\phi\left(\frac{1}{N-1} \sum_{i=1}^{N-1} X_i\right)\right],
 \]
 whenever the expectations exist.
\end{Proposition}

\begin{proof}[Proof of \autoref{theorem:weak_con_MTM_d_fixed}]
 We start with Result 1. The strategy to prove the result is the same as that to prove Theorem 2.1.1 in \cite{huggins2014information}. The probability that $\mathbf{Y}_J$ belongs to a set $A$, for fixed $\mathbf{x}$, as a function of $N$ is
 \[
  Q_{w, \sigma}^N(\mathbf{x}, A) := \P_{\mathbf{x}, N}(\mathbf{Y}_J \in A) = \E_{\bx}\left[\1_{\bY_1 \in A} \, \frac{w(\mathbf{x}, \mathbf{Y}_1)}{\frac{1}{N} \sum_{i = 1}^N w(\mathbf{x}, \mathbf{Y}_i)}\right],
 \]
using \autoref{prop:dist_MTM}. To simplify the notation for this part of the proof, we omit the dependence on $w, \sigma$ and $\bx$ that are fixed throughout, and use $Q^N(A) := Q_{w, \sigma}^N(\mathbf{x}, A)$ and $Q(A) := Q_{w, \sigma}(\mathbf{x}, A)$.

We have that
\begin{align*}
  Q^N(A) &= \E_{\bx}\left[\1_{\bY_1 \in A} \, w(\mathbf{x}, \mathbf{Y}_1) \, \E_{\bx}\left[\frac{1}{\frac{1}{N} \sum_{i = 1}^N w(\mathbf{x}, \mathbf{Y}_i)} \mid \bY_1\right]\right].
\end{align*}
With
\[
 h(\bY_1) := \E_{\bx}\left[\frac{1}{\frac{1}{N} \sum_{i = 1}^N w(\mathbf{x}, \mathbf{Y}_i)} \mid \bY_1\right],
\]
we have that
\begin{align*}
Q^N(A) = \E_{\bx}\left[\1_{\bY_1 \in A} \, w(\mathbf{x}, \mathbf{Y}_1) \, h(\bY_1)\right] &= \int_A w(\mathbf{x}, \mathbf{y}_1) \, h(\by_1) \, q_\sigma(\bx, \by_1) \, \d\by_1 \cr
&= \int_A \E_{\bx}[w(\mathbf{x}, \mathbf{Y}_1)] \, h(\by_1) \, Q_{w,\sigma}(\bx, \by_1) \, \d\by_1.
\end{align*}
This implies that we have an expression for the following Radon--Nikodym derivative:
\begin{align*}
 \frac{\d Q^N}{\d Q}(\by_1) &= \E_{\bx}[w(\mathbf{x}, \mathbf{Y}_1)] \E_{\bx}\left[\frac{1}{\frac{1}{N} \sum_{i = 1}^N w(\mathbf{x}, \mathbf{Y}_i)} \mid \bY_1 = \by_1\right] \cr
  &\geq \frac{ \E_{\bx}[w(\mathbf{x}, \mathbf{Y}_1)]}{\frac{1}{N} \sum_{i = 1}^N \E_{\bx}[w(\mathbf{x}, \mathbf{Y}_i)\mid \bY_1 = \by_1]} \cr
 &= \frac{ N\E_{\bx}[w(\mathbf{x}, \mathbf{Y}_1)]}{(N-1) \E_{\bx}[w(\mathbf{x}, \mathbf{Y}_1)] + w(\mathbf{x}, \mathbf{y}_1)},
\end{align*}
using Jensen's inequality.

To prove Result 1, we use that the total variation between $Q^N(A)$ and $Q(A)$ can be bounded above using the Kullback--Leibler divergence, and more precisely, by
\[
 \sqrt{\frac{1}{2} \int Q(\by_1) \, \d\by_1 \log \frac{\d Q}{\d Q^N}(\by_1)}.
\]
We have that
\begin{align*}
 \int Q(\by_1) \, \d\by_1 \log \frac{\d Q(\by_1)}{\d Q^N(\by_1)} &\leq \int Q(\by_1) \, \d\by_1 \log \frac{(N-1) \E_{\bx}[w(\mathbf{x}, \mathbf{Y}_1)] + w(\mathbf{x}, \mathbf{y}_1)}{ N\E_{\bx}[w(\mathbf{x}, \mathbf{Y}_1)]} \cr
 &\leq \log \int Q(\by_1) \, \d\by_1 \frac{(N-1) \E_{\bx}[w(\mathbf{x}, \mathbf{Y}_1)] + w(\mathbf{x}, \mathbf{y}_1)}{ N\E_{\bx}[w(\mathbf{x}, \mathbf{Y}_1)]} \cr
 &= \log\left(1 + \frac{\int w(\mathbf{x}, \mathbf{y}_1) Q(\by_1) \, \d\by_1 - \E_{\bx}[w(\mathbf{x}, \mathbf{Y}_1)]}{N\E_{\bx}[w(\mathbf{x}, \mathbf{Y}_1)]}\right) \cr
 &= \log\left(1 + \frac{\int w(\mathbf{x}, \mathbf{y}_1)^2 q_\sigma(\bx, \by_1) \, \d\by_1 - \E_{\bx}[w(\mathbf{x}, \mathbf{Y}_1)]^2}{N\E_{\bx}[w(\mathbf{x}, \mathbf{Y}_1)]^2}\right) \cr
 &\leq \frac{\var_{\bx}[w(\mathbf{x}, \mathbf{Y}_1)]}{N\E_{\bx}[w(\mathbf{x}, \mathbf{Y}_1)]^2},
\end{align*}
using Jensen's inequality and that $\log(1+x)\leq x$ for all $x > -1$. This concludes the proof of Result 1 as, by assumption $\E_{\bx}\left[w(\mathbf{x}, \mathbf{Y}_1)^4\right] < \infty$, and $\E_{\bx}[w(\mathbf{x}, \mathbf{Y}_1)] > 0$ given that $w$ is a strictly positive function.

 We now prove Result 2. To achieve this, we use \autoref{theorem:weak_con_gen}, which requires that:
  \begin{description}
    \item[(a)] for every $h \in \mathcal{C}_{\text{b}}$ (the space of bounded continuous functions),
        \[
         \int |P_Nh(\mathbf{x}) - P_{\text{ideal}}h(\mathbf{x})| \, \pi(\mathbf{x}) \, \d\mathbf{x} \rightarrow 0,
        \]
        where $P_N$ is the Markov kernel simulated by MTM and $P_{\text{ideal}}$ is the Markov kernel simulated by the ideal scheme;
    \item[(b)] $P_{\text{ideal}}h$ is continuous for any $h \in \mathcal{C}_{\text{b}}$.
  \end{description}

 Let us first define $P_N$ and $P_{\text{ideal}}$:
 \begin{align*}
  & P_N(\mathbf{x}, A)  := \sum_{j = 1}^N \int_{\mathbf{y}_j \in A} \frac{w(\mathbf{x}, \mathbf{y}_j)}{\sum_{i = 1}^N w(\mathbf{x}, \mathbf{y}_i)} \prod_{i=1}^N q_\sigma(\mathbf{x}, \mathbf{y}_i) \prod_{i=1}^{N - 1} q_\sigma(\mathbf{y}_j, \mathbf{z}_i) \, \alpha(\mathbf{x}, \mathbf{y}_j) \, \d(\mathbf{y}_{1:N}, \mathbf{z}_{1:N - 1}) \cr
  &\qquad  + \1_{\mathbf{x} \in A} \sum_{j = 1}^N \int \frac{w(\mathbf{x}, \mathbf{y}_j)}{\sum_{i = 1}^N w(\mathbf{x}, \mathbf{y}_i)} \prod_{i=1}^N q_\sigma(\mathbf{x}, \mathbf{y}_i) \prod_{i=1}^{N - 1} q_\sigma(\mathbf{y}_j, \mathbf{z}_i) \, (1 - \alpha(\mathbf{x}, \mathbf{y}_j)) \, \d(\mathbf{y}_{1:N}, \mathbf{z}_{1:N - 1}),
 \end{align*}
 and
 \begin{align*}
  P_{\text{ideal}}(\mathbf{x}, A) := \int_{\mathbf{y} \in A} Q_{w, \sigma}(\mathbf{x}, \mathbf{y}) \, \alpha_{\text{ideal}}(\mathbf{x}, \mathbf{y}) \, \d\mathbf{y}  + \1_{\mathbf{x} \in A} \int Q_{w, \sigma}(\mathbf{x}, \mathbf{y}) \, (1 - \alpha_{\text{ideal}}(\mathbf{x}, \mathbf{y})) \, \d\mathbf{y},
 \end{align*}
 where
 \[
   \alpha_{\text{ideal}}(\mathbf{x}, \mathbf{y}) := 1 \wedge \frac{\pi(\mathbf{y}) \, Q_{w, \sigma}(\mathbf{y}, \mathbf{x})}{\pi(\mathbf{x}) \, Q_{w, \sigma}(\mathbf{x}, \mathbf{y})}.
 \]
  In $P_N(\mathbf{x}, A)$, the integrals are the same for all $j$. Therefore,
  \begin{align*}
  & P_N(\mathbf{x}, A)  := \int_{\mathbf{y}_1 \in A} \frac{w(\mathbf{x}, \mathbf{y}_1)}{\frac{1}{N}\sum_{i = 1}^N w(\mathbf{x}, \mathbf{y}_i)} \prod_{i=1}^N q_\sigma(\mathbf{x}, \mathbf{y}_i) \prod_{i=1}^{N - 1} q_\sigma(\mathbf{y}_1, \mathbf{z}_i) \, \alpha(\mathbf{x}, \mathbf{y}_1) \, \d(\mathbf{y}_{1:N}, \mathbf{z}_{1:N - 1}) \cr
  &\qquad  + \1_{\mathbf{x} \in A} \int \frac{w(\mathbf{x}, \mathbf{y}_1)}{\frac{1}{N}\sum_{i = 1}^N w(\mathbf{x}, \mathbf{y}_i)} \prod_{i=1}^N q_\sigma(\mathbf{x}, \mathbf{y}_i) \prod_{i=1}^{N - 1} q_\sigma(\mathbf{y}_1, \mathbf{z}_i) \, (1 - \alpha(\mathbf{x}, \mathbf{y}_1)) \, \d(\mathbf{y}_{1:N}, \mathbf{z}_{1:N - 1}).
 \end{align*}
  Consequently,
 \begin{align*}
  & P_Nh(\mathbf{x}) = \int \frac{w(\mathbf{x}, \mathbf{y}_1)}{\frac{1}{N}\sum_{i = 1}^N w(\mathbf{x}, \mathbf{y}_i)} \prod_{i=1}^N q_\sigma(\mathbf{x}, \mathbf{y}_i) \prod_{i=1}^{N - 1} q_\sigma(\mathbf{y}_1, \mathbf{z}_i) \, \alpha(\mathbf{x}, \mathbf{y}_1) \, h(\mathbf{y}_1) \, \d(\mathbf{y}_{1:N}, \mathbf{z}_{1:N - 1}) \cr
  &\qquad  + h(\mathbf{x}) \int \frac{w(\mathbf{x}, \mathbf{y}_1)}{\frac{1}{N}\sum_{i = 1}^N w(\mathbf{x}, \mathbf{y}_i)} \prod_{i=1}^N q_\sigma(\mathbf{x}, \mathbf{y}_i) \prod_{i=1}^{N - 1} q_\sigma(\mathbf{y}_1, \mathbf{z}_i) \, (1 - \alpha(\mathbf{x}, \mathbf{y}_1)) \, \d(\mathbf{y}_{1:N}, \mathbf{z}_{1:N - 1}),
 \end{align*}
 and
  \begin{align}\label{eqn_Ph_ideal}
  P_{\text{ideal}}h(\mathbf{x}) := \int Q_{w, \sigma}(\mathbf{x}, \mathbf{y}) \, \alpha_{\text{ideal}}(\mathbf{x}, \mathbf{y}) \, h(\mathbf{y}) \, \d\mathbf{y}  + h(\mathbf{x}) \int Q_{w, \sigma}(\mathbf{x}, \mathbf{y}) \, (1 - \alpha_{\text{ideal}}(\mathbf{x}, \mathbf{y})) \, \d\mathbf{y}.
 \end{align}

 We first prove that
 \[
  \int |P_Nh(\mathbf{x}) - P_{\text{ideal}}h(\mathbf{x})| \, \pi(\mathbf{x}) \, \d\mathbf{x} \rightarrow 0.
 \]
 Using the triangle inequality, it suffices to prove that
 \begin{align*}
  &\int \left| \int \frac{w(\mathbf{x}, \mathbf{y}_1)}{\frac{1}{N}\sum_{i = 1}^N w(\mathbf{x}, \mathbf{y}_i)} \prod_{i=1}^N q_\sigma(\mathbf{x}, \mathbf{y}_i) \prod_{i=1}^{N - 1} q_\sigma(\mathbf{y}_1, \mathbf{z}_i) \, \alpha(\mathbf{x}, \mathbf{y}_1) \, h(\mathbf{y}_1) \, \d(\mathbf{y}_{1:N}, \mathbf{z}_{1:N - 1}) \right. \cr
  &\qquad \left. -  \int Q_{w, \sigma}(\mathbf{x}, \mathbf{y}) \, \alpha_{\text{ideal}}(\mathbf{x}, \mathbf{y}) \, h(\mathbf{y}) \, \d\mathbf{y}\right| \, \pi(\mathbf{x}) \, \d\mathbf{x} \rightarrow 0,
 \end{align*}
 and
  \begin{align*}
  &\int \left| h(\mathbf{x}) \int \frac{w(\mathbf{x}, \mathbf{y}_1)}{\frac{1}{N}\sum_{i = 1}^N w(\mathbf{x}, \mathbf{y}_i)} \prod_{i=1}^N q_\sigma(\mathbf{x}, \mathbf{y}_i) \prod_{i=1}^{N - 1} q_\sigma(\mathbf{y}_1, \mathbf{z}_i) \, (1 - \alpha(\mathbf{x}, \mathbf{y}_1)) \, \d(\mathbf{y}_{1:N}, \mathbf{z}_{1:N - 1}) \right. \cr
  &\qquad \left. -  h(\mathbf{x}) \int Q_{w, \sigma}(\mathbf{x}, \mathbf{y}) \, (1 - \alpha_{\text{ideal}}(\mathbf{x}, \mathbf{y})) \, \d\mathbf{y} \right| \, \pi(\mathbf{x}) \, \d\mathbf{x} \rightarrow 0.
 \end{align*}
 We prove the first convergence and the other one follows using the same arguments. Using that we can rewrite
  \begin{align*}
   &\int Q_{w, \sigma}(\mathbf{x}, \mathbf{y}) \, \alpha_{\text{ideal}}(\mathbf{x}, \mathbf{y}) \, h(\mathbf{y}) \, \d\mathbf{y} \cr
    &\qquad = \int \frac{w(\mathbf{x}, \mathbf{y}_1)}{\int  w(\mathbf{x}, \mathbf{y}_1) \, q_\sigma(\mathbf{x}, \mathbf{y}_1) \, \d\mathbf{y}_1} \, \alpha_{\text{ideal}}(\mathbf{x}, \mathbf{y}_1) \, h(\mathbf{y}_1) \, \d(\mathbf{y}_{1:N}, \mathbf{z}_{1:N - 1}),
  \end{align*}
  Jensen's inequality and that $h$ is bounded, let us say by a positive constant $M$,
  \begin{align*}
  &\int \left| \int \frac{w(\mathbf{x}, \mathbf{y}_1)}{\frac{1}{N}\sum_{i = 1}^N w(\mathbf{x}, \mathbf{y}_i)} \prod_{i=1}^N q_\sigma(\mathbf{x}, \mathbf{y}_i) \prod_{i=1}^{N - 1} q_\sigma(\mathbf{y}_1, \mathbf{z}_i) \, \alpha(\mathbf{x}, \mathbf{y}_1) \, h(\mathbf{y}_1) \, \d(\mathbf{y}_{1:N}, \mathbf{z}_{1:N - 1}) \right. \cr
  &\qquad \left. -  \int Q_{w, \sigma}(\mathbf{x}, \mathbf{y}) \, \alpha_{\text{ideal}}(\mathbf{x}, \mathbf{y}) \, h(\mathbf{y}) \, \d\mathbf{y}\right| \, \pi(\mathbf{x}) \, \d\mathbf{x} \cr
  &\leq M \iint \left| \frac{w(\mathbf{x}, \mathbf{y}_1)}{\frac{1}{N}\sum_{i = 1}^N w(\mathbf{x}, \mathbf{y}_i)} \, \alpha(\mathbf{x}, \mathbf{y}_1) - \frac{w(\mathbf{x}, \mathbf{y}_1)}{\int  w(\mathbf{x}, \mathbf{y}_1) \, q_\sigma(\mathbf{x}, \mathbf{y}_1) \, \d\mathbf{y}_1} \, \alpha_{\text{ideal}}(\mathbf{x}, \mathbf{y}_1)\right| \cr
  &\qquad \times  \prod_{i=1}^N q_\sigma(\mathbf{x}, \mathbf{y}_i) \prod_{i=1}^{N - 1} q_\sigma(\mathbf{y}_1, \mathbf{z}_i) \, \d(\mathbf{y}_{1:N}, \mathbf{z}_{1:N - 1}) \, \pi(\mathbf{x}) \, \d\mathbf{x} \cr
  &= M \E\left[\left|\frac{w(\mathbf{X}, \mathbf{Y}_1)}{\frac{1}{N}\sum_{i = 1}^N w(\mathbf{X}, \mathbf{Y}_i)} \, \alpha(\mathbf{X}, \mathbf{Y}_1) - \frac{w(\mathbf{X}, \mathbf{Y}_1)}{\int  w(\mathbf{X}, \mathbf{y}_1) \, q_\sigma(\mathbf{X}, \mathbf{y}_1) \, \d\mathbf{y}_1} \, \alpha_{\text{ideal}}(\mathbf{X}, \mathbf{Y}_1) \right|\right].
 \end{align*}
 From the strong law of large numbers, we have that with probability 1
 \[
   \frac{1}{N} \sum_{i = 1}^N w(\mathbf{x}, \mathbf{Y}_i) \rightarrow \int  w(\mathbf{x}, \mathbf{y}_1) \, q_\sigma(\mathbf{x}, \mathbf{y}_1) \, \d\mathbf{y}_1, \quad \text{as $N \rightarrow \infty$.}
 \]
 for all $\bx$. Therefore, with probability 1,
 \[
  \frac{1}{N}\sum_{i = 1}^N w(\mathbf{X}, \mathbf{Y}_i) \rightarrow \int  w(\mathbf{X}, \mathbf{y}_1) \, q_\sigma(\mathbf{X}, \mathbf{y}_1) \, \d\mathbf{y}_1.
 \]
 Also, with probability 1,
 \begin{align*}
  \alpha(\mathbf{X}, \mathbf{Y}_1) &= 1 \wedge \frac{\pi(\mathbf{Y}_1) \, q_\sigma(\mathbf{Y}_1, \mathbf{X}) \, w(\mathbf{Y}_1, \mathbf{X}) \bigg/ \left(\frac{1}{N}\left(\sum_{i = 1}^{N - 1} w(\mathbf{Y}_1, \mathbf{Z}_i) + w(\mathbf{Y}_1, \mathbf{X})\right)\right)}{\pi(\mathbf{X}) \, q_\sigma(\mathbf{X}, \mathbf{Y}_1) \, w(\mathbf{X}, \mathbf{Y}_1) \bigg/ \left(\frac{1}{N}\left(\sum_{i = 2}^N w(\mathbf{X}, \mathbf{Y}_i) + w(\mathbf{X}, \mathbf{Y}_1)\right)\right)} \cr
  &\rightarrow 1 \wedge \frac{\pi(\mathbf{Y}_1) \, q_\sigma(\mathbf{Y}_1, \mathbf{X}) \, w(\mathbf{Y}_1, \mathbf{X}) \bigg/ \int w(\mathbf{Y}_1, \mathbf{z}_1) \, q_\sigma(\mathbf{Y}_1, \mathbf{z}_1) \, \d\mathbf{z}_1}{\pi(\mathbf{X}) \, q_\sigma(\mathbf{X}, \mathbf{Y}_1) \, w(\mathbf{X}, \mathbf{Y}_1) \bigg/ \int w(\mathbf{X}, \mathbf{y}_1) \, q_\sigma(\mathbf{X}_1, \mathbf{y}_1) \, \d\mathbf{y}_1} \cr
  &= 1 \wedge \frac{\pi(\mathbf{Y}_1) \, Q_{w, \sigma}(\mathbf{Y}_1, \mathbf{X})}{\pi(\mathbf{X}) \, Q_{w, \sigma}(\mathbf{X}, \mathbf{Y}_1)} = \alpha_{\text{ideal}}(\mathbf{X}, \mathbf{Y}_1).
 \end{align*}
 Therefore, with probability 1,
 \begin{align*}
 \left|\frac{w(\mathbf{X}, \mathbf{Y}_1)}{\frac{1}{N}\sum_{i = 1}^N w(\mathbf{X}, \mathbf{Y}_i)} \, \alpha(\mathbf{X}, \mathbf{Y}_1) - \frac{w(\mathbf{X}, \mathbf{Y}_1)}{\int  w(\mathbf{X}, \mathbf{y}_1) \, q_\sigma(\mathbf{X}, \mathbf{y}_1) \, \d\mathbf{y}_1} \, \alpha_{\text{ideal}}(\mathbf{X}, \mathbf{Y}_1) \right| \rightarrow 0.
 \end{align*}

 To be able to conclude that the expectation converges, we prove that the random variable is uniformly integrable. We more specifically prove that
 \[
  \sup_N \E\left[\left(\frac{w(\mathbf{X}, \mathbf{Y}_1)}{\frac{1}{N}\sum_{i = 1}^N w(\mathbf{X}, \mathbf{Y}_i)} \, \alpha(\mathbf{X}, \mathbf{Y}_1) - \frac{w(\mathbf{X}, \mathbf{Y}_1)}{\int  w(\mathbf{X}, \mathbf{y}_1) \, q_\sigma(\mathbf{X}, \mathbf{y}_1) \, \d\mathbf{y}_1} \, \alpha_{\text{ideal}}(\mathbf{X}, \mathbf{Y}_1) \right)^2\right] < \infty,
 \]
 which implies that the random variable is uniformly integrable. We have that
 \begin{align*}
  &\E\left[\left(\frac{w(\mathbf{X}, \mathbf{Y}_1)}{\frac{1}{N}\sum_{i = 1}^N w(\mathbf{X}, \mathbf{Y}_i)} \, \alpha(\mathbf{X}, \mathbf{Y}_1) - \frac{w(\mathbf{X}, \mathbf{Y}_1)}{\int  w(\mathbf{X}, \mathbf{y}_1) \, q_\sigma(\mathbf{X}, \mathbf{y}_1) \, \d\mathbf{y}_1} \, \alpha_{\text{ideal}}(\mathbf{X}, \mathbf{Y}_1) \right)^2\right] \cr
  & \leq 2\E\left[\left(\frac{w(\mathbf{X}, \mathbf{Y}_1)}{\frac{1}{N}\sum_{i = 1}^N w(\mathbf{X}, \mathbf{Y}_i)} \, \alpha(\mathbf{X}, \mathbf{Y}_1)\right)^2\right] \cr
   &\qquad + 2\E\left[\left(\frac{w(\mathbf{X}, \mathbf{Y}_1)}{\int  w(\mathbf{X}, \mathbf{y}_1) \, q_\sigma(\mathbf{X}, \mathbf{y}_1) \, \d\mathbf{y}_1} \, \alpha_{\text{ideal}}(\mathbf{X}, \mathbf{Y}_1)\right)^2\right] \cr
  & \leq 2\E\left[\left(\frac{w(\mathbf{X}, \mathbf{Y}_1)}{\frac{1}{N}\sum_{i = 1}^N w(\mathbf{X}, \mathbf{Y}_i)}\right)^2\right] + 2\E\left[\left(\frac{w(\mathbf{X}, \mathbf{Y}_1)}{\int  w(\mathbf{X}, \mathbf{y}_1) \, q_\sigma(\mathbf{X}, \mathbf{y}_1) \, \d\mathbf{y}_1} \right)^2\right] \cr
  & \leq 2\E\left[w(\mathbf{X}, \mathbf{Y}_1)^4\right]^{1/2}\E\left[w(\mathbf{X}, \mathbf{Y}_1)^{-4}\right]^{1/2} +  2\E\left[w(\mathbf{X}, \mathbf{Y}_1)^4 \right]^{1/2} \E\left[\left(\int  w(\mathbf{X}, \mathbf{y}_1) \, q_\sigma(\mathbf{X}, \mathbf{y}_1) \, \d\mathbf{y}_1\right)^{-4}\right]^{1/2} \cr
  &\leq 4\E\left[w(\mathbf{X}, \mathbf{Y}_1)^4\right]^{1/2}\E\left[w(\mathbf{X}, \mathbf{Y}_1)^{-4}\right]^{1/2},
 \end{align*}
 which is finite. We used that for any real numbers $a,b$, $(a+b)^2 \leq 2a^2 + 2b^2$, that $0\leq \alpha, \alpha_{\text{ideal}} \leq 1$, Cauchy--Schwarz inequality, Jensen's inequality (at the conditional expectation level), and \autoref{prop:bound_convex} for
  \[
  \E_{\bx}\left[\left(\frac{1}{N} \sum_{i = 1}^N w(\mathbf{x}, \mathbf{Y}_i)\right)^{-4}\right] \leq \E_{\bx}\left[w(\mathbf{x}, \mathbf{Y}_1)^{-4}\right],
 \]
 with $\E\left[w(\mathbf{X}, \mathbf{Y}_1)^{-4}\right] < \infty$.

 We now show that $P_{\text{ideal}}h$ is continuous. Define $\mathbf{x}_{\boldsymbol\epsilon} := \mathbf{x} + \boldsymbol\epsilon$ with $\|\boldsymbol\epsilon\| \leq \varepsilon$, where $\varepsilon > 0$ can be chosen to be arbitrarily small. We want to prove that
 \begin{align*}
  \lim_{\varepsilon \rightarrow 0} P_{\text{ideal}}h(\mathbf{x}_{\boldsymbol\epsilon}) = P_{\text{ideal}}h(\mathbf{x}).
 \end{align*}
 This is true under the assumptions if we can interchange the limit and the integral in \eqref{eqn_Ph_ideal}. We are allowed to do it using the dominated convergence theorem because $h$ is bounded, $0 \leq\alpha_{\text{ideal}} \leq 1$ and $Q_{w, \sigma}(\mathbf{x} + \boldsymbol\epsilon, \mathbf{y}) \leq f(\mathbf{x}, \mathbf{y})$, an integrable (in $\mathbf{y}$) function, for all $\mathbf{x}$.
\end{proof}

\subsection{Proof of \autoref{prop:ass_b}}

\begin{proof}[Proof of \autoref{prop:ass_b}]
 We prove the result for the case where $w(\bx, \by) = \pi(\by)/\pi(\bx)$. The proof is analogous for the case where  $w(\bx, \by) = \sqrt{\pi(\by)/\pi(\bx)}$.

 We have that
 \begin{align*}
  Q_{w, \sigma}(\bx + \boldsymbol\epsilon, \by) &= \frac{w(\bx + \boldsymbol\epsilon, \by) \, q_{\sigma}(\bx + \boldsymbol\epsilon, \by)}{\int w(\bx + \boldsymbol\epsilon, \by) \, q_{\sigma}(\bx + \boldsymbol\epsilon, \by) \, \d\by} = \frac{\pi(\by) \, q_{\sigma}(\bx + \boldsymbol\epsilon, \by)}{\int \pi(\by) \, q_{\sigma}(\bx + \boldsymbol\epsilon, \by) \, \d\by}.
 \end{align*}
 Also,
 \begin{align*}
    &q_{\sigma}(\bx + \boldsymbol\epsilon, \by) = \frac{1}{(2 \pi \sigma^2)^{d/2}} \exp\left(-\frac{1}{2\sigma^2}(\by - (\bx+\boldsymbol\epsilon))^T (\by - (\bx+\boldsymbol\epsilon))\right) \cr
    &\qquad = \frac{1}{(2 \pi \sigma^2)^{d/2}} \exp\left(-\frac{1}{2\sigma^2}(\by - \bx)^T (\by - \bx)\right) \exp\left(\frac{1}{2\sigma^2}\boldsymbol\epsilon^T (\by - \bx)\right)\exp\left(-\frac{1}{2\sigma^2} \, \boldsymbol\epsilon^T \boldsymbol\epsilon\right).
 \end{align*}
 Therefore,
 \[
  q_{\sigma}(\bx + \boldsymbol\epsilon, \by) \leq q_{\sigma}(\bx, \by)  \exp\left(\frac{\varepsilon}{2\sigma^2}  \|\by - \bx\|\right),
 \]
 given that
 \[
  \exp\left(-\frac{1}{2\sigma^2} \, \boldsymbol\epsilon^T \boldsymbol\epsilon\right) \leq 1,
 \]
 and
 \[
  \boldsymbol\epsilon^T (\by - \bx) \leq |\boldsymbol\epsilon^T (\by - \bx)| \leq \|\boldsymbol\epsilon\| \|\by - \bx\| \leq \varepsilon \|\by - \bx\|,
 \]
 by Cauchy--Schwarz inequality.

 We thus have an upper bound of the numerator of $Q_{w, \sigma}(\bx + \boldsymbol\epsilon, \by)$ that does not depend on $\boldsymbol\epsilon$. We now find a lower bound of the denominator that does not depend on $\boldsymbol\epsilon$. By Fatou's lemma, we have that
 \[
  \liminf_{\varepsilon \rightarrow 0} \int \pi(\by) \, q_{\sigma}(\bx + \boldsymbol\epsilon, \by) \, \d\by \geq \int \pi(\by) \, q_{\sigma}(\bx, \by) \, \d\by.
 \]
 We have that $\|\boldsymbol\epsilon\| \leq \varepsilon$ with $\varepsilon > 0$ an arbitrarily small value. We thus know that we can choose $\varepsilon$ such that
 \[
 \int \pi(\by) \, q_{\sigma}(\bx + \boldsymbol\epsilon, \by) \, \d\by \geq \int \pi(\by) \, q_{\sigma}(\bx, \by) \, \d\by - \xi,
 \]
 with $\xi$ an arbitrarily small value. In particular, we can choose $\varepsilon$ such that
 \[
  \xi \leq \frac{1}{2} \int \pi(\by) \, q_{\sigma}(\bx, \by) \, \d\by,
 \]
 implying that
 \[
 \int \pi(\by) \, q_{\sigma}(\bx + \boldsymbol\epsilon, \by) \, \d\by \geq \frac{1}{2} \int \pi(\by) \, q_{\sigma}(\bx, \by) \, \d\by.
 \]
 Note that $\int \pi(\by) \, q_{\sigma}(\bx, \by) \, \d\by > 0$ because $\pi(\by)$ is assumed to be strictly positive. Note also that $\int \pi(\by) \, q_{\sigma}(\bx, \by) \, \d\by < \infty$. Indeed,
 \[
  \int \pi(\by) \, q_{\sigma}(\bx, \by) \, \d\by \leq \frac{1}{(2 \pi \sigma^2)^{d/2}} \int \pi(\by) \, \d\by < \infty.
 \]
 In the case where  $w(\bx, \by) = \sqrt{\pi(\by)/\pi(\bx)}$, we use Cauchy--Schwarz inequality instead of the boundedness of $q_{\sigma}$ to reach the same conclusion.

 To summarize, we know that we can choose $\varepsilon$ so that
 \[
  Q_{w, \sigma}(\bx + \boldsymbol\epsilon, \by) \leq \frac{\pi(\by) \, q_{\sigma}(\bx, \by)  \exp\left(\frac{\varepsilon}{2\sigma^2}  \|\by - \bx\|\right)}{\frac{1}{2} \int \pi(\by) \, q_{\sigma}(\bx, \by) \, \d\by},
 \]
 which is independent of $\boldsymbol\epsilon$. We know that the denominator on the right-hand side (RHS) is strictly positive and finite. To conclude the proof, we need to show that the numerator is integrable.

 We have that
 \begin{align*}
  \int \pi(\by) \, q_{\sigma}(\bx, \by)  \exp\left(\frac{\varepsilon}{2\sigma^2}  \|\by - \bx\|\right) \, \d\by &=  \int_{A_{\bx}} \pi(\by) \, q_{\sigma}(\bx, \by)  \exp\left(\frac{\varepsilon}{2\sigma^2}  \|\by - \bx\|\right) \, \d\by \cr
 &\quad + \int_{A_{\bx}^\mathsf{c}} \pi(\by) \, q_{\sigma}(\bx, \by)  \exp\left(\frac{\varepsilon}{2\sigma^2}  \|\by - \bx\|\right) \, \d\by,
 \end{align*}
 where $A_{\bx}$ is the set of values of $\by$, with $\bx$ fixed, such that $\|\by - \bx\| \leq 1$. Given that this set is compact and that the integrand is upper bounded, we know that
 \[
  \int_{A_{\bx}} \pi(\by) \, q_{\sigma}(\bx, \by)  \exp\left(\frac{\varepsilon}{2\sigma^2}  \|\by - \bx\|\right) \, \d\by < \infty.
 \]

 Now, let us analyse the other integral, we have that
 \begin{align*}
  \int_{A_{\bx}^\mathsf{c}} \pi(\by) \, q_{\sigma}(\bx, \by)  \exp\left(\frac{\varepsilon}{2\sigma^2}  \|\by - \bx\|\right) \, \d\by &\leq \int_{A_{\bx}^\mathsf{c}} \pi(\by) \, q_{\sigma}(\bx, \by)  \exp\left(\frac{\varepsilon}{2\sigma^2}  \|\by - \bx\|^2\right) \, \d\by \cr
  &\hspace{-5mm}\leq \int \pi(\by) \, \frac{1}{(2 \pi \sigma^2)^{d/2}} \exp\left(-\frac{1 - \varepsilon}{2\sigma^2}\|\by - \bx\|^2\right) \,  \d\by \cr
  &\hspace{-5mm}\leq \frac{1}{(2 \pi \sigma^2)^{d/2}} \int \pi(\by) \, \d\by < \infty,
 \end{align*}
 which concludes the proof.
\end{proof}

\subsection{Proof of \autoref{prop:upper_bound_acc}}\label{sec:proofs_upper_bound_acc}

Before presenting the proof of \autoref{prop:upper_bound_acc}, we provide two lemmas which are used in it and in several other proofs. After presenting the proof of \autoref{prop:upper_bound_acc}, we present a result stating that $\lim_{\|\bx\| \rightarrow \infty}  \E_{\bx}[\alpha_{\text{ideal}}(\bx, \bY)] = 0$ for any $\sigma$ and $d$ under weaker assumptions than those supposed in \autoref{prop:upper_bound_acc}. As mentioned in \autoref{sec:tails}, the assumptions are essentially that $U := - \log \pi$ is strongly convex and $L$-smooth. In the proof of that result, it will be noticed that $L$-smoothness is not necessary but makes the proof simpler. It will also be noticed in the statement of the result that, instead of strong convexity, we assume that $\|\nabla U(\bx)\| \rightarrow \infty$ as $\|\bx\|  \rightarrow \infty$, which is weaker, but is in fact the crucial assumption for having $\lim_{\|\bx\| \rightarrow \infty}  \E_{\bx}[\alpha_{\text{ideal}}(\bx, \bY)] = 0$.

\begin{Lemma}\label{lemma:acc_prob}
 When the target density is defined as in \eqref{eqn:target_product} and the proposal distribution is $Q_{w, \sigma}$ with $q_\sigma(\bx, \cdot \,) = \mathcal{N}(\mathbf{x}, \sigma^2 \I_d)$ and $w(\bx, \by) = \pi(\by)/\pi(\bx)$, we have that
 \[
 \alpha_{\text{ideal}}(\mathbf{x}, \mathbf{y}) = 1 \wedge \exp\left(\frac{1}{2(1+\sigma^2)} \sum_{i=1}^d (y_i^2 - x_i^2)\right).
\]
\end{Lemma}

\begin{proof}[Proof of \autoref{lemma:acc_prob}]
 We have that
 \begin{align*}
  \alpha_{\text{ideal}}(\mathbf{x}, \mathbf{y}) &= 1 \wedge \frac{\pi(\mathbf{y}) \, Q_{w, \sigma}(\mathbf{y}, \mathbf{x})}{\pi(\mathbf{x}) \, Q_{w, \sigma}(\mathbf{x}, \mathbf{y})} \cr
  &= 1 \wedge \prod_{i = 1}^d \frac{\int \varphi(z_i) \, (1 / \sigma) \varphi((z_i - x_i) / \sigma) \, \d z_i}{\int \varphi(z_i) \, (1 / \sigma) \varphi((z_i - y_i) / \sigma)  \, \d z_i} \cr
  &= 1 \wedge \prod_{i = 1}^d \frac{\int \varphi(x_i + \sigma u_i) \, \varphi(u_i) \, \d u_i}{\int \varphi(y_i + \sigma u_i) \, \varphi(u_i)  \, \d u_i} \cr
  &= 1 \wedge \prod_{i = 1}^d \frac{\int \exp\left(-\frac{\sigma^2}{2}(u_i + x_i/\sigma)^2\right) \, \varphi(u_i) \, \d u_i}{\int \exp\left(-\frac{\sigma^2}{2}(u_i + y_i/\sigma)^2\right) \, \varphi(u_i)  \, \d u_i} \cr
  &= 1 \wedge \exp\left(\frac{1}{2(1+\sigma^2)} \sum_{i=1}^d (y_i^2 - x_i^2)\right),
 \end{align*}
using the definition of $Q_{w, \sigma}$, the factorization of the target and proposal densities, a change of variable and the equality
\[
 \int \exp\left(-\frac{\sigma^2}{2}(u_i + x_i/\sigma_d)^2\right) \, \varphi(u_i) \, \d u_i = \E\left[\exp\left(-\frac{\sigma^2}{2}Z_i\right)\right] = \frac{\exp\left(-\frac{x_i^2}{2(1+\sigma^2)}\right)}{(1+\sigma^2)^{1/2}},
\]
with $Z_i$ that follows a non-central chi-squared distribution.
\end{proof}

\begin{Lemma}\label{lemma:dist_Y_ideal}
 When the target density is defined as in \eqref{eqn:target_product} and the proposal distribution is $Q_{w, \sigma}$ with $q_\sigma(\bx, \cdot \,) = \mathcal{N}(\mathbf{x}, \sigma^2 \I_d)$ and $w(\bx, \by) = \pi(\by)/\pi(\bx)$, we have that $\bY \sim Q_{w, \sigma}(\bx, \cdot \,)$ is equal in distribution to
 \[
 \frac{\bx}{1 + \sigma^2} + \sqrt{\frac{\sigma^2}{1 + \sigma^2}} \, \bU,
\]
where the components of $\bU := (U_1, \ldots, U_d)^T$ are $d$ independent standard-normal random variables.
\end{Lemma}

\begin{proof}[Proof of \autoref{lemma:dist_Y_ideal}]
 The PDF $Q_{w, \sigma}(\mathbf{x}, \cdot \, )$ is such that
\[
 Q_{w, \sigma}(\mathbf{x}, \mathbf{y}) \propto \prod_{i = 1}^d \varphi(y_i) \, (1 / \sigma) \varphi((y_i - x_i) / \sigma) \propto \prod_{i = 1}^d \exp\left(-\frac{1}{2}\frac{1 + \sigma^2}{\sigma^2}\left(y_i - \frac{x_i}{1 + \sigma^2}\right)^2\right),
\]
implying that the components in $\bY := (Y_1, \ldots, Y_d)^T$ are $d$ conditionally independent random variables with
\[
 Y_i \sim \mathcal{N}\left(\frac{X_i}{1 + \sigma^2}, \frac{\sigma^2}{1 + \sigma^2}\right),
 \]
  for all $i$. This concludes the proof.
\end{proof}

\begin{proof}[Proof of \autoref{prop:upper_bound_acc}]
Using \autoref{lemma:acc_prob},
\[
 \E[\alpha_{\text{ideal}}(\mathbf{x}, \mathbf{Y})] = \E\left[1 \wedge \exp\left(\frac{1}{2(1+\sigma^2)} \sum_{i=1}^d (Y_i^2 - x_i^2)\right)\right].
\]
We have that
\begin{align*}
& \E\left[1 \wedge \exp\left(\frac{1}{2(1+\sigma^2)} \sum_{i=1}^d (Y_i^2 - x_i^2)\right)\right] \leq \E\left[ \exp\left(\frac{1}{2(1+\sigma^2)} \sum_{i=1}^d (Y_i^2 - x_i^2)\right)\right] \cr
 &\hspace{10mm}= \exp\left(-\frac{1}{2(1+\sigma^2)} \, \|\mathbf{x}\|^2\right) \E\left[ \exp\left(\frac{1}{2(1+\sigma^2)} \sum_{i=1}^d \left(\frac{x_i}{1+\sigma^2} + \sqrt{\frac{\sigma^2}{1 + \sigma^2}} \, U_i\right)^2\right)\right] \cr
 &\hspace{10mm}= \exp\left(-\frac{1}{2(1+\sigma^2)} \, \|\mathbf{x}\|^2\right) \E\left[ \exp\left(\frac{\sigma^2}{2(1+\sigma^2)^2} \sum_{i=1}^d \left(\frac{x_i}{\sigma \sqrt{1+\sigma^2}} + U_i\right)^2\right)\right] \cr
 &\hspace{10mm}= \exp\left(-\frac{1}{2(1+\sigma^2)} \, \|\mathbf{x}\|^2\right) \exp\left(\frac{\frac{1}{2(1+\sigma^2)^3} \, \|\mathbf{x}\|^2}{1 - \frac{\sigma^2}{(1 + \sigma^2)^2}}\right) \left(1 - \frac{\sigma^2}{(1 + \sigma^2)^2}\right)^{-d/2},
 \end{align*}
using \autoref{lemma:dist_Y_ideal} and the explicit expression of the moment generating function of a non-central chi-squared distribution. Note that the latter can be used because $\sigma^2 / (1+\sigma^2)^2 < 1$. The first two terms on the RHS above can be combined and simplified; their product is equal to
\[
 \exp\left(-\|\mathbf{x}\|^2 \frac{\sigma^2}{2((1+\sigma^2)^2 - \sigma^2)}\right).
\]
\end{proof}

\begin{Proposition}\label{prop:upper_bound_acc2}
 Consider a current state $\bx$ and that $\bY \sim Q_{w, \sigma}(\bx, \cdot \,)$ with $q_\sigma(\bx, \cdot \,) = \mathcal{N}(\mathbf{x},$ $ \sigma^2 \I_d)$ and $w(\bx, \by) = \pi(\by)/\pi(\bx)$. Assume that $U := - \log \pi$ is continuously differentiable. Assume that $U$ is $L$-smooth, meaning that its gradient, $\nabla U$, is $L$-Lipschitz. Finally, assume that $\|\nabla U(\bx)\| \rightarrow \infty$ as $\|\bx\|  \rightarrow \infty$. Then, for any $\sigma$ and $d$, it holds that
 \[
  \lim_{\|\bx\| \rightarrow \infty}  \E_{\bx}[\alpha_{\text{ideal}}(\bx, \bY)] = 0.
 \]
\end{Proposition}

\begin{proof}[Proof of \autoref{prop:upper_bound_acc2}]
  Let us define $\Z(\bx) := \int \pi(\by) \, q_\sigma(\bx, \by) \,\d\by$. We first provide a lower bound on $\Z(\bx) / \pi(\bx)$ which will be useful to prove our result. To provide this lower bound, we will use that
  \[
   |U(\by) - U(\bx) - \nabla U(\bx)^T (\by - \bx)| \leq \frac{L}{2} \|\by - \bx\|^2,
  \]
  given that $U$ is $L$-smooth \citep[Lemma 3.4]{bubeck2015convex}. Using this, we have that
  \begin{align*}
   \frac{\Z(\bx)}{\pi(\bx)} &= \int \exp(-(U(\by) - U(\bx))) \, q_\sigma(\bx, \by) \,\d\by \cr
   &\geq \int \exp\left(-\nabla U(\bx)^T (\by - \bx) - \frac{L}{2} \|\by - \bx\|^2\right) \, q_\sigma(\bx, \by) \,\d\by \cr
   &= \int \exp\left(-\nabla U(\bx)^T \bz - \frac{L}{2} \|\bz\|^2\right) \, \frac{1}{(2\pi \sigma^2)^{d/2}} \exp\left(-\frac{1}{2\sigma^2} \|\bz\|^2\right) \,\d\bz \cr
   &= \left(\frac{\sigma^{-2}}{L + \sigma^{-2}}\right)^{d/2} \int \exp\left(-\nabla U(\bx)^T \bz\right) \, \frac{1}{(2\pi(L + \sigma^{-2}))^{d/2}} \exp\left(-\frac{L + \sigma^{-2}}{2} \|\bz\|^2\right) \,\d\bz \cr
   &= \left(\frac{\sigma^{-2}}{L + \sigma^{-2}}\right)^{d/2} \exp\left(\frac{1}{2(L + \sigma^{-2})}\|\nabla U(\bx)\|^2\right),
  \end{align*}
  where the third line follows from a change of variables $\bz = \by - \bx$, and the last line follows from the explicit expression of the moment generating function of a multivariate normal distribution.

  Now, we make use of that bound. When $\bY \sim Q_{w, \sigma}(\bx, \cdot \,)$ with $q_\sigma(\bx, \cdot \,) = \mathcal{N}(\mathbf{x}, \sigma^2 \I_d)$ and $w(\bx, \by) = \pi(\by)/\pi(\bx)$,
  \[
   \alpha_{\text{ideal}}(\bx, \by) = 1 \wedge \frac{\Z(\bx)}{\Z(\by)} \leq \frac{\Z(\bx)}{\Z(\by)}.
  \]
  Therefore, we have that
  \begin{align*}
   \E_{\bx}[\alpha_{\text{ideal}}(\bx, \bY)] &\leq \int \frac{\pi(\by) \, q_{\sigma}(\bx, \by)}{\Z(\by)} \, \d\by \cr
   &= \int \frac{\pi(\bx + \bz) \, q_{\sigma}(\mathbf{0}, \bz)}{\Z(\bx + \bz)} \, \d\bz \cr
   & \leq \int \left(\frac{\sigma^{-2}}{L + \sigma^{-2}}\right)^{-d/2} \exp\left(-\frac{1}{2(L + \sigma^{-2})}\|\nabla U(\bx + \bz)\|^2\right) \, q_{\sigma}(\mathbf{0}, \bz) \, \d\bz \
  \end{align*}
  after a change of variables $\bz = \by - \bx$. We conclude the proof using the bounded convergence theorem: for any $\sigma, d, L$ and $\bz$,
  \[
   \left(\frac{\sigma^{-2}}{L + \sigma^{-2}}\right)^{-d/2} \exp\left(-\frac{1}{2(L + \sigma^{-2})}\|\nabla U(\bx + \bz)\|^2\right) \rightarrow 0,
   \]
   as $\|\bx\| \rightarrow \infty$ and
   \[
    \left(\frac{\sigma^{-2}}{L + \sigma^{-2}}\right)^{-d/2} \exp\left(-\frac{1}{2(L + \sigma^{-2})}\|\nabla U(\bx + \bz)\|^2\right) \leq  \left(\frac{\sigma^{-2}}{L + \sigma^{-2}}\right)^{-d/2}.
   \]
  \end{proof}

\subsection{Proof of \autoref{theorem:scaling_limit_ideal}}

\begin{proof}[Proof of \autoref{theorem:scaling_limit_ideal}]
 We prove a weak convergence towards a diffusion, denoted by $\{Z(t): t \geq 0\}$, in the Skorokhod topology (for more details about this type of convergence, see Chapter 3 of \cite{ethier1986markov}). In order to prove the result, we demonstrate the convergence of the finite-dimensional distributions of $\{Z_{d, \text{ideal}}(t): t\geq0\}$ to those of $\{Z(t): t \geq 0\}$. To achieve this, we verify Condition (c) of Theorem 8.2 from Chapter 4 of \cite{ethier1986markov}. The weak convergence then follows from Corollary 8.6 of Chapter 4 of \cite{ethier1986markov}. The remaining conditions of Theorem 8.2 and the conditions of Corollary 8.6 are either straightforward or easily derived from the proof given here.

The proof of the convergence of the finite-dimensional distributions relies on the convergence of (what we call) the \textit{pseudo-generator} of $\{Z_{d, \text{ideal}}(t): t\geq0\}$, an operator that we now introduce. The proof follows.

\textbf{Pseudo-generator.} The process $\{Z_{d, \text{ideal}}(t): t\geq0\}$ is a jump process for which the time in between the (possible) jumps is deterministic: we know that every $1/d^{2\tau}$ unit of time, the process jumps if the proposal is accepted. The pseudo-generator is a discrete version of infinitesimal generators of stochastic processes. It is defined as follows:
\[
 \phi_{d, \text{ideal}}(t) := d^{2\tau} \E[h(Z_{d, \text{ideal}}(t + 1 / d^{2\tau})) - h(Z_{d, \text{ideal}}(t)) \mid \mathcal{F}_{\bZ_{d, \text{ideal}}}(t)],
\]
where $h$ is a test function and $\mathcal{F}_{\bZ_{d, \text{ideal}}}(t)$ is the natural filtration associated to $\{\bZ_{d, \text{ideal}}(t): t \geq 0\}$. The Markov property, the fact that $\bZ_{d, \text{ideal}}(0) \sim \pi_d$ and that $\{\bX_{d, \text{ideal}}(m): m \in \N\}$ is time-homogeneous imply that for any $t$,
\begin{align*}
 \phi_{d, \text{ideal}}(t) &= d^{2\tau} \E[h(Z_{d, \text{ideal}}(t + 1 / d^{2\tau})) - h(Z_{d, \text{ideal}}(t)) \mid \bZ_{d, \text{ideal}}(t)] \cr
 &\hspace{-1mm}\stackrel{\text{dist.}}{=} d^{2\tau} \E[(h(Y_{1}) - h(X_1)) \, \alpha_{\text{ideal}}(\mathbf{X}_d, \mathbf{Y}_d) \mid \mathbf{X}_d],
\end{align*}
where ``\, $\stackrel{\text{dist.}}{=}$ \,'' denotes an equality in distribution, $\mathbf{X}_d \sim \pi_d$ and $Y_{1}$ is the first coordinate of $\mathbf{Y}_{d} \sim Q_{w, \sigma}(\mathbf{X}_d, \cdot \,)$.

We prove the convergence of $\phi_{d, \text{ideal}}(t)$ towards $Gh(Z_{d, \text{ideal}}(t))$ in some sense, where $G$ is the generator of the diffusion. The form of $G$ allows to restrict our attention to test functions $h \in \mathcal{C}_c^\infty(\R)$, the space of infinitely differentiable functions on $\R$ with compact support (Theorem 2.5 from Chapter 8 of \cite{ethier1986markov}).

\textbf{Proof of the convergence of the finite-dimensional distributions.} Condition (c) of Theorem 8.2 from Chapter 4 of \cite{ethier1986markov} essentially reduces to the following convergence:
\[
 \E|\phi_{d, \text{ideal}}(t) - Gh(Z_{d, \text{ideal}}(t))| \rightarrow 0 \quad \text{as} \quad d \rightarrow \infty,
\]
for all $t$. The generator is such that
\begin{align*}
 Gh(Z_{d, \text{ideal}}(t)) &= \ell^2 (\vartheta_{w, \tau}(\ell) / 2)(\log \varphi(Z_{d, \text{ideal}}(t)))' h'(Z_{d, \text{ideal}}(t)) + \ell^2 (\vartheta_{w, \tau}(\ell) / 2) h''(Z_{d, \text{ideal}}(t)) \cr
 &\hspace{-1mm}\stackrel{\text{dist.}}{=} \ell^2 (\vartheta_{w, \tau}(\ell) / 2)(\log \varphi(X_1))' h'(X_1) + \ell^2 (\vartheta_{w, \tau}(\ell) / 2) h''(X_1),
\end{align*}
where the equality in distribution follows from the fact that the process starts in stationarity, that is $\bZ_{d, \text{ideal}}(0) \sim \pi_d$. We can thus see $\phi_{d, \text{ideal}}(t) - Gh(Z_{d, \text{ideal}}(t))$ in the expectation above as a difference of two functions of $\mathbf{X}_d\sim\pi_d$, and will write $Gh(X_1)$ instead of $Gh(Z_{d, \text{ideal}}(t))$ in the expectation. Note that the form of the generator indicates that $\varphi$ is the unique invariant PDF of the diffusion.

 We prove that
 \[
  \E|d^{2\tau} \E[(h(Y_{1}) - h(X_1)) \, \alpha_{\text{ideal}}(\mathbf{X}_d, \mathbf{Y}_d) \mid \mathbf{X}_d] - Gh(X_1)| \rightarrow 0.
  \]
  The key here is to use a Taylor expansion in $(h(Y_{1}) - h(X_1)) \, \alpha_{\text{ideal}}(\mathbf{X}_d, \mathbf{Y}_d)$ to obtain derivatives of $h$ as in $Gh(X_1)$. Specifically, we write
 \begin{align*}
  h(Y_1) - h(X_1) = h'(X_1)(Y_1 - X_1) + h''(X_1)\frac{(Y_1 - X_1)^2}{2} + h'''(W)\frac{(Y_1 - X_1)^3}{6},
 \end{align*}
 where $W$ belongs to $(X_1, Y_1)$ or $(Y_1, X_1)$ (depending which one of $X_1$ and $Y_1$ is smaller). Therefore, using the triangle inequality,
 \begin{align}\label{eqn:proof_thm_2}
  &\E|d^{2\tau} \E[(h(Y_{1}) - h(X_1)) \, \alpha_{\text{ideal}}(\mathbf{X}_d, \mathbf{Y}_d) \mid \mathbf{X}_d] - Gh(X_1)|   \cr
  &\quad= \E|d^{2\tau} \E[h'(X_1)(Y_1 - X_1) \, \alpha_{\text{ideal}}(\mathbf{X}_d, \mathbf{Y}_d) \mid \mathbf{X}_d] - \ell^2 (\vartheta_{w, \tau}(\ell) / 2)(\log \varphi(X_1))' h'(X_1)| \cr
  &\qquad + \E\left|d^{2\tau} \E\left[h''(X_1)\frac{(Y_1 - X_1)^2}{2}  \, \alpha_{\text{ideal}}(\mathbf{X}_d, \mathbf{Y}_d) \mid \mathbf{X}_d\right] - \ell^2 (\vartheta_{w, \tau}(\ell) / 2) h''(X_1)\right| \cr
  &\qquad+\E\left|d^{2\tau} \E\left[h'''(W)\frac{(Y_1 - X_1)^3}{6} \, \alpha_{\text{ideal}}(\mathbf{X}_d, \mathbf{Y}_d) \mid \mathbf{X}_d\right]\right|.
 \end{align}
 We now prove that each term on the RHS converges to 0. We prove this for the case $w(\mathbf{x}_d, \mathbf{y}_d) = \pi(\mathbf{y}_d) / \pi(\mathbf{x}_d)$; the case $w(\mathbf{x}_d, \mathbf{y}_d) = \sqrt{\pi(\mathbf{y}_d) / \pi(\mathbf{x}_d)}$ is proved similarly.

We have that
\begin{align*}
 &\E\left|d^{2\tau} \E\left[h'''(W)\frac{(Y_1 - X_1)^3}{6}\, \alpha_{\text{ideal}}(\mathbf{X}_d, \mathbf{Y}_d) \mid \mathbf{X}_d\right]\right| \leq \frac{M}{6} d^{2\tau} \E[\E[|Y_1 - X_1|^3 \mid \mathbf{X}_d]] \cr
 &\leq \frac{M}{6} d^{2\tau}\left( \E\left|\frac{\sigma_d^2X_1}{1+\sigma_d^2}\right|^3 + 3\E\left[\left(\frac{\sigma_d^2X_1}{1+\sigma_d^2}\right)^2\right]\sqrt{\frac{\sigma_d^2}{1 + \sigma_d^2}}\E|U_1| + 3\E\left[\left|\frac{\sigma_d^2X_1}{1+\sigma_d^2}\right|\right]\frac{\sigma_d^2}{1 + \sigma_d^2}\E[U_1^2] \right. \cr
 &\hspace{90mm}\left. + \left(\frac{\sigma_d^2}{1+\sigma_d^2}\right)^{3/2}\E|U_1|^3 \right),
\end{align*}
using Jensen's inequality, that $0\leq \alpha_{\text{ideal}} \leq 1$, that there exists a positive constant $M$ such that $|h'''|\leq M$, \autoref{lemma:dist_Y_ideal} and the triangle inequality. The random variables $X_1$ and $U_1$ are independent and both follow a standard normal distribution, implying that $\E|X_1|^p\E|U_1|^q$ is finite and independent of $d$ for any $p$ and $q$. Recall that $\sigma_d = \ell / d^\tau$. The sum above thus converges to 0.

 For the other terms in \eqref{eqn:proof_thm_2}, we view the function $\alpha_{\text{ideal}}(\mathbf{x}_d, \mathbf{y}_d)$ (for any realization of $\bX_d$ and $\bY_d$) as a function of $y_1$ (while keeping the other variables fixed) and we use a Taylor expansion around $x_1$ to obtain a function independent of $x_1$ and $y_1$. To see this, we recall that
 \begin{align*}
  \alpha_{\text{ideal}}(\mathbf{x}_d, \mathbf{y}_d) &= 1 \wedge \exp\left(\frac{1}{2(1+\sigma_d^2)} \sum_{i=1}^d (y_i^2 - x_i^2)\right),
 \end{align*}
using \autoref{lemma:acc_prob}. We thus write
\begin{align*}
 \alpha_{\text{ideal}}(\mathbf{x}_d, \mathbf{y}_d) = \alpha_{\text{ideal}}(\mathbf{x}_d, \mathbf{y}_d^*) &+ \left(\left.\frac{\partial}{\partial y_1}  \alpha_{\text{ideal}}(\mathbf{x}_d, \mathbf{y}_d)\right|_{y_1 = x_1}\right)(y_1 - x_1) \cr
 &\quad+  \left(\left.\frac{\partial^2}{\partial y_1^2}  \alpha_{\text{ideal}}(\mathbf{x}_d, \mathbf{y}_d)\right|_{y_1 = w}\right)\frac{(y_1 - x_1)^2}{2},
\end{align*}
where $\mathbf{y}_d^* := (x_1, y_2, \ldots, y_d)$ and $w$ belongs to $(x_1, y_1)$ or $(y_1, x_1)$. We have that
\[
 \alpha_{\text{ideal}}(\mathbf{x}_d, \mathbf{y}_d^*) = 1 \wedge \exp\left(\frac{1}{2(1+\sigma_d^2)} \sum_{i=2}^d (y_i^2 - x_i^2)\right),
\]
\begin{align}\label{eqn_partial_deriv_alpha}
 \left.\frac{\partial}{\partial y_1}  \alpha_{\text{ideal}}(\mathbf{x}_d, \mathbf{y}_d)\right|_{y_1 = x_1} = \frac{x_1}{1+\sigma_d^2} \exp\left(\frac{1}{2(1+\sigma_d^2)} \sum_{i=2}^d (y_i^2 - x_i^2)\right) \1\left(\sum_{i=2}^d (y_i^2 - x_i^2) < 0\right),
\end{align}
\begin{align*}
& \left.\frac{\partial^2}{\partial y_1^2}  \alpha_{\text{ideal}}(\mathbf{x}_d, \mathbf{y}_d)\right|_{y_1 = w} = \left(\frac{1}{1+\sigma_d^2} \exp\left(\frac{1}{2(1+\sigma_d^2)} \sum_{i=2}^d (y_i^2 - x_i^2) + w^2 - x_1^2\right) \right. \cr
 &+\left.\frac{w^2}{1+\sigma_d^2} \exp\left(\frac{1}{2(1+\sigma_d^2)} \sum_{i=2}^d (y_i^2 - x_i^2) + w^2 - x_1^2\right)\right)\1\left(\sum_{i=2}^d (y_i^2 - x_i^2) + w^2 - x_1^2 < 0\right).
\end{align*}
We replace $\alpha_{\text{ideal}}(\mathbf{X}_d, \mathbf{Y}_d)$ in
\[
 \E|d^{2\tau} \E[h'(X_1)(Y_1 - X_1) \, \alpha_{\text{ideal}}(\mathbf{X}_d, \mathbf{Y}_d) \mid \mathbf{X}_d] - \ell^2 (\vartheta_{w, \tau}(\ell) / 2)(\log \varphi(X_1))' h'(X_1)|
 \]
  in \eqref{eqn:proof_thm_2} by the sum above. We want to prove that the expectation converges to 0. Using the triangle inequality and that $M$ can be chosen such that $|h'|\leq M$, it is sufficient to show that
\begin{align*}
 &M\E\left|d^{2\tau} \E\left[(Y_1 - X_1)\left(1 \wedge \exp\left(\frac{1}{2(1+\sigma_d^2)} \sum_{i=2}^d (Y_i^2 - X_i^2)\right)\right)\right.\right.  \cr
 &\quad+ \left. (Y_1 - X_1)^2 \frac{X_1}{1+\sigma_d^2} \exp\left(\frac{1}{2(1+\sigma_d^2)} \sum_{i=2}^d (Y_i^2 - X_i^2)\right) \1\left(\sum_{i=2}^d (Y_i^2 - X_i^2) < 0\right)  \mid \mathbf{X}_d\right] \cr
 &\qquad\left. - \ell^2 (\vartheta_{w, \tau}(\ell) / 2)(\log \varphi(X_1))' \right.| \rightarrow 0,
\end{align*}
and that the following expectation converges to 0:
\small
\begin{align*}
 &Md^{2\tau}\E\left|\frac{(Y_1 - X_1)^3}{2}\left(\frac{1}{1+\sigma_d^2} \exp\left(\frac{1}{2(1+\sigma_d^2)} \sum_{i=2}^d (Y_i^2 - X_i^2) + W^2 - X_1^2\right) \right.\right. \cr
 &+\left.\left.\frac{W^2}{1+\sigma_d^2} \exp\left(\frac{1}{2(1+\sigma_d^2)} \sum_{i=2}^d (Y_i^2 - X_i^2) + W^2 - X_1^2\right)\right)\1\left(\sum_{i=2}^d (Y_i^2 - X_i^2) + W^2 - X_1^2 < 0\right)\right|.
\end{align*}
\normalsize
We start with the last term. Using the triangle inequality, it is lesser than or equal to
\begin{align*}
 \frac{Md^{2\tau}}{2(1+\sigma_d^2)}\E|Y_1 - X_1|^3 + \frac{Md^{2\tau}}{2(1+\sigma_d^2)}\E[|Y_1 - X_1|^3W^2].
\end{align*}
We have seen before that $d^{2\tau}\E|Y_1 - X_1|^3 \rightarrow 0$. Also, given $X_1$ and writing $Y_1 = \frac{X_1}{1+\sigma_d^2}+\sqrt{\frac{\sigma_d^2}{1+\sigma_d^2}}U_1$ under the conditional expectation, we can show that $|W| \leq |X_1| + \sqrt{\frac{\sigma_d^2}{1+\sigma_d^2}}|U_1|$, and consequently, that $d^{2\tau}\E[|Y_1 - X_1|^3W^2] \rightarrow 0$ in the same way we proved that $d^{2\tau}\E|Y_1 - X_1|^3 \rightarrow 0$.

For the other term, we first note that $(\log \varphi(X_1))' = -X_1$. We simplify the notation by defining
\[
 f_1(\mathbf{X}_d) := \E\left[1 \wedge \exp\left(\frac{1}{2(1+\sigma_d^2)} \sum_{i=2}^d (Y_i^2 - X_i^2)\right) \mid \mathbf{X}_d\right],
\]
and
\[
 f_2(\mathbf{X}_d) := \E\left[\exp\left(\frac{1}{2(1+\sigma_d^2)} \sum_{i=2}^d (Y_i^2 - X_i^2)\right) \1\left(\sum_{i=2}^d (Y_i^2 - X_i^2) < 0\right)\mid \mathbf{X}_d\right].
\]
Using the conditional independence among $Y_1, \ldots, Y_d$ given $\mathbf{X}_d$,
\begin{align*}
 &\E\left[(Y_1 - X_1)\left(1 \wedge \exp\left(\frac{1}{2(1+\sigma_d^2)} \sum_{i=2}^d (Y_i^2 - X_i^2)\right)\right)\right.  \cr
 &\quad+ \left. (Y_1 - X_1)^2 \frac{X_1}{1+\sigma_d^2} \exp\left(\frac{1}{2(1+\sigma_d^2)} \sum_{i=2}^d (Y_i^2 - X_i^2)\right) \1\left(\sum_{i=2}^d (Y_i^2 - X_i^2) < 0\right)  \mid \mathbf{X}_d\right] \cr
 &= \E\left[Y_1 - X_1 \mid \mathbf{X}_d\right]f_1(\mathbf{X}_d) + \frac{X_1}{1+\sigma_d^2}\E\left[(Y_1 - X_1)^2 \mid \mathbf{X}_d\right] f_2(\mathbf{X}_d) \cr
 &= \E\left[-\frac{\sigma_d^2 X_1}{1+\sigma_d^2}+\sqrt{\frac{\sigma_d^2}{1+\sigma_d^2}}U_1 \mid \mathbf{X}_d\right]f_1(\mathbf{X}_d) \cr
  &\qquad + \frac{X_1}{1+\sigma_d^2}\E\left[\frac{\sigma_d^4 X_1^2}{(1+\sigma_d^2)^2}-2\frac{\sigma_d^2 X_1}{1+\sigma_d^2}\sqrt{\frac{\sigma_d^2}{1+\sigma_d^2}}U_1+\frac{\sigma_d^2}{1+\sigma_d^2}U_1^2 \mid \mathbf{X}_d\right] f_2(\mathbf{X}_d)\cr
 &=-\frac{\sigma_d^2 X_1}{1+\sigma_d^2}f_1(\mathbf{X}_d)+\frac{\sigma_d^4 X_1^3}{(1+\sigma_d^2)^3}f_2(\mathbf{X}_d) + \frac{\sigma_d^2X_1}{(1+\sigma_d^2)^2} f_2(\mathbf{X}_d) \cr
 &=-\sigma_d^2 X_1\left(\frac{f_1(\mathbf{X}_d)}{1+\sigma_d^2} - \frac{f_2(\mathbf{X}_d)}{(1+\sigma_d^2)^2}\right) + \frac{\sigma_d^4 X_1^3}{(1+\sigma_d^2)^3}f_2(\mathbf{X}_d).
\end{align*}
Therefore,
\begin{align*}
 &\E\left|d^{2\tau} \E\left[(Y_1 - X_1)\left(1 \wedge \exp\left(\frac{1}{2(1+\sigma_d^2)} \sum_{i=2}^d (Y_i^2 - X_i^2)\right)\right)\right.\right.  \cr
 &\quad+ \left. (Y_1 - X_1)^2 \frac{X_1}{1+\sigma_d^2} \exp\left(\frac{1}{2(1+\sigma_d^2)} \sum_{i=2}^d (Y_i^2 - X_i^2)\right) \1\left(\sum_{i=2}^d (Y_i^2 - X_i^2) < 0\right)  \mid \mathbf{X}_d\right] \cr
 &\qquad\left. - \ell^2 (\vartheta_{w, \tau}(\ell) / 2)(\log \varphi(X_1))'  \right.| \cr
 &\leq \E\left|-X_1\ell^2\left(\left(\frac{f_1(\mathbf{X}_d)}{1+\sigma_d^2} - \frac{f_2(\mathbf{X}_d)}{(1+\sigma_d^2)^2}\right) - \frac{\vartheta_{w, \tau}(\ell)}{2}\right) + \frac{\sigma_d^4 X_1^3}{(1+\sigma_d^2)^3}f_2(\mathbf{X}_d)\right| \cr
 &\leq \E\left|-X_1\ell^2\left(\left(\frac{f_1(\mathbf{X}_d)}{1+\sigma_d^2} - \frac{f_2(\mathbf{X}_d)}{(1+\sigma_d^2)^2}\right) - \frac{\vartheta_{w, \tau}(\ell)}{2}\right)\right| + d^{2\tau}\E\left|\frac{\sigma_d^4 X_1^3}{(1+\sigma_d^2)^3}\right|,
\end{align*}
using the triangle inequality and that $0\leq f_2(\mathbf{X}_d) \leq 1$. As previously,
\[
 d^{2\tau}\E\left|\frac{\sigma_d^4 X_1^3}{(1+\sigma_d^2)^3}\right| \rightarrow 0.
\]
We also have that
\[
 \frac{f_1(\mathbf{X}_d)}{1+\sigma_d^2} - \frac{f_2(\mathbf{X}_d)}{(1+\sigma_d^2)^2} =  \frac{\E[\1\left(\sum_{i=2}^d (Y_i^2 - X_i^2) \geq 0\right)\mid \mathbf{X}_d]}{(1+\sigma_d^2)^2} + \frac{\sigma_d^2f_1(\mathbf{X}_d)}{(1+\sigma_d^2)^2}.
\]
Using that $0\leq f_1(\mathbf{X}_d) \leq 1$, the triangle inequality and because
\[
 \frac{\ell^2\sigma_d^2}{(1+\sigma_d^2)^2}\E|X_1|\rightarrow 0,
\]
we can now focus on
\[
 \E\left|-X_1\ell^2\left(\frac{\E\left[\1\left(\sum_{i=2}^d (Y_i^2 - X_i^2) \geq 0\right)\mid \mathbf{X}_d\right]}{(1+\sigma_d^2)^2} - \frac{\vartheta_{w, \tau}(\ell)}{2}\right)\right|.
\]
We have that $Y_i^2 - X_i^2 = (Y_i - X_i)(Y_i +X_i)$, and as previously, we use that given $\mathbf{X}_d$, we can write $Y_i = \frac{X_i}{1 + \sigma_d^2} + \sqrt{\frac{\sigma_d^2}{1 + \sigma_d^2}} U_i$, and consequently,
\begin{align*}
 Y_i - X_i &= -\frac{\sigma_d^2}{1 + \sigma_d^2} X_i + \frac{1}{\sqrt{1 + \sigma_d^2}} \sigma_d U_i, \cr
 Y_i + X_i &= \frac{2 + \sigma_d^2}{1 + \sigma_d^2} X_i + \frac{1}{\sqrt{1 + \sigma_d^2}} \sigma_d U_i,
\end{align*}
and
\begin{align*}
 (Y_i - X_i)(Y_i + X_i) &= -\frac{2 + \sigma_d^2}{1 + \sigma_d^2} \sigma_d^2 X_i^2 - \frac{1}{(1 + \sigma_d^2)^{3/2}} \sigma_d^3 X_i U_i \cr
  &\qquad + \frac{(2 + \sigma_d^2)}{(1 + \sigma_d^2)^{3/2}} \sigma_d U_i X_i + \frac{1}{1 + \sigma_d^2} \sigma_d^2 U_i^2.
\end{align*}
We define $S_d := \sum_{i=2}^d Y_i^2 - X_i^2$ and $W_d := -\frac{2 + \sigma_d^2}{1 + \sigma_d^2} \sigma_d^2 \sum_{i=2}^dX_i^2 + \frac{(2 + \sigma_d^2)}{(1 + \sigma_d^2)^{3/2}} \sigma_d \sum_{i=2}^d U_i X_i +\frac{ \ell^2d}{(1+\sigma_d^2)(d-1)}$. If $\tau = 1/2$, in $S_d$, we notice that $\frac{1}{1 + \sigma_d^2} \sigma_d^2\sum_{i=2}^d  U_i^2 \rightarrow \ell^2$ with probability 1 as a result of the strong law of large numbers. Also, $\frac{1}{(1 + \sigma_d^2)^{3/2}} \sigma_d^3  \sum_{i=2}^d X_i U_i \rightarrow 0$ with probability 1 for the same reason. Therefore, given $\mathbf{X}_d$, $S_d$ follows essentially a normal distribution with mean $-\frac{2 + \sigma_d^2}{1 + \sigma_d^2} \sum_{i=2}^d\sigma_d^2 X_i^2 + \ell^2$ and variance $\frac{(2 + \sigma_d^2)^2}{(1 + \sigma_d^2)^{3}} \sigma_d^2 \sum_{i=2}^d X_i^2$. We want to use this and that is why we will prove that $S_d$ and $W_d$ are asymptotically equivalent; $W_d$ has a conditional normal distribution. We have an explicit expression for $\E\left[\1\left(W_d \geq 0\right)\mid \mathbf{X}_d\right]$ and we can use it.

For the rest of the proof, we consider that $\tau = 1/2$. If $\tau > 1/2$, we can use the same strategy as below, but with $W_d := -\frac{2 + \sigma_d^2}{1 + \sigma_d^2} \sigma_d^2 \sum_{i=2}^dX_i^2 + \frac{(2 + \sigma_d^2)}{(1 + \sigma_d^2)^{3/2}} \sigma_d \sum_{i=2}^d U_i X_i$ because $\frac{1}{1 + \sigma_d^2} \sigma_d^2\sum_{i=2}^d  U_i^2 \rightarrow 0$. In this case, $W_d$ has a conditional normal distribution whose mean is $ -\frac{2 + \sigma_d^2}{1 + \sigma_d^2} \sigma_d^2 \sum_{i=2}^dX_i^2$ and variance $\frac{(2 + \sigma_d^2)^2}{(1 + \sigma_d^2)^{3}} \sigma_d^2 \sum_{i=2}^d X_i^2$. Both converge to 0 with probability 1, but the mean converges quicker than the standard deviation, implying that the limit of the explicit expression for $\E\left[\1\left(W_d \geq 0\right)\mid \mathbf{X}_d\right]$ is $\Phi(0) = 1/2$, which allows to conclude.

 Let us now return to the case $\tau = 1/2$. Using the triangle inequality,
\begin{align}\label{eqn:proof_thm_2_2}
  \E\left|-X_1\ell^2\left(\frac{\E\left[\1\left(S_d \geq 0\right)\mid \mathbf{X}_d\right]}{(1+\sigma_d^2)^2} - \frac{\vartheta_{w, \tau}(\ell)}{2}\right)\right| &\leq \E\left|-X_1\ell^2\left(\frac{\E\left[\1\left(W_d \geq 0\right)\mid \mathbf{X}_d\right]}{(1+\sigma_d^2)^2} - \frac{\vartheta_{w, \tau}(\ell)}{2}\right)\right| \cr
  &\hspace{-20mm}+ \E\left|-\frac{X_1\ell^2}{(1+\sigma_d^2)^2}\left(\E\left[\1\left(W_d \geq 0\right)\mid \mathbf{X}\right] - \E\left[\1\left(S_d \geq 0\right)\mid \mathbf{X}\right]\right)\right|.
\end{align}
We now prove that the last expectation converges to 0. Using the Cauchy–Schwarz inequality and Jensen's inequality, it is sufficient to prove that $\E[(\1\left(W_d \geq 0\right) - \1\left(S_d \geq 0\right))^2]\rightarrow 0$ given that
\[
 \frac{\ell^2}{(1+\sigma_d^2)^2} \E[X_1^2]^{1/2} = \frac{\ell^2}{(1+\sigma_d^2)^2} \rightarrow \ell^2.
\]
We have that
\begin{align*}
 \E[(\1\left(W_d \geq 0\right) - \1\left(S_d \geq 0\right))^2] &=\P(W_d \geq 0, S_d < 0) + \P(W_d < 0, S_d \geq 0) \cr
 &= \P(W_d \geq 0, S_d < 0, |W_d - S_d| > d^{-1/4}) \cr
 &\quad + \P(W_d \geq 0, S_d < 0, |W_d - S_d| \leq  d^{-1/4}) \cr
 &\quad + \P(W_d < 0, S_d \geq 0, |W_d - S_d| > d^{-1/4})  \cr
 &\quad + \P(W_d < 0, S_d \geq 0, |W_d - S_d| \leq d^{-1/4})  \cr
 &\leq 2 \P(|W_d - S_d| > d^{-1/4}) \cr
 &\quad + \P(W_d \geq 0, S_d < 0, W_d - S_d \leq  d^{-1/4}) \cr
 &\quad + \P(W_d < 0, S_d \geq 0, S_d - W_d \leq d^{-1/4}) \cr
 &\leq 2 \P(|W_d - S_d| > d^{-1/4}) + \P(-d^{-1/4} \leq W_d \leq d^{-1/4}).
\end{align*}
Using Markov's inequality,
\begin{align*}
 \P(|W_d - S_d| > d^{-1/4}) &\leq \frac{\E|W_d - S_d|}{d^{-1/4}} \cr &= \frac{\E\left|\frac{ \ell^2d}{(1+\sigma_d^2)(d-1)} - \frac{1}{1 + \sigma_d^2} \sigma_d^2 \sum_{i=2}^dU_i^2 + \frac{1}{(1 + \sigma_d^2)^{3/2}} \sigma_d^3 \sum_{i=2}^d X_i U_i\right|}{d^{-1/4}}
\end{align*}
Also,
\begin{align*}
 &d^{1/4}\E\left|\frac{ \ell^2d}{(1+\sigma_d^2)(d-1)} - \frac{1}{1 + \sigma_d^2} \sigma_d^2 \sum_{i=2}^dU_i^2\right| \cr
 &\qquad\leq d^{1/4}\E\left[\left(\frac{ \ell^2d}{(1+\sigma_d^2)(d-1)} - \frac{1}{1 + \sigma_d^2} \sigma_d^2 \sum_{i=2}^dU_i^2\right)^2\right]^{1/2} \cr
 &\qquad=d^{1/4}\var\left[\frac{1}{1 + \sigma_d^2} \sigma_d^2 \sum_{i=2}^dU_i^2\right]^{1/2} = \frac{\ell^2d^{1/4}\sqrt{d-1}}{(1 + \sigma_d^2)d} \rightarrow 0,
\end{align*}
and
\begin{align*}
 d^{1/4}\E\left|\frac{1}{(1 + \sigma_d^2)^{3/2}} \sigma_d^3 \sum_{i=2}^d X_i U_i\right| \leq \frac{\ell^3(d-1)d^{1/4}}{(1 + \sigma_d^2)^{3/2}d^{3/2}}\frac{2}{\pi}\rightarrow 0.
\end{align*}
To compute $\P(-d^{-1/4} \leq W_d \leq d^{-1/4})$, we use that given $\mathbf{X}_d$,
\[
  W_d \sim \mathcal{N}\left( -\frac{2 + \sigma_d^2}{1 + \sigma_d^2} \sigma_d^2 \sum_{i=2}^dX_i^2 + \frac{ \ell^2d}{(1+\sigma_d^2)(d-1)}, \frac{(2 + \sigma_d^2)^2}{(1 + \sigma_d^2)^{3}} \sigma_d^2 \sum_{i=2}^d X_i^2 \right).
\]
Therefore,
\begin{align*}
 \P(-d^{-1/4} \leq W_d \leq d^{-1/4}) &= \E\left[\Phi\left(\frac{d^{-1/4}  +\frac{2 + \sigma_d^2}{1 + \sigma_d^2} \sigma_d^2 \sum_{i=2}^dX_i^2 - \frac{ \ell^2d}{(1+\sigma_d^2)(d-1)}}{\sqrt{\frac{(2 + \sigma_d^2)^2}{(1 + \sigma_d^2)^{3}} \sigma_d^2 \sum_{i=2}^d X_i^2}}\right)\right] \cr
 &\quad - \E\left[\Phi\left(\frac{-d^{-1/4}  +\frac{2 + \sigma_d^2}{1 + \sigma_d^2} \sigma_d^2 \sum_{i=2}^dX_i^2 - \frac{ \ell^2d}{(1+\sigma_d^2)(d-1)}}{\sqrt{\frac{(2 + \sigma_d^2)^2}{(1 + \sigma_d^2)^{3}} \sigma_d^2 \sum_{i=2}^d X_i^2}}\right)\right].
\end{align*}
As $d \rightarrow \infty$,
\[
 \frac{\pm d^{-1/4}  +\frac{2 + \sigma_d^2}{1 + \sigma_d^2} \sigma_d^2 \sum_{i=2}^dX_i^2 - \frac{ \ell^2d}{(1+\sigma_d^2)(d-1)}}{\sqrt{\frac{(2 + \sigma_d^2)^2}{(1 + \sigma_d^2)^{3}} \sigma_d^2 \sum_{i=2}^d X_i^2}} \rightarrow \frac{\ell}{2} \quad \text{with probability 1.}
\]
So, using Lebesgue's dominated convergence theorem, we know that $\P(-d^{-1/4} \leq W_d \leq d^{-1/4})\rightarrow 0$.

We return to the other term in \eqref{eqn:proof_thm_2_2}:
\begin{align*}
& \E\left|-X_1\ell^2\left(\frac{\E\left[\1\left(W_d \geq 0\right)\mid \mathbf{X}_d\right]}{(1+\sigma_d^2)^2} - \frac{\vartheta_{w, \tau}(\ell)}{2}\right)\right| \cr
&\qquad  \leq \E[X_1^2]^{1/2} \ell^2 \E\left[\left(\frac{\E\left[\1\left(-W_d \leq 0\right)\mid \mathbf{X}_d\right]}{(1+\sigma_d^2)^2} - \frac{\vartheta_{w, \tau}(\ell)}{2}\right)^2\right]^{1/2},
\end{align*}
using the Cauchy–Schwarz inequality. We have that
\[
 \E\left[\1\left(-W_d \leq 0\right)\mid \mathbf{X}_d\right] = \Phi\left( \frac{-\frac{2 + \sigma_d^2}{1 + \sigma_d^2} \sigma_d^2 \sum_{i=2}^dX_i^2 + \frac{ \ell^2d}{(1+\sigma_d^2)(d-1)}}{\sqrt{\frac{(2 + \sigma_d^2)^2}{(1 + \sigma_d^2)^{3}} \sigma_d^2 \sum_{i=2}^d X_i^2}}\right)\rightarrow \Phi\left(-\frac{\ell}{2}\right),
\]
with probability 1. Therefore,
\[
 \E\left[\left(\frac{\E\left[\1\left(-W_d \leq 0\right)\mid \mathbf{X}_d\right]}{(1+\sigma_d^2)^2} - \frac{\vartheta_{w, \tau}(\ell)}{2}\right)^2\right]^{1/2} \rightarrow 0,
\]
using Lebesgue's dominated convergence theorem, which concludes this part given that $\E[X_1^2] = 1$.

There remains to prove that
\[
 \E\left|d^{2\tau} \E\left[h''(X_1)\frac{(Y_1 - X_1)^2}{2} \, \alpha_{\text{ideal}}(\mathbf{X}_d, \mathbf{Y}_d) \mid \mathbf{X}_d\right] - \ell^2 \frac{\vartheta_{w, \tau}(\ell)}{2} h''(X_1)\right| \rightarrow 0,
\]
in \eqref{eqn:proof_thm_2}. We proceed as before with a Taylor expansion around $x_1$ of $\alpha_{\text{ideal}}(\mathbf{x}_d, \mathbf{y}_d)$, viewed as function of $y_1$. This time it is less complicated because we write
\[
 \alpha_{\text{ideal}}(\mathbf{x}_d, \mathbf{y}_d) = \alpha_{\text{ideal}}(\mathbf{x}_d, \mathbf{y}_d^*) + \left(\left.\frac{\partial}{\partial y_1}  \alpha_{\text{ideal}}(\mathbf{x}_d, \mathbf{y}_d)\right|_{y_1 = w}\right)(y_1 - x_1).
\]
Using that $M$ can be chosen such that $|h''|\leq M$, the triangle inequality, and that $0\leq \exp(x)\1(x < 0) \leq 1$ (see the partial derivative $\frac{\partial}{\partial y_1}  \alpha_{\text{ideal}}(\mathbf{x}_d, \mathbf{y}_d) $ \eqref{eqn_partial_deriv_alpha}),
\begin{align*}
 &\E\left|d^{2\tau} \E\left[h''(X_1)\frac{(Y_1 - X_1)^2}{2} \, \alpha_{\text{ideal}}(\mathbf{X}_d, \mathbf{Y}_d) \mid \mathbf{X}_d\right] - \ell^2 \frac{\vartheta_{w, \tau}(\ell)}{2} h''(X_1)\right| \cr
  &\leq M \E\left|d^{2\tau} \E\left[\frac{(Y_1 - X_1)^2}{2} \mid \mathbf{X}_d\right]f_1(\mathbf{X}_d) - \ell^2 \frac{\vartheta_{w, \tau}(\ell)}{2} \right| + \frac{d^{2\tau}}{2(1+\sigma_d^2)} \E[|W||Y_1 - X_1|^3].
\end{align*}
From what we have seen before, we know that
\[
 \frac{d^{2\tau}}{2(1+\sigma_d^2)}\E[|W||Y_1 - X_1|^3] \rightarrow 0.
\]
We also know that
\[
 \E\left[\frac{(Y_1 - X_1)^2}{2} \mid \mathbf{X}_d\right] = \E\left[\frac{\sigma_d^4 X_1^2}{2(1+\sigma_d^2)^2}-\frac{\sigma_d^2 X_1}{1+\sigma_d^2}\sqrt{\frac{\sigma_d^2}{1+\sigma_d^2}}U_1+\frac{\sigma_d^2}{2(1+\sigma_d^2)}U_1^2 \mid \mathbf{X}_d\right].
\]
Therefore, using the triangle inequality and that $0\leq f_1(\mathbf{X}_d) \leq 1$,
\begin{align*}
  \E\left|d^{2\tau} \E\left[\frac{(Y_1 - X_1)^2}{2} \mid \mathbf{X}_d\right]f_1(\mathbf{X}_d) - \ell^2 \frac{\vartheta_{w, \tau}(\ell)}{2} \right| &\leq \ell^2\E\left|\frac{1}{2(1+\sigma_d^2)} f_1(\mathbf{X}_d) - \frac{\vartheta_{w, \tau}(\ell)}{2} \right| \cr
  &\quad + d^{2\tau}\E\left[\frac{\sigma_d^4 X_1^2}{2(1+\sigma_d^2)^2}\right].
\end{align*}
We have that
\[
  d^{2\tau}\E\left[\frac{\sigma_d^4 X_1^2}{2(1+\sigma_d^2)^2}\right]= \frac{\ell^4}{2d(1+\sigma_d^2)^2} \rightarrow 0.
\]
To conclude the proof, there thus remains to show that
\[
 \E\left|\frac{1}{1+\sigma_d^2} \, \E\left[1 \wedge \exp\left(\frac{1}{2(1+\sigma_d^2)} \sum_{i=2}^d (Y_i^2 - X_i^2)\right) \mid \mathbf{X}_d\right] - \vartheta_{w, \tau}(\ell) \right| \rightarrow 0.
\]
We proceed similarly as before when we proved that
 \[
  \E\left|\E\left[\1\left(\sum_{i=2}^d (Y_i^2 - X_i^2) \geq 0\right) \mid \mathbf{X}_d\right] - \Phi\left(-\frac{\ell}{2}\right) \right| \longrightarrow 0,
 \]
 when $\tau = 1/2$.

 Using the triangle inequality,
 \begin{align*}
  &\E\left|\frac{1}{1+\sigma_d^2} \, \E\left[1 \wedge \exp\left(\frac{1}{2(1+\sigma_d^2)} \sum_{i=2}^d (Y_i^2 - X_i^2)\right) \mid \mathbf{X}_d\right] - \vartheta_{w, \tau}(\ell) \right| \cr
  &\quad\leq  \E\left|\frac{1}{1+\sigma_d^2} \, \E\left[1 \wedge \exp\left(\frac{1}{2(1+\sigma_d^2)} \sum_{i=2}^d (Y_i^2 - X_i^2)\right) \mid \mathbf{X}_d\right] \right. \cr
  &\qquad\qquad\left.- \E\left[1 \wedge \exp\left(\frac{1}{2(1+\sigma_d^2)} \sum_{i=2}^d (Y_i^2 - X_i^2)\right) \mid \mathbf{X}_d\right] \right| \cr
  &\qquad + \E\left|\E\left[1 \wedge \exp\left(\frac{1}{2(1+\sigma_d^2)} \sum_{i=2}^d (Y_i^2 - X_i^2)\right) \mid \mathbf{X}_d\right] - \vartheta_{w, \tau}(\ell) \right|.
 \end{align*}
 The first term on the RHS converges to 0 given that $0 \leq 1 \wedge \exp(x) \leq 1$.

 Given $\mathbf{X}_d$, we saw that we can write
 \begin{align*}
 (Y_i - X_i)(Y_i + X_i) &= -\frac{2 + \sigma_d^2}{1 + \sigma_d^2} \sigma_d^2 X_i^2 - \frac{1}{(1 + \sigma_d^2)^{3/2}} \sigma_d^3 X_i U_i \cr
  &\qquad + \frac{(2 + \sigma_d^2)}{(1 + \sigma_d^2)^{3/2}} \sigma U_i X_i + \frac{1}{1 + \sigma_d^2} \sigma_d^2 U_i^2.
\end{align*}
We define $S_d := \sum_{i=2}^d Y_i^2 - X_i^2$ and $W_d := -\frac{2 + \sigma_d^2}{1 + \sigma_d^2} \sigma_d^2 \sum_{i=2}^dX_i^2 + \frac{(2 + \sigma_d^2)}{(1 + \sigma_d^2)^{3/2}} \sigma_d \sum_{i=2}^d U_i X_i +\frac{ \ell^2d}{(1+\sigma_d^2)(d-1)}$. For the rest of the proof, we consider that $\tau = 1/2$. If $\tau > 1/2$, we can use the same strategy as below, but with $W_d := -\frac{2 + \sigma_d^2}{1 + \sigma_d^2} \sigma_d^2 \sum_{i=2}^dX_i^2 + \frac{(2 + \sigma_d^2)}{(1 + \sigma_d^2)^{3/2}} \sigma_d \sum_{i=2}^d U_i X_i$ because $\frac{1}{1 + \sigma_d^2} \sigma_d^2\sum_{i=2}^d  U_i^2 \rightarrow 0$. In this case, $W_d$ has a conditional normal distribution whose mean is $ -\frac{2 + \sigma_d^2}{1 + \sigma_d^2} \sigma_d^2 \sum_{i=2}^dX_i^2$ and variance $\frac{(2 + \sigma_d^2)^2}{(1 + \sigma_d^2)^{3}} \sigma_d^2 \sum_{i=2}^d X_i^2$. Both converge to 0 with probability 1, but the mean converges quicker than the standard deviation, implying that the limit of the explicit expression for $\E\left[1\wedge \exp\left(\frac{1}{2(1+\sigma_d^2)}W_d\right)  \mid \mathbf{X}_d\right]$ is $2\Phi(0) = 1$, which allows to conclude.

Let us return to the case $\tau = 1/2$. Using the triangle inequality,
 \begin{align*}
   &\E\left[\left|\E\left[1\wedge \exp\left(\frac{1}{2(1+\sigma_d^2)}S_d\right)  \mid \mathbf{X}_d\right] - \vartheta_{w, \tau}(\ell) \right|\right] \cr
   &\quad\leq  \E\left[\left|\E\left[1\wedge \exp\left(\frac{1}{2(1+\sigma_d^2)}S_d\right)  \mid \mathbf{X}_d\right] - \E\left[1\wedge \exp\left(\frac{1}{2(1+\sigma_d^2)}W_d\right)  \mid \mathbf{X}_d\right] \right|\right] \cr
    &\qquad + \E\left[\left|\E\left[1\wedge \exp\left(\frac{1}{2(1+\sigma_d^2)}W_d\right)  \mid \mathbf{X}_d\right] - \vartheta_{w, \tau}(\ell) \right|\right].
 \end{align*}

 We now show that each term vanishes. We start with the first one:
 \begin{align*}
  &\E\left[\left|\E\left[1\wedge \exp\left(\frac{1}{2(1+\sigma_d^2)}S_d\right)  \mid \mathbf{X}_d\right] - \E\left[1\wedge \exp\left(\frac{1}{2(1+\sigma_d^2)}W_d\right)  \mid \mathbf{X}_d\right] \right|\right] \cr
  &\quad \leq \E\left[\left|1\wedge \exp\left(\frac{1}{2(1+\sigma_d^2)}S_d\right)  - 1\wedge \exp\left(\frac{1}{2(1+\sigma_d^2)}W_d\right) \right|\right] \cr
  &\quad \leq \frac{1}{2(1+\sigma_d^2)}\E|S_d - W_d|,
 \end{align*}
 using Jensen's inequality and that the function $1 \wedge \exp(x)$ is 1-Lipschitz continuous. It has been proved previously that $\E|S_d - W_d| \rightarrow 0$.

 Given $\mathbf{X}_d$,
\[
  \frac{1}{2(1+\sigma_d^2)}W_d \sim \mathcal{N}\left(\mu_d, s_d^2 \right),
\]
with
\[
 \mu_d:= -\frac{2 + \sigma_d^2}{2(1 + \sigma_d^2)^2} \sigma_d^2 \sum_{i=1}^dX_i^2 + \frac{ \ell^2d}{2(1+\sigma_d^2)^2(d-1)},
\]
and
\[
 s_d^2:= \frac{(2 + \sigma_d^2)^2}{4(1 + \sigma_d^2)^{5}} \sigma_d^2 \sum_{i=1}^d X_i^2.
\]
Therefore,
\begin{align*}
 \E\left[1\wedge \exp\left(\frac{1}{2(1+\sigma_d^2)}W_d\right)  \mid \mathbf{X}_d\right] &= \Phi\left(\frac{\mu_d}{s_d}\right)+\exp\left(\mu_d+\frac{s_d^2}{2}\right)\Phi\left(-s_d-\frac{\mu_d}{s_d}\right) \cr
 &\rightarrow 2\Phi\left(-\frac{\ell}{2}\right),
\end{align*}
 with probability 1. Lebesgue's dominated convergence theorem allows to establish that
 \[
  \E\left[\left|\E\left[1\wedge \exp\left(\frac{1}{2(1+\sigma_d^2)}W_d\right)  \mid \mathbf{X}_d\right] - \vartheta_{w, \tau}(\ell) \right|\right] \rightarrow 0,
 \]
 which concludes the proof.
\end{proof}

\subsection{Proof of \autoref{theorem:MTM_app_ideal}}

\begin{proof}[Proof of \autoref{theorem:MTM_app_ideal}]
 We first prove that if
   \begin{align*}
  \E\left[d^{2\tau}\left|\frac{w(\mathbf{X}_d, \mathbf{Y}_1)}{\frac{1}{N_d}\sum_{i = 1}^{N_d} w(\mathbf{X}_d, \mathbf{Y}_i)} -  \frac{w(\mathbf{X}_d, \mathbf{Y}_1)}{\E[w(\mathbf{X}_d, \mathbf{Y}_1) \mid \mathbf{X}_d]}\right|\right] \leq \frac{d^{2\tau}}{N_d^{1/2}} \, \varrho_1(d) \rightarrow 0,
 \end{align*}
 and
 \begin{align*}
  \E[d^{2\tau}|\alpha(\mathbf{X}_d, \mathbf{Y}_J) - \alpha_{\text{ideal}}(\mathbf{X}_d, \mathbf{Y}_J)|] \leq \frac{d^{2\tau}}{N_d^{1/2}} \, \varrho_2(d) \rightarrow 0,
 \end{align*}
 then $\{Z_{d, \text{MTM}}(t): t \geq 0\}$ converges weakly towards the same Langevin diffusion $\{Z(t): t \geq 0\}$ as in \autoref{theorem:scaling_limit_ideal}, where $\varrho_1(d)$ and $\varrho_2(d)$ are explicitly defined below.

 We saw in the proof of \autoref{theorem:scaling_limit_ideal} that to prove a weak convergence towards a diffusion, it is essentially sufficient to prove that the pseudo-generator of $\{Z_{d, \text{MTM}}(t): t \geq 0\}$ converges towards the generator of the diffusion in the $1$-norm. We thus first derive the pseudo-generator of $\{Z_{d, \text{MTM}}(t): t \geq 0\}$. It is defined as follows:
\[
 \phi_{d, \text{MTM}}(t) := d^{2\tau} \E[h(Z_{d, \text{MTM}}(t + 1 / d^{2\tau})) - h(Z_{d, \text{MTM}}(t)) \mid \mathcal{F}_{\bZ_{d, \text{MTM}}}(t)],
\]
where $h$ is a test function and $\mathcal{F}_{\bZ_{d, \text{MTM}}}(t)$ is the natural filtration associated to $\{\bZ_{d, \text{MTM}}(t): t \geq 0\}$. The Markov property, the fact that $\bZ_{d, \text{MTM}}(0) \sim \pi_d$ and that $\{\mathbf{X}_{d, \text{MTM}}(m): m \in \N\}$ is time-homogeneous imply that for any $t$,
\begin{align*}
 \phi_{d, \text{MTM}}(t) &= d^{2\tau} \E[h(Z_{d, \text{MTM}}(t + 1 / d^{2\tau})) - h(Z_{d, \text{MTM}}(t)) \mid \bZ_{d, \text{MTM}}(t)] \cr
 &\hspace{-1mm}\stackrel{\text{dist.}}{=} d^{2\tau} \E[(h(Y_{J, 1}) - h(X_1)) \, \alpha(\mathbf{X}_d, \mathbf{Y}_J) \mid \mathbf{X}_d],
\end{align*}
where the last equality is in distribution, $\mathbf{X}_d \sim \pi_d$ and $Y_{J, 1}$ is the first coordinate of $\mathbf{Y}_{J}$, a proposal generated by MTM.

The convergence in the $1$-norm of the pseudo-generator of $\{Z_{d, \text{MTM}}(t): t \geq 0\}$ towards the generator of the diffusion thus corresponds to
\begin{align*}
 &\E|d^{2\tau} \E[(h(Y_{J, 1}) - h(X_1)) \, \alpha(\mathbf{X}_d, \mathbf{Y}_J) \mid \mathbf{X}_d] - Gh(X_1)| \cr
  & \leq \E|d^{2\tau} \E[(h(Y_{J, 1}) - h(X_1)) \, \alpha(\mathbf{X}_d, \mathbf{Y}_J) \mid \mathbf{X}_d] - d^{2\tau} \E[(h(Y_{1}) - h(X_1)) \, \alpha_{\text{ideal}}(\mathbf{X}_d, \mathbf{Y}_d) \mid \mathbf{X}_d]| \cr
  &\quad + \E|d^{2\tau} \E[(h(Y_{1}) - h(X_1)) \, \alpha_{\text{ideal}}(\mathbf{X}_d, \mathbf{Y}_d) \mid \mathbf{X}_d] - Gh(X_1)|,
\end{align*}
using the triangle inequality, where $Y_{1}$ is the first coordinate of $\mathbf{Y}_{d} \sim Q_{w, \sigma}(\mathbf{X}_d, \cdot \,)$. We saw in the proof of \autoref{theorem:scaling_limit_ideal} that
\[
  \E|d^{2\tau} \E[(h(Y_{1}) - h(X_1)) \, \alpha_{\text{ideal}}(\mathbf{X}_d, \mathbf{Y}_d) \mid \mathbf{X}_d] - Gh(X_1)| \rightarrow 0.
\]

Using the triangle inequality,
 \begin{align*}
  &\E|d^{2\tau} \E[(h(Y_{J, 1}) - h(X_1)) \, \alpha(\mathbf{X}_d, \mathbf{Y}_J) \mid \mathbf{X}_d] - d^{2\tau} \E[(h(Y_{1}) - h(X_1)) \, \alpha_{\text{ideal}}(\mathbf{X}_d, \mathbf{Y}_d) \mid \mathbf{X}_d]| \cr
  &\quad \leq \E|d^{2\tau} \E[(h(Y_{J, 1}) - h(X_1)) \, \alpha(\mathbf{X}_d, \mathbf{Y}_J) \mid \mathbf{X}_d] \cr
   &\qquad\qquad- d^{2\tau} \E[(h(Y_{J, 1}) - h(X_1)) \, \alpha_{\text{ideal}}(\mathbf{X}_d, \mathbf{Y}_J) \mid \mathbf{X}_d]| \cr
  &\qquad + \E|d^{2\tau}\E[(h(Y_{J, 1}) - h(X_1)) \, \alpha_{\text{ideal}}(\mathbf{X}_d, \mathbf{Y}_J) \mid \mathbf{X}_d] \cr
  &\qquad\qquad-  d^{2\tau} \E[(h(Y_1) - h(X_1)) \, \alpha_{\text{ideal}}(\mathbf{X}_d, \mathbf{Y}_d) \mid \mathbf{X}_d]|.
 \end{align*}
 We now prove that each of the two expectations on the RHS converges to 0 if
    \begin{align*}
  \E\left[d^{2\tau}\left|\frac{w(\mathbf{X}_d, \mathbf{Y}_1)}{\frac{1}{N_d}\sum_{i = 1}^{N_d} w(\mathbf{X}_d, \mathbf{Y}_i)} -  \frac{w(\mathbf{X}_d, \mathbf{Y}_1)}{\E[w(\mathbf{X}_d, \mathbf{Y}_1) \mid \mathbf{X}_d]}\right|\right] \leq \frac{d^{2\tau}}{N_d^{1/2}} \, \varrho_1(d) \rightarrow 0,
 \end{align*}
 and
 \begin{align*}
  \E[d^{2\tau}|\alpha(\mathbf{X}_d, \mathbf{Y}_J) - \alpha_{\text{ideal}}(\mathbf{X}_d, \mathbf{Y}_J)|] \leq \frac{d^{2\tau}}{N_d^{1/2}} \, \varrho_2(d) \rightarrow 0.
 \end{align*}

 We start with the second one. Given any realisation $\mathbf{x}_d$, we have an explicit expression for the conditional expressions and use them; we thus use the notation $\E_{\mathbf{x}_d}$. Using \autoref{prop:dist_MTM},
 \begin{align*}
  &|\E_{\mathbf{x}_d}[(h(Y_{J, 1}) - h(x_1)) \, \alpha_{\text{ideal}}(\mathbf{x}_d, \mathbf{Y}_J)] -   \E_{\mathbf{x}_d}[(h(Y_1) - h(x_1)) \, \alpha_{\text{ideal}}(\mathbf{x}_d, \mathbf{Y}_d)]| \cr
  & \leq \int |h(y_{1,1}) - h(x_1)| \, \alpha_{\text{ideal}}(\mathbf{x}_d, \mathbf{y}_1) \left|\frac{w(\mathbf{x}_d, \mathbf{y}_1)}{\frac{1}{N_d}\sum_{i = 1}^{N_d} w(\mathbf{x}_d, \mathbf{y}_i)} -  \frac{w(\mathbf{x}_d, \mathbf{y}_1)}{\int w(\mathbf{x}_d, \mathbf{y}_1) \, q_{\sigma_d}(\mathbf{x}_d, \mathbf{y}_1) \, d\mathbf{y}_1}\right| \cr
   &\qquad \times \prod_{i=1}^{N_d} q_{\sigma_d}(\mathbf{x}_d, \mathbf{y}_i) \, d\mathbf{y}_{1:N_d} \cr
  &\leq 2M \E_{\mathbf{x}_d}\left|\frac{w(\mathbf{x}_d, \mathbf{Y}_1)}{\frac{1}{N_d}\sum_{i = 1}^{N_d} w(\mathbf{x}_d, \mathbf{Y}_i)} -  \frac{w(\mathbf{x}_d, \mathbf{Y}_1)}{\E_{\mathbf{x}_d}[w(\mathbf{x}_d, \mathbf{Y}_1)]}\right|,
 \end{align*}
 using Jensen's inequality and the triangle inequality, along with the fact that there exists a positive constant $M$ such that $|h|\leq M$, where $y_{1,1}$ is the first coordinate of $\by_1$.

 Regarding the first one,
 \begin{align*}
  &\E|d^{2\tau} \E[(h(Y_{J, 1}) - h(X_1)) \, \alpha(\mathbf{X}_d, \mathbf{Y}_J) \mid \mathbf{X}_d] - d^{2\tau} \E[(h(Y_{J, 1}) - h(X_1)) \, \alpha_{\text{ideal}}(\mathbf{X}_d, \mathbf{Y}_J) \mid \mathbf{X}_d]| \cr
  &\quad \leq 2M \E[d^{2\tau}|\alpha(\mathbf{X}_d, \mathbf{Y}_J) - \alpha_{\text{ideal}}(\mathbf{X}_d, \mathbf{Y}_J)|],
 \end{align*}
 using Jensen's inequality and the triangle inequality, along with the fact that $|h|\leq M$.

 Therefore, if
   \begin{align*}
  \E\left[d^{2\tau}\left|\frac{w(\mathbf{X}_d, \mathbf{Y}_1)}{\frac{1}{N_d}\sum_{i = 1}^{N_d} w(\mathbf{X}_d, \mathbf{Y}_i)} -  \frac{w(\mathbf{X}_d, \mathbf{Y}_1)}{\E[w(\mathbf{X}_d, \mathbf{Y}_1) \mid \mathbf{X}_d]}\right|\right] \leq \frac{d^{2\tau}}{N_d^{1/2}} \, \varrho_1(d) \rightarrow 0,
 \end{align*}
 and
 \begin{align*}
  \E[d^{2\tau}|\alpha(\mathbf{X}_d, \mathbf{Y}_J) - \alpha_{\text{ideal}}(\mathbf{X}_d, \mathbf{Y}_J)|] \leq \frac{d^{2\tau}}{N_d^{1/2}} \, \varrho_2(d) \rightarrow 0,
 \end{align*}
 then $\{Z_{d, \text{MTM}}(t): t \geq 0\}$ converges weakly towards $\{Z(t): t \geq 0\}$.

 We now prove that each of these two expectations converges to 0. We first prove that
    \begin{align}\label{eqn:proof_thm3_1}
  \E\left[d^{2\tau}\left|\frac{w(\mathbf{X}_d, \mathbf{Y}_1)}{\frac{1}{N_d}\sum_{i = 1}^{N_d} w(\mathbf{X}_d, \mathbf{Y}_i)} -  \frac{w(\mathbf{X}_d, \mathbf{Y}_1)}{\E[w(\mathbf{X}_d, \mathbf{Y}_1) \mid \mathbf{X}_d]}\right|\right] \leq \frac{d^{2\tau}}{N_d^{1/2}} \, \varrho_1(d),
 \end{align}
 and next we prove that
  \begin{align}\label{eqn:proof_thm3_2}
  \E[d^{2\tau}|\alpha(\mathbf{X}_d, \mathbf{Y}_J) - \alpha_{\text{ideal}}(\mathbf{X}_d, \mathbf{Y}_J)|] \leq \frac{d^{2\tau}}{N_d^{1/2}} \, \varrho_2(d).
 \end{align}

 We have that
 \begin{align*}
 &\E\left[d^{2\tau}\left|\frac{w(\mathbf{X}_d, \mathbf{Y}_1)}{\sum_{i = 1}^{N_d} w(\mathbf{X}_d, \mathbf{Y}_i)} -  \frac{w(\mathbf{X}_d, \mathbf{Y}_1)}{\E[w(\mathbf{X}_d, \mathbf{Y}_1) \mid \mathbf{X}_d]}\right|\right] \cr
 &\quad = \E\left[\E\left[d^{2\tau}\left|\frac{w(\mathbf{X}_d, \mathbf{Y}_1)}{\sum_{i = 1}^{N_d} w(\mathbf{X}_d, \mathbf{Y}_i)} -  \frac{w(\mathbf{X}_d, \mathbf{Y}_1)}{\E[w(\mathbf{X}_d, \mathbf{Y}_1) \mid \mathbf{X}_d]}\right|\mid \mathbf{X}_d\right] \right].
 \end{align*}
For any realisation $\mathbf{x}_d$, we have an explicit expression for the conditional expectation and therefore write it as follows:
 \begin{align*}
  &\E_{\mathbf{x}_d}\left[d^{2\tau}\left|\frac{w(\mathbf{x}_d, \mathbf{Y}_1)}{\sum_{i = 1}^{N_d} w(\mathbf{x}_d, \mathbf{Y}_i)} -  \frac{w(\mathbf{x}_d, \mathbf{Y}_1)}{\E_{\mathbf{x}_d}[w(\mathbf{x}_d, \mathbf{Y}_1)]}\right|\right] \cr
  &\quad = \E_{\mathbf{x}_d}\left[d^{2\tau}\left|\frac{w(\mathbf{x}_d, \mathbf{x}_d + \sigma_d\mathbf{U}_1)}{\sum_{i = 1}^{N_d} w(\mathbf{x}_d, \mathbf{x}_d + \sigma_d\mathbf{U}_i)} -  \frac{w(\mathbf{x}_d, \mathbf{x}_d + \sigma_d\mathbf{U}_1)}{\E_{\mathbf{x}_d}[w(\mathbf{x}_d, \mathbf{x}_d + \sigma_d\mathbf{U}_1)]}\right|\right],
 \end{align*}
 using that $\mathbf{Y}_i \sim q_{\sigma_d}(\mathbf{x}_d, \cdot \,)$ is equal in distribution to $\mathbf{x}_d + \sigma_d\mathbf{U}_i$ with $\mathbf{U}_i := (U_{i,1}, \ldots, U_{i,d})$, $U_{i,1}, \ldots, U_{i,d}$ being $d$ (conditionally) independent standard normal random variables. We prove the result for the case $w(\mathbf{x}_d, \mathbf{y}_d) = \pi(\mathbf{y}_d) / \pi(\mathbf{x}_d)$; the case $w(\mathbf{x}_d, \mathbf{y}_d) = \sqrt{\pi(\mathbf{y}_d) / \pi(\mathbf{x}_d)}$ is proved similarly.

  Using the definition of the GB weight function and the Cauchy--Schwarz inequality,
 \begin{align*}
  &\E_{\mathbf{x}_d}\left[d^{2\tau}\left|\frac{w(\mathbf{x}_d, \mathbf{x}_d + \sigma_d\mathbf{U}_1)}{\sum_{i = 1}^{N_d} w(\mathbf{x}_d, \mathbf{x}_d + \sigma_d\mathbf{U}_i)} -  \frac{w(\mathbf{x}_d, \mathbf{x}_d + \sigma_d\mathbf{U}_1)}{\E_{\mathbf{x}_d}[w(\mathbf{x}_d, \mathbf{x}_d + \sigma_d\mathbf{U}_1)]}\right|\right]  \cr
 &\quad =  \E_{\mathbf{x}_d}\left[d^{2\tau} \left|\frac{\pi_d(\mathbf{x}_d + \sigma_d\mathbf{U}_1)\left(\frac{1}{N_d} \sum_{i = 1}^{N_d} \pi_d(\mathbf{x}_d + \sigma_d\mathbf{U}_i) - \E_{\mathbf{x}_d}[\pi_d(\mathbf{x}_d + \sigma_d\mathbf{U}_1)]\right)}{\frac{1}{N_d} \sum_{i = 1}^{N_d} \pi_d(\mathbf{x}_d + \sigma_d\mathbf{U}_i) \E_{\mathbf{x}_d}[\pi_d(\mathbf{x}_d + \sigma_d\mathbf{U}_1)]}\right|\right] \cr
 &\quad \leq d^{2\tau} \E_{\mathbf{x}_d}\left[\left(\frac{\pi_d(\mathbf{x}_d + \sigma_d\mathbf{U}_1)}{\frac{1}{N_d} \sum_{i = 1}^{N_d} \pi_d(\mathbf{x}_d + \sigma_d\mathbf{U}_i)}\right)^2\right]^{1/2}\cr
 &\qquad \times \E_{\mathbf{x}_d}\left[\left(\frac{1}{N_d} \sum_{i = 1}^{N_d} \frac{\pi_d(\mathbf{x}_d + \sigma_d\mathbf{U}_i)}{\E_{\mathbf{x}_d}[\pi_d(\mathbf{x}_d + \sigma_d\mathbf{U}_1)]} - 1\right)^2\right]^{1/2}.
 \end{align*}

 We analyse these two terms separately. First,
 \begin{align*}
  &\E_{\mathbf{x}_d}\left[\left(\frac{\pi_d(\mathbf{x}_d + \sigma_d\mathbf{U}_1)}{\frac{1}{N_d} \sum_{i = 1}^{N_d} \pi_d(\mathbf{x}_d + \sigma_d\mathbf{U}_i)}\right)^2\right]^{1/2} \cr
   &\quad = \E_{\mathbf{x}_d}\left[\left(\frac{\exp\left(-\frac{1}{2}\sum_{j=1}^d (x_j + \sigma_d U_{1j})^2\right)}{\frac{1}{N_d} \sum_{i = 1}^{N_d} \exp\left(-\frac{1}{2}\sum_{j=1}^d (x_j + \sigma_d U_{ij})^2\right)}\right)^2\right]^{1/2} \cr
   &\quad \leq \E_{\mathbf{x}_d}\left[\exp\left(-2\sigma_d^2\sum_{j=1}^d (U_{1j} + x_j / \sigma_d)^2\right)\right]^{1/4} \cr
   &\qquad \times \E_{\mathbf{x}_d}\left[\left(\frac{1}{N_d} \sum_{i = 1}^{N_d} \exp\left(-\frac{\sigma_d^2}{2}\sum_{j=1}^d (U_{ij} + x_j / \sigma_d)^2\right)\right)^{-4}\right]^{1/4} \cr
   &\quad\leq \E_{\mathbf{x}_d}\left[\exp\left(-2\sigma_d^2\sum_{j=1}^d (U_{1j} + x_j / \sigma_d)^2\right)\right]^{1/4}\E_{\mathbf{x}_d}\left[\exp\left(2\sigma_d^2\sum_{j=1}^d (U_{1j} + x_j / \sigma_d)^2\right)\right]^{1/4} \cr
   &\quad= \frac{\exp\left(-\frac{\|\bx_d \|^2}{2(1 + 4 \sigma_d^2)}\right)}{(1 + 4 \sigma_d^2)^{d/8}} \frac{\exp\left(\frac{\|\bx_d \|^2}{2(1 - 4 \sigma_d^2)}\right)}{(1 - 4 \sigma_d^2)^{d/8}} = \frac{\exp\left(\frac{4 \sigma_d^2 \|\bx_d \|^2}{1 - 16 \sigma_d^4}\right)}{(1 + 4 \sigma_d^2)^{d/8} (1 - 4 \sigma_d^2)^{d/8}},
 \end{align*}
 using the Cauchy--Schwarz inequality, \autoref{prop:bound_convex}, and the explicit expression of the moment generating function of a non-central chi-squared distribution. Note that
 \[
 \E_{\mathbf{x}_d}\left[\exp\left(2\sigma_d^2\sum_{j=1}^d (U_{1j} + x_j / \sigma_d)^2\right)\right]
 \]
  exists for large enough $d$; it more precisely exists when $4\sigma_d^2 < 1$.

 We now turn to the other term:
    \begin{align*}
     \E_{\mathbf{x}_d}\left[\left(\frac{1}{N_d} \sum_{i = 1}^{N_d} \frac{ \pi_d(\mathbf{x}_d + \sigma_d\mathbf{U}_i)}{\E_{\mathbf{x}_d}[ \pi_d(\mathbf{x}_d + \sigma_d\mathbf{U}_1)]} - 1\right)^2\right]^{1/2} &= \frac{1}{N_d^{1/2}} \var_{\mathbf{x}_d}\left[\frac{ \pi_d(\mathbf{x}_d + \sigma_d\mathbf{U}_1)}{\E_{\mathbf{x}_d}[ \pi_d(\mathbf{x}_d + \sigma_d\mathbf{U}_1)]}\right]^{1/2} \cr
     &\hspace{-10mm}\leq\frac{1}{N_d^{1/2}}\frac{\E_{\mathbf{x}_d}\left[\exp\left(-\sum_{j=1}^d (x_j + \sigma_d U_{1j})^2\right)\right]^{1/2}}{\E_{\mathbf{x}_d}\left[\exp\left(-\frac{1}{2}\sum_{j=1}^d (x_j + \sigma_d U_{1j})^2\right)\right]} \cr
     &\hspace{-10mm}= \frac{1}{N_d^{1/2}}\frac{(1 + \sigma_d^2)^{d/2}}{(1 + 2\sigma_d^2)^{d/4}}\exp\left(\frac{\sigma_d^2 \|\bx_d\|^2 }{2(1 + 2\sigma_d^2)(1 + \sigma_d^2)}\right).
    \end{align*}

    Putting all that together yields
    \begin{align*}
     &\E\left[d^{2\tau}\left|\frac{w(\mathbf{X}_d, \mathbf{Y}_1)}{\sum_{i = 1}^{N_d} w(\mathbf{X}_d, \mathbf{Y}_i)} -  \frac{w(\mathbf{X}_d, \mathbf{Y}_1)}{\E[w(\mathbf{X}_d, \mathbf{Y}_1) \mid \mathbf{X}_d]}\right|\right] \cr
     &\quad\leq \frac{d^{2\tau}}{N_d^{1/2}} \frac{(1 + \sigma_d^2)^{d/2}}{(1 + 2\sigma_d^2)^{d/4}(1 - 16  \sigma_d^4)^{d/8}} \, \E\left[\exp\left(\frac{4 \sigma_d^2 \|\bX_d\|^2}{(1 - 16 \sigma_d^4)}\right) \exp\left(\frac{\sigma_d^2 \|\bX\|^2}{2(1 + 2\sigma_d^2)(1 + \sigma_d^2)}\right)\right] \cr
     &\quad= \frac{d^{2\tau}}{N_d^{1/2}} \frac{(1 + \sigma_d^2)^{d/2}}{(1 + 2\sigma_d^2)^{d/4}(1 - 16  \sigma_d^4)^{d/8}} \, \E\left[\exp\left(\frac{\sigma_d^2 (9 + 24 \sigma_d^2) \|\bX\|^2}{2(1 - 16 \sigma_d^4)(1 + 2\sigma_d^2)(1 + \sigma_d^2)}\right)\right] \cr
       &\quad= \frac{d^{2\tau}}{N_d^{1/2}} \frac{(1 + \sigma_d^2)^{d/2}}{(1 + 2\sigma_d^2)^{d/4}(1 - 16  \sigma_d^4)^{d/8}}\left(1 - \frac{\sigma_d^2 (9 + 24 \sigma_d^2)}{(1 - 16 \sigma_d^4)(1 + 2\sigma_d^2)(1 + \sigma_d^2)}\right)^{-d/2},
    \end{align*}
    using the explicit expression of the moment generating function of a chi-squared distribution. Note that the expectation in the penultimate line exists for large enough $d$.

    When $\tau \geq 1/2$,
    \[
      \varrho_1(d) := \frac{(1 + \sigma_d^2)^{d/2}}{(1 + 2\sigma_d^2)^{d/4}(1 - 16  \sigma_d^4)^{d/8}}\left(1 - \frac{\sigma_d^2 (9 + 24 \sigma_d^2)}{(1 - 16 \sigma_d^4)(1 + 2\sigma_d^2)(1 + \sigma_d^2)}\right)^{-d/2}
    \]
    converges to a constant. So the expectation converges to 0 if
    \[
        N_d = d^{4\tau(1 + \rho)},
    \]
    with any $\rho > 0$.

    When $w(\mathbf{x}_d, \mathbf{y}_d) = \sqrt{\pi_d(\mathbf{y}_d) / \pi_d(\mathbf{x}_d)}$, the terms are different, but the speed is the same. Therefore having $\tau = 1/2$ with $N_d = d^{4\tau(1 + \rho)}$ also makes the expectation vanish, but we can also use $\tau = 1/6$ with $N_d = (1 + \nu)^d$ to make the expectation vanish, $\nu$ being any positive constant.

    There remains to prove the bound in \eqref{eqn:proof_thm3_2}, that is
    \begin{align*}
  \E[d^{2\tau}|\alpha(\mathbf{X}_d, \mathbf{Y}_J) - \alpha_{\text{ideal}}(\mathbf{X}_d, \mathbf{Y}_J)|] \leq \frac{d^{2\tau}}{N_d^{1/2}} \, \varrho_2(d).
 \end{align*}
  We first use that the function $1 \wedge x$ is 1-Lipschitz continuous:
 \[
  \E[d^{2\tau}|\alpha(\mathbf{X}_d, \mathbf{Y}_J) - \alpha_{\text{ideal}}(\mathbf{X}_d, \mathbf{Y}_J)|] \leq \E[d^{2\tau}|r(\mathbf{X}_d, \mathbf{Y}_J) - r_{\text{ideal}}(\mathbf{X}_d, \mathbf{Y}_J)|],
 \]
 where
  \[
  r(\mathbf{X}_d, \mathbf{Y}_J) := \frac{\pi_d(\mathbf{Y}_J) \, q_{\sigma_d}(\mathbf{Y}_J, \mathbf{X}_d) \, w(\mathbf{Y}_J, \mathbf{X}_d) \bigg/ \left(\sum_{i = 1}^{N_d - 1} w(\mathbf{Y}_J, \mathbf{Z}_i) + w(\mathbf{Y}_J, \mathbf{X}_d)\right)}{\pi_d(\mathbf{X}_d) \, q_{\sigma_d}(\mathbf{X}_d, \mathbf{Y}_J) \, w(\mathbf{X}_d, \mathbf{Y}_J) \bigg/ \left(\sum_{i = 1}^{N_d} w(\mathbf{X}_d, \mathbf{Y}_i)\right)},
 \]
 and
 \[
  r_{\text{ideal}}(\mathbf{X}_d, \mathbf{Y}_J) := \frac{\pi_d(\mathbf{Y}_J) \, Q_{w, \sigma_d}(\mathbf{Y}_J, \mathbf{X}_d)}{\pi_d(\mathbf{X}_d) \, Q_{w, \sigma_d}(\mathbf{X}_d, \mathbf{Y}_J)}.
 \]

 Recall that
 \[
  Q_{w, \sigma_d}(\mathbf{x}_d, \mathbf{y}_d) := \frac{w(\mathbf{x}_d, \mathbf{y}_d) \, q_{\sigma_d}(\mathbf{x}_d, \mathbf{y}_d)}{\int w(\mathbf{x}_d, \mathbf{y}_d) \, q_{\sigma_d}(\mathbf{x}_d, \mathbf{y}_d) \, \d\mathbf{y}_d}.
 \]
  Using \autoref{prop:dist_MTM}, we can write $\E[d^{2\tau}|r(\mathbf{X}_d, \mathbf{Y}_J) - r_{\text{ideal}}(\mathbf{X}_d, \mathbf{Y}_J)|]$ as $\E[d^{2\tau}|r(\mathbf{X}_d, \mathbf{Y}_1) - r_{\text{ideal}}(\mathbf{X}_d, \mathbf{Y}_1)|]$, where the latter expectation is computed with respect to the following PDF:
 \[
  \pi_d(\mathbf{x}_d) \frac{w(\mathbf{x}_d, \mathbf{y}_1)}{\frac{1}{N_d}\sum_{i = 1}^{N_d} w(\mathbf{x}_d, \mathbf{y}_i)} \prod_{i=1}^{N_d} q_{\sigma_d}(\mathbf{x}_d, \mathbf{y}_i) \prod_{i=1}^{N_d - 1} q_{\sigma_d}(\mathbf{y}_1, \mathbf{z}_i).
 \]
  This means that the expectation $\E[d^{2\tau}|r(\mathbf{X}_d, \mathbf{Y}_1) - r_{\text{ideal}}(\mathbf{X}_d, \mathbf{Y}_1)|]$ can be written as
 \[
  \tilde{\E}\left[d^{2\tau}\left|\frac{w(\mathbf{Y}_1, \mathbf{X}_d)}{\frac{1}{N_d} \sum_{i = 1}^{N_d - 1} w(\mathbf{Y}_1, \mathbf{Z}_i) + w(\mathbf{Y}_1, \mathbf{X}_d)} - \frac{w(\mathbf{Y}_1, \mathbf{X}_d)}{\E_{\mathbf{Y}_1}[w(\mathbf{Y}_1, \mathbf{X}_d)]}\frac{\E_{\mathbf{X}_d}[w(\mathbf{X}_d, \mathbf{Y}_1)]}{\frac{1}{N_d} \sum_{i = 1}^{N_d} w(\mathbf{X}_d, \mathbf{Y}_i)}\right|\right],
 \]
 with respect to a PDF given by
 \[
  \pi_d(\mathbf{y}_1) \prod_{i=2}^{N_d} q_{\sigma_d}(\mathbf{x}_d, \mathbf{y}_i) \prod_{i=1}^{N_d - 1} q_{\sigma_d}(\mathbf{y}_1, \mathbf{z}_i) \, q_{\sigma_d}(\mathbf{y}_1, \mathbf{x}_d),
 \]
 where $\E_{\mathbf{Y}_1}[w(\mathbf{Y}_1, \mathbf{X}_d)]$ is a fonction of the random variable $\mathbf{Y}_1$ for which any realisation $\mathbf{y}_1$ is mapped to
 \[
  \int w(\mathbf{y}_1, \mathbf{x}_d) \, q_{\sigma_d}(\mathbf{y}_1, \mathbf{x}_d) \, \d\mathbf{x}_d;
 \]
 $\E_{\mathbf{X}_d}[w(\mathbf{X}_d, \mathbf{Y}_1)]$ is defined analogously. We noted a change of PDF in the expectation by using the notation `` $\tilde{\E}$ ''.

 Using the triangle inequality,
 \begin{align*}
 & \tilde{\E}\left[d^{2\tau}\left|\frac{w(\mathbf{Y}_1, \mathbf{X}_d)}{\frac{1}{N_d} \sum_{i = 1}^{N_d - 1} w(\mathbf{Y}_1, \mathbf{Z}_i) + w(\mathbf{Y}_1, \mathbf{X}_d)} - \frac{w(\mathbf{Y}_1, \mathbf{X}_d)}{\E_{\mathbf{Y}_1}[w(\mathbf{Y}_1, \mathbf{X}_d)]}\frac{\E_{\mathbf{X}_d}[w(\mathbf{X}_d, \mathbf{Y}_1)]}{\frac{1}{N_d} \sum_{i = 1}^{N_d} w(\mathbf{X}_d, \mathbf{Y}_i)}\right|\right] \cr
 &\quad\leq \tilde{\E}\left[d^{2\tau}\left|\frac{w(\mathbf{Y}_1, \mathbf{X}_d)}{\frac{1}{N_d} \sum_{i = 1}^{N_d - 1} w(\mathbf{Y}_1, \mathbf{Z}_i) + w(\mathbf{Y}_1, \mathbf{X}_d)} - \frac{w(\mathbf{Y}_1, \mathbf{X}_d)}{\E_{\mathbf{Y}_1}[w(\mathbf{Y}_1, \mathbf{X}_d)]}\right|\right] \cr
 &\qquad + \tilde{\E}\left[d^{2\tau}\left|\frac{w(\mathbf{Y}_1, \mathbf{X}_d)}{\E_{\mathbf{Y}_1}[w(\mathbf{Y}_1, \mathbf{X}_d)]} - \frac{w(\mathbf{Y}_1, \mathbf{X}_d)}{\E_{\mathbf{Y}_1}[w(\mathbf{Y}_1, \mathbf{X}_d)]}\frac{\E_{\mathbf{X}_d}[w(\mathbf{X}_d, \mathbf{Y}_1)]}{\frac{1}{N_d} \sum_{i = 1}^{N_d} w(\mathbf{X}_d, \mathbf{Y}_i)}\right|\right].
 \end{align*}
 The first expectation on the RHS can be seen as an expectation with respect to the PDF
 \[
  \pi_d(\mathbf{y}_1) \prod_{i=1}^{N_d - 1} q_{\sigma_d}(\mathbf{y}_1, \mathbf{z}_i) \, q_{\sigma_d}(\mathbf{y}_1, \mathbf{x}_d);
  \]
  so this expectation is equal to that in \eqref{eqn:proof_thm3_1} and it thus converges to 0 under the same conditions.

 For the other one, we have that
 \begin{align*}
  &\tilde{\E}\left[d^{2\tau}\left|\frac{w(\mathbf{Y}_1, \mathbf{X}_d)}{\E_{\mathbf{Y}_1}[w(\mathbf{Y}_1, \mathbf{X}_d)]} - \frac{w(\mathbf{Y}_1, \mathbf{X}_d)}{\E_{\mathbf{Y}_1}[w(\mathbf{Y}_1, \mathbf{X}_d)]}\frac{\E_{\mathbf{X}_d}[w(\mathbf{X}_d, \mathbf{Y}_1)]}{\frac{1}{N_d} \sum_{i = 1}^{N_d} w(\mathbf{X}_d, \mathbf{Y}_i)}\right|\right] \cr
  &= \tilde{\E}\left[d^{2\tau}\frac{w(\mathbf{Y}_1, \mathbf{X}_d)}{\E_{\mathbf{Y}_1}[w(\mathbf{Y}_1, \mathbf{X}_d)]}\frac{1}{\frac{1}{N_d} \sum_{i = 1}^{N_d} w(\mathbf{X}_d, \mathbf{Y}_i)}\left|\frac{1}{N_d} \sum_{i = 1}^{N_d} w(\mathbf{X}_d, \mathbf{Y}_i) - \E_{\mathbf{X}_d}[w(\mathbf{X}_d, \mathbf{Y}_1)]\right|\right].
 \end{align*}
 This expectation can be seen as an expectation with respect to the PDF
 \[
  \pi_d(\mathbf{y}_1) \prod_{i=2}^{N_d} q_{\sigma_d}(\mathbf{x}_d, \mathbf{y}_i) \, q_{\sigma_d}(\mathbf{y}_1, \mathbf{x}_d) = \pi(\mathbf{y}_1) \prod_{i=1}^{N_d} q_{\sigma_d}(\mathbf{x}_d, \mathbf{y}_i),
  \]
  using that $q_{\sigma_d}$ is symmetric. We prove the result for the case $w(\mathbf{x}_d, \mathbf{y}_d) = \pi_d(\mathbf{y}_d) / \pi_d(\mathbf{x}_d)$; the case $w(\mathbf{x}_d, \mathbf{y}_d) = \sqrt{\pi_d(\mathbf{y}_d) / \pi_d(\mathbf{x}_d)}$ is proved similarly. For any realisation $\mathbf{y}_1$, we have that
  \begin{align*}
   \frac{w(\mathbf{y}_1, \mathbf{X}_d)}{\E_{\mathbf{y}_1}[w(\mathbf{y}_1, \mathbf{X}_d)]} &= \frac{\pi_d(\mathbf{X}_d)}{\frac{1}{(2\pi)^{d/2}}\E_{\mathbf{y}_1}\left[\exp\left(-\frac{\sigma_d^2}{2}\sum_{j=1}^d(U_{1j} + y_{1j}/\sigma_d)^2\right)\right]} \cr
   &= \frac{\pi_d(\mathbf{X}_d)}{(2\pi)^{-d/2} (1 + \sigma_d^2)^{-d/2}\exp\left(-\frac{1}{2(1 + \sigma_d^2)} \sum_{j=1}^d y_{1j}^2\right)},
  \end{align*}
  using that $q_{\sigma_d}(\bx_d, \cdot\,)$ is a normal distribution, where, as previously, $U_{11}, \ldots, U_{1N_d}$ are (conditionally) independent standard normal random variables. We can thus rewrite the expectation
  \[
   \tilde{\E}\left[d^{2\tau}\frac{w(\mathbf{Y}_1, \mathbf{X}_d)}{\E_{\mathbf{Y}_1}[w(\mathbf{Y}_1, \mathbf{X}_d)]}\frac{1}{\frac{1}{N_d} \sum_{i = 1}^{N_d} w(\mathbf{X}_d, \mathbf{Y}_i)}\left|\frac{1}{N_d} \sum_{i = 1}^{N_d} w(\mathbf{X}_d, \mathbf{Y}_i) - \E_{\mathbf{X}_d}[w(\mathbf{X}_d, \mathbf{Y}_1)]\right|\right]
  \]
  as
  \begin{align*}
   \stackrel{\approx}{\E}\left[d^{2\tau} (1 + \sigma_d^2)^{d/2}\frac{\exp\left(-\frac{\sigma_d^2}{2(1 + \sigma_d^2)} \, \|\bY_1\|^2\right)}{\frac{1}{N_d} \sum_{i = 1}^{N_d} w(\mathbf{X}_d, \mathbf{Y}_i)}\left|\frac{1}{N_d} \sum_{i = 1}^{N_d} w(\mathbf{X}_d, \mathbf{Y}_i) - \E_{\mathbf{X}_d}[w(\mathbf{X}_d, \mathbf{Y}_1)]\right|\right],
  \end{align*}
  where the expectation is with respect to the PDF $\pi_d(\mathbf{x}_d) \prod_{i=1}^{N_d} q_{\sigma_d}(\mathbf{x}_d, \mathbf{y}_i)$. Using the definition of the GB weight function,
\begin{align*}
 &\stackrel{\approx}{\E}\left[d^{2\tau} (1 + \sigma_d^2)^{d/2}\frac{\exp\left(-\frac{\sigma_d^2}{2(1 + \sigma_d^2)} \sum_{j=1}^d Y_{1j}^2\right)}{\frac{1}{N_d} \sum_{i = 1}^{N_d} w(\mathbf{X}_d, \mathbf{Y}_i)}\left|\frac{1}{N_d} \sum_{i = 1}^{N_d} w(\mathbf{X}_d, \mathbf{Y}_i) - \E_{\mathbf{X}_d}[w(\mathbf{X}_d, \mathbf{Y}_1)]\right|\right] \cr
 &\quad = \, \stackrel{\approx}{\E}\left[d^{2\tau} (1 + \sigma_d^2)^{d/2}\frac{\exp\left(-\frac{\sigma_d^2}{2(1 + \sigma_d^2)} \sum_{j=1}^d Y_{1j}^2\right)}{\frac{1}{N_d} \sum_{i = 1}^{N_d} \exp\left(-\frac{1}{2}\sum_{j=1}^d Y_{ij}^2\right)} \right. \cr
 &\qquad \times \left.\left|\frac{1}{N_d} \sum_{i = 1}^{N_d}\exp\left(-\frac{1}{2}\sum_{j=1}^d Y_{ij}^2\right)- \E_{\mathbf{X}_d}\left[\exp\left(-\frac{1}{2}\sum_{j=1}^d Y_{1j}^2\right)\right]\right|\right] .
\end{align*}
We omit the $\approx$ above $\stackrel{\approx}{\E}$ for the rest of the proof to simplify the notation and note that for the rest of the proof the random variables are such that $\mathbf{X}_d \sim \pi_d$ and $\mathbf{Y}_1, \ldots, \mathbf{Y}_{N_d}$ are conditionally independent given $\mathbf{X}_d$ with a distribution given by $\mathbf{Y}_i \sim \mathcal{N}(\mathbf{X}_d, \sigma_d^2 \I_d)$.

We have that
\begin{align*}
&\E\left[\frac{\exp\left(-\frac{\sigma_d^2}{2(1 + \sigma_d^2)} \sum_{j=1}^d Y_{1j}^2\right)}{\frac{1}{N_d} \sum_{i = 1}^{N_d} \exp\left(-\frac{1}{2}\sum_{j=1}^d Y_{ij}^2\right)} \right. \cr
 &\qquad \times \left.\left|\frac{1}{N_d} \sum_{i = 1}^{N_d}\exp\left(-\frac{1}{2}\sum_{j=1}^d Y_{ij}^2\right)- \E_{\mathbf{X}_d}\left[\exp\left(-\frac{1}{2}\sum_{j=1}^d Y_{1j}^2\right)\right]\right|\right] \cr
 &=\E\left[\E\left[\frac{\exp\left(-\frac{\sigma_d^2}{2(1 + \sigma_d^2)} \sum_{j=1}^d Y_{1j}^2\right)}{\frac{1}{N_d} \sum_{i = 1}^{N_d} \exp\left(-\frac{1}{2}\sum_{j=1}^d Y_{ij}^2\right)} \right.\right. \cr
 &\qquad \times \left. \left.\left|\frac{1}{N_d} \sum_{i = 1}^{N_d}\exp\left(-\frac{1}{2}\sum_{j=1}^d Y_{ij}^2\right)- \E_{\mathbf{X}_d}\left[\exp\left(-\frac{1}{2}\sum_{j=1}^d Y_{1j}^2\right)\right]\right|\mid \mathbf{X}_d\right] \right].
\end{align*}

Given any realisation $\mathbf{x}_d$, we have an explicit expression for the conditional expectation and therefore write it as follows:
\small
 \begin{align*}
  &\E_{\mathbf{x}_d}\left[\frac{\exp\left(-\frac{\sigma_d^2}{2(1 + \sigma_d^2)} \sum_{j=1}^d (x_j + \sigma_d U_{1j})^2\right)}{\frac{1}{N_d} \sum_{i = 1}^{N_d} \exp\left(-\frac{1}{2}\sum_{j=1}^d (x_j + \sigma_d U_{ij})^2\right)} \right. \cr
 &\quad \times \left. \left|\frac{1}{N_d} \sum_{i = 1}^{N_d}\exp\left(-\frac{1}{2}\sum_{j=1}^d (x_j + \sigma_d U_{ij})^2\right)- \E_{\mathbf{x}_d}\left[\exp\left(-\frac{1}{2}\sum_{j=1}^d (x_j + \sigma_d U_{1j})^2\right)\right]\right| \right] \cr
 &\leq \E_{\mathbf{x}_d}\left[\exp\left(-\frac{\sigma_d^2}{1 + \sigma_d^2} \sum_{j=1}^d (x_j + \sigma_d U_{1j})^2\right)\left(\frac{1}{N_d} \sum_{i = 1}^{N_d} \exp\left(-\frac{1}{2}\sum_{j=1}^d (x_j + \sigma_d U_{ij})^2\right)\right)^{-2}\right]^{1/2} \cr
 &\quad\times \E_{\mathbf{x}_d}\left[\left(\frac{1}{N_d} \sum_{i = 1}^{N_d}\exp\left(-\frac{1}{2}\sum_{j=1}^d (x_j + \sigma_d U_{ij})^2\right)- \E_{\mathbf{x}_d}\left[\exp\left(-\frac{1}{2}\sum_{j=1}^d (x_j + \sigma_d U_{1j})^2\right)\right]\right)^{2}\right]^{1/2}
 \end{align*}
 \normalsize
 using that $\mathbf{Y}_i \sim q_{\sigma_d}(\mathbf{x}_d, \cdot \,)$ is equal in distribution to $\mathbf{x}_d + \sigma_d\mathbf{U}_i$ and the Cauchy–Schwarz inequality.

 We analyse the two terms on the RHS separately. First,
 \begin{align*}
     &\E_{\mathbf{x}_d}\left[\exp\left(-\frac{\sigma_d^2}{1 + \sigma_d^2} \sum_{j=1}^d (x_j + \sigma_d U_{1j})^2\right)\left(\frac{1}{N_d} \sum_{i = 1}^{N_d} \exp\left(-\frac{1}{2}\sum_{j=1}^d (x_j + \sigma_d U_{ij})^2\right)\right)^{-2}\right]^{1/2} \cr
     & \quad \leq \E_{\mathbf{x}_d}\left[\exp\left(-\frac{4\sigma_d^2}{1 + \sigma_d^2} \sum_{j=1}^d (x_j + \sigma_d U_{1j})^2\right)\right]^{1/4} \cr
     &\qquad \times \E_{\mathbf{x}_d}\left[\left(\frac{1}{N_d} \sum_{i = 1}^{N_d} \exp\left(-\frac{1}{2}\sum_{j=1}^d (x_j + \sigma_d U_{ij})^2\right)\right)^{-4}\right]^{1/4} \cr
      & \quad \leq \E_{\mathbf{x}_d}\left[\exp\left(-\frac{4\sigma_d^2}{1 + \sigma_d^2} \sum_{j=1}^d (x_j + \sigma_d U_{1j})^2\right)\right]^{1/4}\E_{\mathbf{x}_d}\left[\exp\left(2\sum_{j=1}^d (x_j + \sigma_d U_{1j})^2\right)\right]^{1/4} \cr
      &\quad = \E_{\mathbf{x}_d}\left[\exp\left(-\frac{4\sigma_d^4}{1 + \sigma_d^2} \sum_{j=1}^d (x_j / \sigma_d +  U_{1j})^2\right)\right]^{1/4}\E_{\mathbf{x}_d}\left[\exp\left(2 \sigma_d^2 \sum_{j=1}^d (x_j / \sigma_d + U_{1j})^2\right)\right]^{1/4} \cr
      &\quad = \frac{\exp\left(-\frac{\sigma_d^2\|\bx_d\|^2}{1 + \sigma_d^2 + 8 \sigma_d^4}\right)}{\left(1 + \frac{8 \sigma_d^4}{1 + \sigma_d^2}\right)^{d/8}} \frac{\exp\left(\frac{\|\bx_d\|^2}{2(1 - 4 \sigma_d^2)}\right)}{\left(1 - 4\sigma_d^2\right)^{d/8}} \cr
      &\quad \leq \frac{1}{\left(1 + \frac{8 \sigma_d^4}{1 + \sigma_d^2}\right)^{d/8}} \frac{\exp\left(\frac{\|\bx_d\|^2}{2(1 - 4 \sigma_d^2)}\right)}{\left(1 - 4\sigma_d^2\right)^{d/8}}
    \end{align*}
     using the Cauchy--Schwarz inequality, \autoref{prop:bound_convex}, and the expression for the moment generating function of a non-central chi-squared distribution. Note that
     \[
      \E_{\mathbf{x}_d}\left[\exp\left(2 \sigma_d^2 \sum_{j=1}^d (x_j / \sigma_d + U_{1j})^2\right)\right]
      \]
      exists for large enough $d$; it more precisely exists when $4 \sigma_d^2 < 1$.

     Second,
     \begin{align*}
      &\E_{\mathbf{x}_d}\left[\left(\frac{1}{N_d} \sum_{i = 1}^{N_d}\exp\left(-\frac{1}{2}\sum_{j=1}^d (x_j + \sigma_d U_{ij})^2\right)- \E_{\mathbf{x}_d}\left[\exp\left(-\frac{1}{2}\sum_{j=1}^d (x_j + \sigma_d U_{1j})^2\right)\right]\right)^{2}\right]^{1/2} \cr
      &\quad=\frac{1}{N_d^{1/2}}\var_{\mathbf{x}_d}\left[\exp\left(-\frac{1}{2}\sum_{j=1}^d (x_j + \sigma_d U_{1j})^2\right)\right]^{1/2} \cr
      &\quad\leq \frac{1}{N_d^{1/2}}\E_{\mathbf{x}_d}\left[\exp\left(- \sigma_d^2\sum_{j=1}^d (x_j/\sigma_d + U_{1j})^2\right)\right]^{1/2} \cr
      &\quad=\frac{1}{N_d^{1/2}}\frac{\exp\left(\frac{-\|\bx_d\|^2}{2(1 + 2\sigma_d^2)}\right)}{(1 + 2\sigma_d^2)^{d/4}}.
     \end{align*}

          Putting all the results above together yields
     \begin{align*}
      &\E\left[d^{2\tau} (1 + \sigma_d^2)^{d/2}\frac{\exp\left(-\frac{\sigma_d^2}{2(1 + \sigma_d^2)} \sum_{j=1}^d Y_{1j}^2\right)}{\frac{1}{N_d} \sum_{i = 1}^{N_d} \exp\left(-\frac{1}{2}\sum_{j=1}^d Y_{ij}^2\right)} \right. \cr
 &\quad \times \left.\left|\frac{1}{N_d} \sum_{i = 1}^{N_d}\exp\left(-\frac{1}{2}\sum_{j=1}^d Y_{ij}^2\right)- \E_{\mathbf{X}_d}\left[\exp\left(-\frac{1}{2}\sum_{j=1}^d Y_{1j}^2\right)\right]\right|\right] \cr
 &\quad \leq \frac{d^{2\tau}}{N_d^{1/2}}\frac{(1 + \sigma_d^2)^{d/2}}{\left(1 + \frac{8 \sigma_d^4}{1 + \sigma_d^2}\right)^{d/8}\left(1 - 4\sigma_d^2\right)^{d/8}(1 + 2\sigma_d^2)^{d/4}}\E\left[\exp\left(\frac{3\sigma_d^2 \|\bX_d\|^2}{(1 - 4\sigma_d^2)(1 + 2\sigma_d^2)}\right)\right] \cr
 &\quad= \frac{d^{2\tau}}{N_d^{1/2}}\frac{(1 + \sigma_d^2)^{d/2}}{\left(1 + \frac{8 \sigma_d^4}{1 + \sigma_d^2}\right)^{d/8}\left(1 - 4\sigma_d^2\right)^{d/8}(1 + 2\sigma_d^2)^{d/4} \left(1 - \frac{6 \sigma_d^2}{(1 - 4\sigma_d^2)(1 + 2\sigma_d^2)}\right)^{d/2}}
     \end{align*}
     using the explicit expression for the moment generating function of a chi-squared distribution. Note that the expectation in the penultimate line exist for large enough $d$. The last term, seen as a function of $d$, behaves as that in the bound \eqref{eqn:proof_thm3_1}, and thus converges to 0 under the same conditions. Note that
         \begin{align*}
      \varrho_2(d) &:= \frac{(1 + \sigma_d^2)^{d/2}}{(1 + 2\sigma_d^2)^{d/4}(1 - 16  \sigma_d^4)^{d/8}}\left(1 - \frac{\sigma_d^2 (9 + 24 \sigma_d^2)}{(1 - 16 \sigma_d^4)(1 + 2\sigma_d^2)(1 + \sigma_d^2)}\right)^{-d/2} \cr
      &\qquad + \frac{(1 + \sigma_d^2)^{d/2}}{\left(1 + \frac{8 \sigma_d^4}{1 + \sigma_d^2}\right)^{d/8}\left(1 - 4\sigma_d^2\right)^{d/8}(1 + 2\sigma_d^2)^{d/4} \left(1 - \frac{6 \sigma_d^2}{(1 - 4\sigma_d^2)(1 + 2\sigma_d^2)}\right)^{d/2}}.
    \end{align*}
\end{proof}

\end{document}